\documentclass[runningheads]{llncs}

\usepackage[T1]{fontenc}
\usepackage{graphicx}
\usepackage[bookmarks,unicode,colorlinks=true,breaklinks=true]{hyperref}
\usepackage{color}

\usepackage[english]{babel}

\usepackage{xspace}
\usepackage[tableposition=top]{caption}
\usepackage{wrapfig}
\usepackage{float}

\usepackage{array}
\usepackage{multirow}

\usepackage{xifthen}

\usepackage{amsmath}
\usepackage{amssymb}
\usepackage{mathtools}
\usepackage{microtype}
\usepackage{stmaryrd} 
\usepackage{scalerel}

\usepackage{booktabs}
\usepackage{caption}
\usepackage{subfigure}
\usepackage{multicol}
\usepackage{xifthen}
\usepackage{etoolbox}
\usepackage{tikz}
\usepackage{marvosym} 
\usetikzlibrary{arrows,shapes,calc,automata,positioning}
\tikzstyle{state} = [draw,fill=white,circle,thick,align=center,inner sep=0pt,minimum size=4.5mm]
\tikzstyle{lstate} = [draw,fill=white,rectangle,rounded corners,thick,align=center,inner sep=2pt]
\tikzstyle{dot} = [fill,circle,inner sep=0mm,minimum size=1.25mm,line width=0mm]

\usepackage{pgfplots}
\pgfplotsset{compat=newest}

\usepackage{todonotes}

\usepackage{thmtools}
\usepackage{thm-restate}

\usepackage{adjustbox}

\usetikzlibrary{shapes,positioning,arrows,calc,automata,matrix,fit}

\usepackage{xfrac}

\makeatletter%
\g@addto@macro\normalsize{%
  \setlength\abovedisplayskip{3pt}%
  \setlength\belowdisplayskip{3pt}%
  \setlength\abovedisplayshortskip{-3pt}%
  \setlength\belowdisplayshortskip{3pt}%
}%
\makeatother

\usepackage[linesnumbered,ruled,vlined]{algorithm2e}
\SetAlCapSkip{2ex plus 0.5ex minus 1ex}
\SetEndCharOfAlgoLine{}
\SetArgSty{textrm}
\SetKwInput{Input}{Input}                
\SetKwInput{Output}{Output} 

\SetCommentSty{commentstyle}
\SetKw{Break}{break}
\SetKwProg{Type}{type}{}{end}
\SetKwProg{Function}{function}{}{end}
\SetKwFor{ForEach}{foreach}{do}{}
\SetKwFor{WhileTrue}{repeat}{}{end}
\SetKw{Continue}{continue}
\SetKwRepeat{DoWhile}{repeat}{while}
\SetKwRepeat{DoUntil}{repeat}{until}

\usepackage[noabbrev,capitalise]{cleveref}

\usepackage{enumitem}
\setlist{nosep, noitemsep}
\setlist[enumerate]{nosep}
\newlist{enumthm}{enumerate}{1} 
\setlist[enumthm]{label=\upshape(\roman*),ref={(\roman*)}}
\crefname{enumthmi}{}{} 

\renewcommand{\paragraph}[1]{\smallskip\noindent\emph{#1}}

\newcommand{\Rationals}{\mathbb{Q}}

\newcommand{\nat}{\mathbb{N}}
\newcommand{\reals}{\mathbb{R}}
\newcommand{\rnonneg}{\reals_{\geq 0}}
\newcommand{\realsAndInf}{\overline{\reals}}
\newcommand{\rnonnegAndInf}{\realsAndInf_{\geq 0}}

\newcommand{\tuple}[1]{\ensuremath{( #1 )}}
\newcommand{\paren}[1]{\Bigl({#1} \Bigr)}%
\DeclarePairedDelimiter\abs{\lvert}{\rvert}%
\DeclarePairedDelimiter\norm{\lVert}{\rVert}%
\newcommand{\normM}[1]{{\left\vert\kern-0.25ex\left\vert\kern-0.25ex\left\vert #1 \right\vert\kern-0.25ex\right\vert\kern-0.25ex\right\vert}} 
\DeclarePairedDelimiter\normQ{\lVert}{\rVert_{\Qbf}}
\DeclarePairedDelimiter{\iverson}{\llbracket}{\rrbracket}

\renewcommand\max{\mathsf{max}}

\renewcommand\min{\mathsf{min}}
\newcommand{\diff}[1]{\mathsf{diff}^{#1}}%
\newcommand{\0}{\mathbf{0}}

\newcommand{\Abf}{\mathbf{A}}
\newcommand{\Pbf}{\mathbf{P}}
\newcommand{\Qbf}{\mathbf{Q}}
\newcommand{\Ibf}{\mathbf{I}}
\newcommand{\Rbf}{\mathbf{R}}

\newcommand{\Lbf}{\mathbf{L}}
\newcommand{\Ubf}{\mathbf{U}}
\newcommand{\vbf}{\mathbf{v}}
\newcommand{\xbf}{\mathbf{x}}
\newcommand{\ybf}{\mathbf{y}}
\newcommand{\lbf}{\mathbf{l}}
\newcommand{\ubf}{\mathbf{u}}
\newcommand{\bbf}{\mathbf{b}}

\newcommand{\PRstateMC}[2]{\textnormal{Pr}^{#2}_{#1}}
\newcommand{\PR}{\PRstateMC{}{}}
\newcommand{\EXPstateMC}[3]{\ensuremath{\mathbb{E}^{#3}_{#2}[#1]}}
\newcommand{\EXP}[1]{\EXPstateMC{#1}{}{}}
\newcommand{\genericRV}{\mathsf{v}}
\newcommand{\nextLTL}{\mathord{\scaleto{\bigcirc}{8pt}}}
\DeclareMathOperator{\until}{\pmb{\mathord{\scalerel*{\mathsf{U}}{X}}}}
\newcommand{\event}{\mathord{\scaleto{\Diamond}{8pt}}}
\newcommand{\eventBound}[1]{\mathord{{\scaleto{\Diamond}{8pt}}^{{\! \scaleto{=}{3pt}\scaleto{#1}{4pt}}}}}
\newcommand{\statesTr}{S_{\textnormal{tr}}}
\newcommand{\statesRe}{S_{\textnormal{re}}}
\newcommand{\scc}[1]{\mathsf{SCC}^{#1}}
\newcommand{\sccTr}[1]{\scc{#1}_{\textnormal{tr}}}

\newcommand{\paths}[1]{\mathsf{Paths}^{#1}}
\newcommand{\chains}{\mathsf{Chains}}
\newcommand{\chainsTr}{\mathsf{Chains}_{\textnormal{tr}}}
\newcommand{\cyl}{\mathsf{Cyl}}
\newcommand{\stst}[1]{\ensuremath{\theta^{#1}}}
\newcommand{\rew}{\mathsf{rew}}
\newcommand{\tr}{\mathsf{tr}}
\newcommand{\vt}{\mathsf{vt}}
\newcommand{\evtsdomain}[1]{ \reals^{\lvert #1 \rvert }}
\newcommand{\vioperator}{\Phi}
\newcommand{\maxvioperator}{{\Phi}_{max}}
\newcommand{\minvioperator}{{\Phi}_{min}}
\newcommand{\gsmaxvioperator}{{\Gamma}_{max}}
\newcommand{\gsminvioperator}{{\Gamma}_{min}}
\newcommand{\smallgsmaxvioperator}{{\gamma}^{max}}
\newcommand{\smallgsminvioperator}{{\gamma}^{min}}

\newcommand{\iotainit}{\iota_{\textnormal{init}}} 
\newcommand{\hatiotainit}{\hat{\iota}_{\textnormal{init}}} 
\newcommand{\tauinit}{\tau_{\textnormal{init}}}  
\newcommand{\dtmcTuple}[1]{\tuple{S^{{#1}}, \Pbf^{{#1}}, \iotainit^{#1}}}
\newcommand{\dtmc}{\mathcal{D}}
\newcommand{\dtmcRestrTuple}[2]{\tuple{S^{{#1}}, \Pbf^{{#1}}, \iotainit^{#1#2}}}
\newcommand{\dtmcRestr}[1]{\dtmc{\vert}_{#1}}
\newcommand{\dtmcBv}{\mathcal{B}{\Rsh}^{\hat{v}}}

\newcommand{\ctmc}{\mathcal{C}}
\newcommand{\rate}{\boldsymbol{r}}  
\newcommand{\ctmcTuple}{\tuple{S, \Pbf, \iotainit, \rate}}
\newcommand{\emb}{\mathsf{emb}}

\newcommand{\wrt}{w.r.t.\xspace}

\newcommand{\tool}[1]{\textsf{#1}\xspace}
\newcommand{\storm}{\tool{Storm}}
\newcommand{\prism}{\tool{prism}}

\newcommand{\jani}{\tool{Jani}}
\newcommand{\mcsta}{\tool{mcsta}}

\newcommand{\method}[1]{\texttt{#1}\xspace}

\makeatletter
\RequirePackage[bookmarks,unicode,colorlinks=true,breaklinks=true]{hyperref}%
   \def\@citecolor{blue}%
   \def\@urlcolor{blue}%
   \def\@linkcolor{blue}%

\def\orcidID#1{\smash{\href{http://orcid.org/#1}{\protect\raisebox{-1.25pt}{\protect\includegraphics{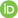}}}}}
\makeatother

\definecolor{myorange}{RGB}{230,159,0} 
\definecolor{myblue}{RGB}{55,126,184}
\definecolor{mygreen}{RGB}{77,175,74}
\definecolor{myred}{RGB}{228,26,28}
\definecolor{mypink}{RGB}{204,121,167}
\definecolor{mygray}{rgb}{0.66, 0.66, 0.66}

\usetikzlibrary{positioning, shapes}
\tikzstyle{every node}=[font=\scriptsize]
\tikzset{dice/.style={draw, minimum size=0.6cm, inner sep=0pt, text width=0.6cm, align=center, fill=myorange!50}}
\tikzset{state/.style={draw, circle, minimum size=0.6cm, inner sep=0pt, text width=0.6cm, align=center}}

\newcommand{\diceFaceOne}{
    \begin{tikzpicture}[scale=0.7]
        \fill (0.5, 0.5) circle (2.2pt);
    \end{tikzpicture}
}

\newcommand{\diceFaceTwo}{
    \begin{tikzpicture}[scale=0.7]
        \fill (0.35, 0.65) circle (2.2pt);
        \fill (0.65, 0.35) circle (2.2pt);
    \end{tikzpicture}
}

\newcommand{\diceFaceThree}{
    \begin{tikzpicture}[scale=0.7]
        \fill (0.35, 0.65) circle (2.2pt);
        \fill (0.5, 0.5) circle (2.2pt);
        \fill (0.65, 0.35) circle (2.2pt);
    \end{tikzpicture}
}

\newcommand{\diceFaceFour}{
    \begin{tikzpicture}[scale=0.7]
        \fill (0.3, 0.3) circle (2.2pt);
        \fill (0.7, 0.3) circle (2.2pt);
        \fill (0.3, 0.7) circle (2.2pt);
        \fill (0.7, 0.7) circle (2.2pt);
    \end{tikzpicture}
}

\newcommand{\diceFaceFive}{
    \begin{tikzpicture}[scale=0.7]
        \fill (0.3, 0.3) circle (2.2pt);
        \fill (0.7, 0.3) circle (2.2pt);
        \fill (0.5, 0.5) circle (2.2pt);
        \fill (0.3, 0.7) circle (2.2pt);
        \fill (0.7, 0.7) circle (2.2pt);
    \end{tikzpicture}
}

\newcommand{\diceFaceSix}{
    \begin{tikzpicture}[scale=0.7]
        \fill (0.3, 0.3) circle (2.2pt);
        \fill (0.7, 0.3) circle (2.2pt);
        \fill (0.3, 0.5) circle (2.2pt);
        \fill (0.7, 0.5) circle (2.2pt);
        \fill (0.3, 0.7) circle (2.2pt);
        \fill (0.7, 0.7) circle (2.2pt);
    \end{tikzpicture}
}

\newcommand{\diceOne}{
    \resizebox{11pt}{9pt}{%
        \begin{tikzpicture}
            \node[dice] (d) {\diceFaceOne};
        \end{tikzpicture}
    }
}
\newcommand{\diceTwo}{
     \kern1pt
     \resizebox{11pt}{9pt}{%
        \begin{tikzpicture}
            \node[dice] (d) {\diceFaceTwo};
        \end{tikzpicture}
    }
}
\newcommand{\diceThree}{
     \kern1pt
     \resizebox{11pt}{9pt}{%
        \begin{tikzpicture}
            \node[dice] (d) {\diceFaceThree};
        \end{tikzpicture}
    }
}
\newcommand{\diceFour}{
     \kern1pt
     \resizebox{11pt}{9pt}{%
        \begin{tikzpicture}
            \node[dice] (d) {\diceFaceFour};
        \end{tikzpicture}
    }
}

\newcommand{\diceSix}{
     \resizebox{11pt}{9pt}{%
        \begin{tikzpicture}
            \node[dice] (d) {\diceFaceSix};
        \end{tikzpicture}
    }
}

\usepackage{etoolbox}
\newrobustcmd*{\mytriangle}[1]{\tikz{\filldraw[draw=#1,fill=#1] (0,0) --
		(0.2cm,0) -- (0.1cm,0.2cm);}}
\newrobustcmd*{\mysquare}[1]{\tikz{\filldraw[draw=#1,fill=#1] (0,0) rectangle (0.2cm,0.2cm);}}
\newrobustcmd*{\mycircle}[1]{\tikz{\filldraw[draw=#1,fill=#1] (0,0) circle [radius=0.1cm];}}
\newrobustcmd*{\mydiamond}[1]{\tikz{\filldraw[draw=#1,fill=#1] (0,0) --
		(0.2cm,0) -- (0.1cm,0.15cm) -- (0,0) -- (0.1cm,-0.15cm) --(0.2cm,0) ;}}

\newlength{\plotwidth}
\newlength{\plotheight}
\newcommand{\legendstyle}{at={(1,1.01)},anchor=north west}
\newcommand{\legendcols}{1}
\tikzset{tool/.code={
\ifthenelse{\equal{#1}{storm.sparse-abs.g-gmres.0.001}}{\tikzset{myorange, line width=1pt}}{}%
\ifthenelse{\equal{#1}{storm.sparse-abs.g-gmres.1e-06}}{\tikzset{myorange, line width=1pt}}{}%
\ifthenelse{\equal{#1}{storm.sparse-rel.g-gmres.0.001}}{\tikzset{myorange, line width=1pt}}{}%
\ifthenelse{\equal{#1}{storm.sparse-rel.g-gmres.1e-06}}{\tikzset{myorange, line width=1pt}}{}%
\ifthenelse{\equal{#1}{storm.sparse-rel.topological-g-gmres.0.001}}{\tikzset{myorange, dashed, line width=1pt}}{}%
\ifthenelse{\equal{#1}{storm.sparse-rel.topological-g-gmres.1e-06}}{\tikzset{myorange, dashed, line width=1pt}}{}%
\ifthenelse{\equal{#1}{storm.sparse-abs.topological-g-gmres.0.001}}{\tikzset{myorange, dashed, line width=1pt}}{}%
\ifthenelse{\equal{#1}{storm.sparse-abs.topological-g-gmres.1e-06}}{\tikzset{myorange, dashed, line width=1pt}}{}%
\ifthenelse{\equal{#1}{storm.sparse-rel.topological-classic-g-gmres.1e-06}}{\tikzset{myorange, line width=1pt}}{}%
\ifthenelse{\equal{#1}{storm.sparse-rel.topological-eqsys-g-gmres.1e-06}}{\tikzset{myorange, dotted, line width=1.3pt}}{}%
\ifthenelse{\equal{#1}{storm.sparse-rel.topological-evt-g-gmres.1e-06}}{\tikzset{myorange, dashed, line width=1pt}}{}%
\ifthenelse{\equal{#1}{storm.sparse-rel.n-power.0.001}}{\tikzset{mypink, line width=1pt}}{}
\ifthenelse{\equal{#1}{storm.sparse-rel.n-power.1e-06}}{\tikzset{mypink, line width=1pt}}{}
\ifthenelse{\equal{#1}{storm.sparse-abs.n-power.0.001}}{\tikzset{mypink, line width=1pt}}{}
\ifthenelse{\equal{#1}{storm.sparse-abs.n-power.1e-06}}{\tikzset{mypink, line width=1pt}}{}
\ifthenelse{\equal{#1}{storm.sparse-abs.topological-n-power.0.001}}{\tikzset{mypink, dashed, line width=1pt}}{}
\ifthenelse{\equal{#1}{storm.sparse-abs.topological-n-power.1e-06}}{\tikzset{mypink, dashed, line width=1pt}}{}
\ifthenelse{\equal{#1}{storm.sparse-rel.topological-n-power.0.001}}{\tikzset{mypink, dashed, line width=1pt}}{}
\ifthenelse{\equal{#1}{storm.sparse-rel.topological-n-power.1e-06}}{\tikzset{mypink, dashed, line width=1pt}}{}
\ifthenelse{\equal{#1}{storm.sound-rel.n-ii.0.001}}{\tikzset{mygreen, line width=1pt}}{}
\ifthenelse{\equal{#1}{storm.sound-rel.n-ii.1-e06}}{\tikzset{mygreen, line width=1pt}}{}
\ifthenelse{\equal{#1}{storm.sound-abs.n-ii.0.001}}{\tikzset{mygreen, line width=1pt}}{}
\ifthenelse{\equal{#1}{storm.sound-abs.n-ii.1e-06}}{\tikzset{mygreen, line width=1pt}}{}
\ifthenelse{\equal{#1}{storm.sound-rel.topological-n-ii.0.001}}{\tikzset{mygreen, dotted, line width=1pt}}{}
\ifthenelse{\equal{#1}{storm.sound-rel.topological-n-ii.1e-06}}{\tikzset{mygreen, dotted, line width=1pt}}{}
\ifthenelse{\equal{#1}{storm.sound-abs.topological-n-ii.0.001}}{\tikzset{mygreen, dotted, line width=1pt}}{}
\ifthenelse{\equal{#1}{storm.sound-abs.topological-n-ii.1e-06}}{\tikzset{mygreen, dotted, line width=1pt}}{}
\ifthenelse{\equal{#1}{storm.sound-rel.topological-evt-n-ii.0.001}}{\tikzset{mygreen, line width=1pt}}{}
\ifthenelse{\equal{#1}{storm.sound-rel.topological-evt-n-ii.0.001}}{\tikzset{mygreen, line width=1.3pt}}{}
\ifthenelse{\equal{#1}{storm.sound-rel.topological-evt-n-ii.0.001}}{\tikzset{mygreen, line width=1pt}}{}
\ifthenelse{\equal{#1}{storm.exact.e-sparselu.ignored}}{\tikzset{myblue,line width=1pt}}{}%
\ifthenelse{\equal{#1}{storm.exact.topological-e-sparselu.ignored}}{\tikzset{myblue, dashed, line width=1pt}}{}%
\ifthenelse{\equal{#1}{storm.exact.topological-classic-sparselu.ignored}}{\tikzset{myblue, line width=1pt}}{}%
\ifthenelse{\equal{#1}{storm.exact.topological-eqsys-e-sparselu.ignored}}{\tikzset{myblue,dotted, line width=1pt}}{}%
\ifthenelse{\equal{#1}{storm.exact.topological-evt-e-sparselu.ignored}}{\tikzset{myblue, dashed, line width=1pt}}{}%
\ifthenelse{\equal{#1}{storm.sparse.e-sparselu.ignored}}{\tikzset{darkgray,line width=1pt}}{}%
\ifthenelse{\equal{#1}{storm.sparse.topological-e-sparselu.ignored}}{\tikzset{darkgray, dashed, line width=1pt}}{}%
\ifthenelse{\equal{#1}{storm.sparse.topological-classic-sparselu.ignored}}{\tikzset{darkgray, line width=1pt}}{}%
\ifthenelse{\equal{#1}{storm.sparse.topological-eqsys-e-sparselu.ignored}}{\tikzset{darkgray, dotted, line width=1.3pt}}{}%
\ifthenelse{\equal{#1}{storm.sparse.topological-evt-e-sparselu.ignored}}{\tikzset{darkgray, dashed, line width=1pt}}{}%
\ifthenelse{\equal{#1}{storm.sparse.topological-elimination.ignored}}{\tikzset{black, dashed, line width=1pt}}{}%
\ifthenelse{\equal{#1}{storm.sparse.elimination.ignored}}{\tikzset{black, line width=1pt}}{}%
\ifthenelse{\equal{#1}{megg.default-abs.ap-naive.0.001}}{\tikzset{myred, line width=1pt}}{}%
\ifthenelse{\equal{#1}{megg.default-abs.ap-sample.0.001}}{\tikzset{myred, dashed, line width=1pt}}{}%
\ifthenelse{\equal{#1}{prism.explicit-rel.default.0.001}}{\tikzset{myred , dotted, line width=1pt}}{}%
\ifthenelse{\equal{#1}{prism.sparse-rel.default.0.001}}{\tikzset{myred, line width=1.3pt}}{}%
\ifthenelse{\equal{#1}{prism.hybrid-rel.default.0.001}}{\tikzset{myred,line width=1pt}}{}%
\ifthenelse{\equal{#1}{prism.explicit-rel.default.1e-06}}{\tikzset{myred, line width=1pt}}{}%
\ifthenelse{\equal{#1}{prism.sparse-rel.default.1e-06}}{\tikzset{myred, line width=1.3pt}}{}%
\ifthenelse{\equal{#1}{prism.hybrid-rel.default.1e-06}}{\tikzset{myred, line width=1pt}}{}%
\ifthenelse{\equal{#1}{prism.default.default.0.001}}{\tikzset{myred,line width=1pt}}{}%
\ifthenelse{\equal{#1}{prism.default.default.1e-06}}{\tikzset{myred, line width=1pt}}{}%
}}

\newcommand{\quantileplot}[4]{%
	\begin{tikzpicture}
	\begin{axis}[
	width=\plotwidth,
	height=\plotheight,
	xmin=10,
	xmax=#4+1,
	ymax=2300, 
	ytick= {1, 10, 100, 1000 },
	yticklabels={$\le$1, $10^1$, $10^2$, $10^3$},
	axis x line=bottom,
	ymin=1,
	axis y line=left,
	ymode=log,
	xlabel=\footnotesize{solved instances (out of #4)},
	ylabel=\footnotesize{time in seconds},
	yticklabel style={font=\footnotesize},
	xticklabel style={font=\footnotesize},
	legend columns = \legendcols,
	legend style={\legendstyle,font=\scriptsize,
		nodes={scale=1, transform shape},inner sep=2pt},
	legend cell align={left}
	]
	\foreach \tool in {#2}{%
		\edef\loopbody{ 
			\noexpand\addplot[tool=\tool] table [x=n,y=\tool shifted, col sep=semicolon] {#1};
		}
		\loopbody
	}
	\draw[densely dotted] (axis cs: 0,10) -- (axis cs: 100,10);
	\draw[densely dotted] (axis cs: 0,100) -- (axis cs: 100,100);
	\draw[densely dotted] (axis cs: 0,1000) -- (axis cs: 100,1000);
	\legend{#3}
	\end{axis}
	\end{tikzpicture}%
}

\newlength{\scatterplotsize}
\newlength{\numberruntimeplotheight}
\tikzset{scatool/.code={
\ifthenelse{\equal{#1}{storm.exact.e-sparselu.ignored}}{\tikzset{darkgray, only marks, mark=o, mark size=1.5pt, line width=1pt}}{}%
\ifthenelse{\equal{#1}{storm.sparse.e-sparselu.ignored}}{\tikzset{myred, only marks, mark=x, mark size=1.5pt,line width=1pt}}{}%
\ifthenelse{\equal{#1}{storm.exact.topological-e-sparselu.ignored}}{\tikzset{darkgray,only marks, mark=o, mark size=1.5pt, line width=1pt}}{}%
\ifthenelse{\equal{#1}{storm.sparse.topological-e-sparselu.ignored}}{\tikzset{myred, only marks, mark=x, mark size=1.5pt, line width=1pt}}{}%
\ifthenelse{\equal{#1}{storm.exact.topological-classic-sparselu.ignored}}{\tikzset{myred, only marks, mark=x, mark size=1.5pt, line width=1pt}}{}
\ifthenelse{\equal{#1}{storm.exact.topological-eqsys-e-sparselu.ignored}}{\tikzset{darkgray, only marks, mark=x, mark size=1.5pt, line width=1pt}}{}
\ifthenelse{\equal{#1}{storm.exact.topological-evt-e-sparselu.ignored}}{\tikzset{darkgray, only marks, mark=o, mark size=1.5pt, line width=1pt}}{}
\ifthenelse{\equal{#1}{storm.sparse.topological-classic-sparselu.ignored}}{\tikzset{pink, only marks, mark=x, mark size=1.5pt, line width=1pt}}{}
\ifthenelse{\equal{#1}{storm.sparse.topological-eqsys-e-sparselu.ignored}}{\tikzset{pink, only marks, mark=x, mark size=1.5pt, line width=1pt}}{}
\ifthenelse{\equal{#1}{storm.sparse.topological-evt-e-sparselu.ignored}}{\tikzset{pink, only marks, mark=x, mark size=1.5pt, line width=1pt}}{}
\ifthenelse{\equal{#1}{storm.sparse-rel.n-power.0.001}}{\tikzset{myred, only marks, mark=x, mark size=1.5pt, line width=1pt}}{}
\ifthenelse{\equal{#1}{storm.sparse-rel.n-power.1e-06}}{\tikzset{myred, only marks, mark=x, mark size=1.5pt, line width=1pt}}{}
\ifthenelse{\equal{#1}{storm.sparse-abs.n-power.0.001}}{\tikzset{myred, only marks, mark=x, mark size=1.5pt, line width=1pt}}{}
\ifthenelse{\equal{#1}{storm.sparse-abs.n-power.1e-06}}{\tikzset{myred, only marks, mark=x, mark size=1.5pt, line width=1pt}}{}
\ifthenelse{\equal{#1}{storm.sparse-rel.topological-n-power.0.001}}{\tikzset{myred, only marks, mark=x, mark size=1.5pt, line width=1pt}}{}
\ifthenelse{\equal{#1}{storm.sparse-rel.topological-n-power.1e-06}}{\tikzset{myred, only marks, mark=x, mark size=1.5pt, line width=1pt}}{}

\ifthenelse{\equal{#1}{storm.sound-rel.n-ii.0.001}}{\tikzset{mygreen, only marks,mark=+, mark size=1.5pt,line width=1pt}}{}
\ifthenelse{\equal{#1}{storm.sound-rel.n-ii.1e-06}}{\tikzset{mygreen, only marks,mark=+, mark size=1.5pt,line width=1pt}}{}
\ifthenelse{\equal{#1}{storm.sound-abs.n-ii.0.001}}{\tikzset{mygreen, only marks,mark=+, mark size=1.5pt, line width=1pt}}{}
\ifthenelse{\equal{#1}{storm.sound-abs.n-ii.1e-06}}{\tikzset{mygreen, only marks,mark=+, mark size=1.5pt, line width=1pt}}{}
\ifthenelse{\equal{#1}{storm.sound-rel.topological-n-ii.0.001}}{\tikzset{mygreen, only marks,mark=+, mark size=1.5pt,line width=1pt}}{}
\ifthenelse{\equal{#1}{storm.sound-rel.topological-n-ii.1e-06}}{\tikzset{mygreen, only marks,mark=+, mark size=1.5pt,line width=1pt}}{}
\ifthenelse{\equal{#1}{storm.sound-abs.topological-n-ii.0.001}}{\tikzset{mygreen, only marks,mark=+, mark size=1.5pt, line width=1pt}}{}
\ifthenelse{\equal{#1}{storm.sound-abs.topological-n-ii.1e-06}}{\tikzset{mygreen, only marks,mark=+, mark size=1.5pt, line width=1pt}}{}
\ifthenelse{\equal{#1}{storm.sound-rel.topological-evt-n-ii.0.001}}{\tikzset{mygreen, only marks,mark=+, mark size=1.5pt, line width=1pt}}{}
\ifthenelse{\equal{#1}{storm.sound-rel.topological-evt-n-ii.1e-06}}{\tikzset{mygreen, only marks,mark=+, mark size=1.5pt, line width=1pt}}{}
\ifthenelse{\equal{#1}{storm.sparse-abs.g-gmres.0.001}}{\tikzset{myorange, only marks, mark=*, mark size=1.5pt,line width=1pt}}{}%
\ifthenelse{\equal{#1}{storm.sparse-abs.g-gmres.1e-06}}{\tikzset{myorange, only marks, mark=*, mark size=1.5pt,line width=1pt}}{}%
\ifthenelse{\equal{#1}{storm.sparse-rel.g-gmres.0.001}}{\tikzset{myorange, only marks, mark=*, mark size=1.5pt,line width=1pt}}{}%
\ifthenelse{\equal{#1}{storm.sparse-rel.g-gmres.1e-06}}{\tikzset{myorange, only marks, mark=*, mark size=1.5pt,line width=1pt}}{}%
\ifthenelse{\equal{#1}{storm.sparse-abs.topological-g-gmres.0.001}}{\tikzset{myorange, only marks, mark=*, mark size=1.5pt,line width=1pt}}{}%
\ifthenelse{\equal{#1}{storm.sparse-abs.topological-g-gmres.1e-06}}{\tikzset{myorange, only marks, mark=*, mark size=1.5pt,line width=1pt}}{}%
\ifthenelse{\equal{#1}{storm.sparse-rel.topological-g-gmres.0.001}}{\tikzset{myorange, only marks, mark=*, mark size=1.5pt,line width=1pt}}{}%
\ifthenelse{\equal{#1}{storm.sparse-rel.topological-g-gmres.1e-06}}{\tikzset{myorange, only marks, mark=*, mark size=1.5pt,line width=1pt}}{}%
\ifthenelse{\equal{#1}{storm.sparse-rel.topological-classic-g-gmres.0.001}}{\tikzset{myorange, only marks, mark=x, mark size=1.5pt, line width=1pt}}{}
\ifthenelse{\equal{#1}{storm.sparse-rel.topological-classic-g-gmres.1e-06}}{\tikzset{myorange, only marks, mark=x, mark size=1.5pt, line width=1pt}}{}
\ifthenelse{\equal{#1}{storm.sparse-rel.topological-eqsys-g-gmres.0.001}}{\tikzset{myred, only marks, mark=*, mark size=1.5pt,line width=1pt}}{}%
\ifthenelse{\equal{#1}{storm.sparse-rel.topological-eqsys-g-gmres.1e-06}}{\tikzset{myred, only marks, mark=*, mark size=1.5pt,line width=1pt}}{}%
\ifthenelse{\equal{#1}{storm.sparse-rel.topological-evt-g-gmres.0.001}}{\tikzset{myblue, only marks, mark=o, mark size=1.5pt, line width=1pt}}{}
\ifthenelse{\equal{#1}{storm.sparse-rel.topological-evt-g-gmres.1e-06}}{\tikzset{myblue, only marks, mark=o, mark size=1.5pt, line width=1pt}}{}
}}
\newcommand{\precisionplot}[6]{%
	\begin{tikzpicture}
		\begin{axis}[
			width=\plotwidth,
			height=\plotheight,
			xmin=0,
			ymin=0.00000005,
			ymax=15,
			axis x line=bottom,
			axis y line=left,
			ytick= {0.000001,0.00001,0.0001, 0.001, 0.01,0.1},
			ymode=log,
			extra y ticks = {0.0000001, 1, 5},
			extra y tick labels = {$\leq10^{-7}$, $\geq10^0$, OOR},
			xmode=log,
			xlabel= #4,
			ylabel= #5,
			y label style={at={(axis description cs:-0.2,0.5)},anchor=north},
			yticklabel style={font=\tiny},
			xticklabel style={font=\tiny},
			legend columns=\legendcols,
			legend style={\legendstyle,font=\scriptsize,
				nodes={scale=1, transform shape},inner sep=2pt},
			legend cell align={left}, 
			table/col sep=semicolon,
			]
			\foreach \scatool in {#2}{%
				\edef\loopbody{ 
					\noexpand\addplot[scatool=\scatool ] table [x={#6},y=\scatool, col sep=semicolon] {#1};
				}
				\loopbody
			}
			\draw[densely dotted] (axis cs: 0,0.0000001) -- (axis cs: 10000000000,0.0000001);
			\draw[densely dotted] (axis cs: 0,0.000001) -- (axis cs: 1000000000000,0.000001);
			\draw[densely dotted] (axis cs: 0,0.001) -- (axis cs: 1000000000000,0.001);
			\draw[densely dotted] (axis cs: 0,1) -- (axis cs: 1000000000000,1);
			\draw[densely dotted] (axis cs: 0,5) -- (axis cs: 1000000000000,5);
			\legend{#3}
		\end{axis}
	\end{tikzpicture}%
}

\newcommand{\numberruntimeplot}[7]{%
	\begin{tikzpicture}
		\begin{axis}[
			width=\scatterplotsize,
			height=\numberruntimeplotheight,
			xmin=0,
			xmax=#7,
			axis x line=bottom,
			ymin=1,
			ymax=30000,
			ytick={0.1,10,1000},
			axis y line=left,
			ymode=log,
			extra y ticks = {1, 6000,20000}, 
			extra y tick labels = {$\leq 1$, OOR, INC}, 
			xmode=log,
			xlabel= #4,
			ylabel= #5,
			xlabel style={font=\scriptsize},
			ylabel style={yshift=-6pt, font=\scriptsize},
			yticklabel style={font=\scriptsize},
			xticklabel style={font=\scriptsize},
			legend columns=\legendcols,
				legend style={\legendstyle, font=\scriptsize,
				nodes={scale=1, transform shape},inner sep=2pt},
			legend cell align={left}, 
			table/col sep=semicolon,
			]
			\foreach \scatool in {#2}{%
				\edef\loopbody{ 
					\noexpand\addplot[scatool=\scatool ] table [x={#6},y=\scatool, col sep=semicolon] {#1};
				}
				\loopbody
			}
			\draw[densely dotted] (axis cs: 0,10) -- (axis cs: 100000000000,10);
			\draw[densely dotted] (axis cs: 0,100) -- (axis cs: 100000000000,100);
			\draw[densely dotted] (axis cs: 0,1000) -- (axis cs: 100000000000,1000);
			\draw[densely dotted] (axis cs: 0,6000) -- (axis cs: 100000000000,6000);
			\draw[densely dotted] (axis cs: 0,20000) -- (axis cs: 100000000000,20000);
			\legend{#3}
		\end{axis}
	\end{tikzpicture}%
}

\newcommand{\scatterplotstorm}[6]{%
	\begin{tikzpicture}
	\begin{axis}[
		width=\scatterplotsize,
		height=\scatterplotsize,
		axis equal image,
		xmin=1,
		ymin=1,
		ymax=30000,
		xmax=30000,
		xmode=log,
		ymode=log,
		axis x line=bottom,
		axis y line=left,
		xtick={10,100,1000},
		ytick={10,100,1000},
		extra x ticks = {1, 6000, 20000}, 
		extra x tick labels = {$\leq 1$, OOR, INC}, 
		extra y ticks = {1, 6000, 20000}, 
		extra y tick labels = {$\leq 1$, OOR, INC}, 
		xlabel={#3},
		xlabel style={yshift=8pt, font=\scriptsize},
		ylabel={#5},
		ylabel style={yshift=-7pt, font=\scriptsize},
		yticklabel style={font=\scriptsize},
		xticklabel style={rotate=290,anchor=west,font=\scriptsize},
		legend columns=\legendcols,
		legend style={font=\scriptsize,\legendstyle, 
				nodes={scale=1, transform shape},
				inner sep=1.5pt},
			legend cell align={left}
	]
	\addplot[
	scatter,
	only marks,
	scatter/classes={
		dtmc={mark=square*,myblue,mark size=1},
		ctmc={mark=diamond*,mypink,mark size=1.25}
	},
	scatter src=explicit symbolic
	]%
	table [col sep=semicolon,x=#2,y=#4,meta=Type] {#1};
	\ifthenelse{\NOT\equal{#6}{false}}{\legend{DTMC, CTMC}}{}
	\addplot[no marks] coordinates {(0.01,0.01) (6000,6000) };
	\addplot[no marks, densely dotted] coordinates {(0.1,1) (600,6000)};
	\addplot[no marks, densely dotted] coordinates {(1,0.1) (6000,600)};
	\draw[densely dotted] (axis cs: 0.1,6000) -- (axis cs: 20000,6000);
	\draw[densely dotted] (axis cs: 0.1,20000) -- (axis cs: 20000,20000);
	\draw[densely dotted] (axis cs: 6000,0.1) -- (axis cs: 6000, 20000);
	\draw[densely dotted] (axis cs: 20000, 0.1) -- (axis cs: 20000,20000);
	\end{axis}
	\end{tikzpicture}
}
\newcommand{\scatterplotstormcyclic}[6]{%
	\begin{tikzpicture}
		
		\begin{axis}[
			width=\scatterplotsize,
			height=\scatterplotsize,
			axis equal image,
			xmin=1,
			ymin=1,
			ymax=30000,
			xmax=30000,
			xmode=log,
			ymode=log,
			axis x line=bottom,
			axis y line=left,
			xtick={10,100,1000},
			ytick={10,100,1000},
			extra x ticks = {1, 6000, 20000}, 
			extra x tick labels = {$\leq 1$, OOR, INC}, 
			extra y ticks = {1, 6000, 20000}, 
			extra y tick labels = {$\leq 1$, OOR, INC}, 
			xlabel={#3},
			xlabel style={yshift=12pt,font=\scriptsize},
			xticklabel style={font=\scriptsize,rotate=290}, 
			ylabel={#5},
			ylabel style={yshift=-8pt, font=\scriptsize},
			yticklabel style={font=\scriptsize},
			legend columns=\legendcols,
				legend style={\legendstyle, font=\scriptsize,
				nodes={scale=1, transform shape},inner sep=2pt},
			]
			\addplot[
			scatter,
			only marks,
			scatter/classes={
				acyclic={mark=triangle*,myred,mark size=1.5},
				cyclic={mark=*,mygreen,mark size=1.5}
			},
			scatter src=explicit symbolic, 
			]%
			table [col sep=semicolon,x=#2,y=#4,meta=topology] {#1};
			\ifthenelse{\NOT\equal{#6}{false}}{\legend{acyclic, cyclic}}{}
			\addplot[no marks] coordinates {(0.01,0.01) (6000,6000) };
			\addplot[no marks, densely dotted] coordinates {(0.1,1) (600,6000)};
			\addplot[no marks, densely dotted] coordinates {(1,0.1) (6000,600)};
			\draw[densely dotted] (axis cs: 0.1,6000) -- (axis cs: 20000,6000);
			\draw[densely dotted] (axis cs: 0.1,20000) -- (axis cs: 20000,20000);
			\draw[densely dotted] (axis cs: 6000,0.1) -- (axis cs: 6000, 20000);
			\draw[densely dotted] (axis cs: 20000, 0.1) -- (axis cs: 20000,20000);
	\end{axis}
	\pgfresetboundingbox
		\useasboundingbox (-1,-1.0) rectangle (4.4,3.5);
	\end{tikzpicture}
}

\begin{document}

\title{Accurately Computing Expected Visiting Times and Stationary Distributions in Markov Chains}
\titlerunning{Computing EVTs and Stationary Distributions in Markov Chains}
%
%
%
\author{Hannah Mertens$^{(\text{\Letter})}$\orcidID{0009-0009-6815-3285} \and
Joost-Pieter Katoen\orcidID{0000-0002-6143-1926} \and\\
Tim Quatmann\orcidID{0000-0002-2843-5511} \and
Tobias Winkler\orcidID{0000-0003-1084-6408}}
\authorrunning{H. Mertens et al.}
%
\institute{RWTH Aachen University, Aachen, Germany\\
\email{\{hannah.mertens,katoen,tim.quatmann,tobias.winkler\}@cs.rwth-aachen.de}}

\maketitle    

\begin{abstract}
We study the accurate and efficient computation of the expected number of times each state is visited in discrete- and continuous-time Markov chains.
To obtain sound accuracy guarantees efficiently, we lift interval iteration and topological approaches known from the computation of reachability probabilities and expected rewards.
We further study applications of expected visiting times, including the sound computation of the stationary distribution and expected rewards conditioned on reaching multiple goal states.
The implementation of our methods in the probabilistic model checker \storm{} scales to large systems with millions of states.
Our experiments on the quantitative verification benchmark set show that the computation of stationary distributions via expected visiting times consistently outperforms existing approaches --- sometimes by several orders of magnitude.
\end{abstract}

\section{Introduction}\label{sec:intro}

\paragraph{Expected visiting times.}
Common questions for the quantitative analysis of Markov chains include reachability probabilities, stationary distributions, and expected rewards~\cite{Kat16}. 
Many authors~\cite{KS76,GT02,Sha+06,Gok+04,Mee+11,Pie+10,Fio+13} have recognized the importance of another quantity called \emph{expected visiting times (EVTs)}, which describe the expected time a system spends in each state. 
EVTs are characterized as the unique solution of a linear equation system~\cite{KS76}.
They are not only relevant in their own right, but also useful to obtain various other quantities, including the ones mentioned above.
This applies particularly to \emph{forward analyses} which aim at computing, e.g., the distribution over terminal states given an initial distribution.

\paragraph{Sound approximation of EVTs.}
In the context of (probabilistic) model checking, the two main requirements for any numeric procedure are \emph{scalability} and \emph{soundness}, i.e., the error in the reported result has to be bounded by a predefined threshold.
Scalability is typically achieved via numerically robust iterative methods~\cite{Saa03,Var99,Wim+08} such as the \emph{Jacobi} or \emph{Gauss-Seidel method}~\cite{Var99}. 
In general, these methods do not converge to the exact solution after a finite number of iterations. 
Thus, the procedure is usually stopped as soon as a termination criterion is satisfied~\cite{Saa03}. 
However, standard stopping criteria such as small difference of consecutive iterations are not sound in the above sense:
They do not actually indicate how close the approximation is to the true solution.
Since the correctness of results in model checking, especially for safety-critical systems, is crucial, 
several authors have proposed sound iterative algorithms~\cite{HM18,Bai+17,QK18,HK20}. 
While these works focus on computing quantities such as reachability probabilities and expected rewards, the sound computation of EVTs has not yet been studied. 

\paragraph{Motivating example: Verifying sampling algorithms.}
To illustrate the use of EVTs in probabilistic verification tasks, consider the Markov chain in \Cref{fig:lumbroso}.
It is a finite-state model of a program --- the \emph{Fast Dice Roller}~\cite{Lum13} --- which takes as input an integer $N \geq 1$ and produces a uniformly distributed output in $\{1,\ldots,N\}$ using unbiased coin flips only.
The Fast Dice Roller thus solves a generalized \emph{Bernoulli Factory} problem~\cite{DBLP:journals/tomacs/KeaneO94}.
Our model in \Cref{fig:lumbroso} is for the case where $N=6$ is fixed.
How can we establish that each of the terminal states $\diceOne, \dots, \diceSix$ is indeed reached with probability sufficiently close to or exactly $\tfrac 1 6$? 
\begin{wrapfigure}[14]{r}{0.40\linewidth} 
    \centering
    \vspace{-0.7cm}
        \begin{tikzpicture}
    \node[] (helper) {};
    \node[state] (s0) [right=0.5 of helper] {$s_0$};
    
    \node[state] (s1) [above right=0.95 and 0.01 of s0] {$s_1$};
    \node[state] (s2) [below right=0.95 and 0.01 of s0] {$s_2$};
    
    \node[state] (s3) [above right=0.15 and 0.7 of s1] {$s_3$};
    \node[state] (s4) [below right=0.55 and 0.7 of s1] {$s_4$};
    \node[state] (s5) [above right=0.35 and 0.7 of s2] {$s_5$};
    \node[state] (s6) [below right=0.15 and 0.7 of s2] {$s_6$};

    \node[dice, right=0.6 of s3] (d1) {\diceFaceOne};
    \node[dice, below right=0.25 and 0.7 of s3] (d2) {\diceFaceTwo};

    \node[dice, right=0.6 of s4] (d3) {\diceFaceThree};
    \node[dice, below right=0.25 and 0.7 of s4] (d4) {\diceFaceFour};

    \node[dice, above right=0.25 and 0.7 of s6] (d5) {\diceFaceFive};
    \node[dice, right=0.6 of s6] (d6) {\diceFaceSix};

\path[->]
      (helper) edge[] node[above=0.01] {$1$} (s0)

      (s0) edge[] node[left=0.01] {$0.5$} (s1)
      (s0) edge[] node[left=0.01] {$0.5$} (s2)

      (s1) edge[] node[above=0.08,pos=0.6] {$0.5$} (s4)
      (s1) edge[] node[above=0.01,pos=0.3] {$0.5$} (s3)   

      (s2) edge[bend right=30] node[above=0.01,pos=0.3] {$0.5$} (s5)
      (s2) edge[] node[below=0.01,pos=0.3]  {$0.5$} (s6)
   
      (s3) edge[] node[above=0.01,pos=0.3] {$0.5$} (d1)
      (s3) edge[] node[below=0.01,pos=0.3] {$0.5$} (d2)
      
      (s4) edge[] node[above=0.01,pos=0.3] {$0.5$} (d3)
      (s4) edge[] node[below=0.01,pos=0.3] {$0.5$} (d4)

      (s5) edge[] node[below=0.3,pos=0.6] {$0.5$} (s1)
      (s5) edge[bend right=30] node[above=0.03,pos=0.5] {$0.5$} (s2)

      (s6) edge[] node[above=0.01,pos=0.3] {$0.5$} (d5)
      (s6) edge[] node[below=0.01,pos=0.3] {$0.5$} (d6)
      ;

\path[->]
    (d1) edge [out=25, in=-25,loop,looseness=5] node[right]{1} (d1)
    (d2) edge [out=25, in=-25,loop,looseness=5] node[right]{1} (d2)
    (d3) edge [out=25, in=-25,loop,looseness=5] node[right]{1} (d3)
    (d4) edge [out=25, in=-25,loop,looseness=5] node[right]{1} (d4)
    (d5) edge [out=25, in=-25,loop,looseness=5] node[right]{1} (d5)
    (d6) edge [out=25, in=-25,loop,looseness=5] node[right]{1} (d6)
    ;

\end{tikzpicture}%




 


%
    \caption{Fast Dice Roller}\label{fig:lumbroso}
\end{wrapfigure}
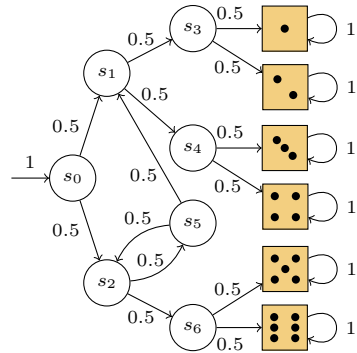
The standard approach for answering this question is to solve $N$ linear equation systems, one for each terminal state~\cite[Ch.~10]{BK08}.
An alternative (and seemingly less well-known) method is to compute the EVTs of each state in the Markov chain, which requires solving just a single linear equation system.
All $N$ desired probabilities can then easily be derived from the EVTs~\cite{KS76}:
For instance, the states $s_3,s_4,s_6$ can all be shown\footnote{See \Cref{sec:computingEVTs} for how to compute these numbers. An EVT of $\tfrac 1 3$ means that \emph{on average}, the state is visited $\tfrac 1 3$ times. In general, EVTs of reachable recurrent states are $\infty$, and the EVT of reachable transient states can take any value in $(0,\infty)$.} to have EVT $\tfrac 1 3$, and thus the reachability probabilities of the terminal states are all $\tfrac 1 6 = \tfrac 1 2 \cdot \tfrac 1 3$.
Similarly, EVTs are useful for computing conditional expected rewards. 
For Bernoulli Factories, this allows us to examine if some outcomes take longer to compute on expectation than others, which is important to analyze possible side channel attacks in a security context.
Furthermore, we show in this paper how computing stationary distributions reduces to EVTs.
Such distributions provide insight into a system's long-term behaviour; applications include the \emph{mean-payoff} of a given policy in a Markov decision process (MDP)~\cite[Pr.~8.1.1]{Put94}, the distribution computed by a Chemical Reaction Network~\cite{DBLP:journals/nc/CardelliKL18}, and the semantics of a probabilistic NetKAT network~\cite{DBLP:conf/pldi/SmolkaKKFHK019}.

\paragraph{Contributions.}
In summary, the contributions of this paper are as follows:
\begin{itemize}
    \item We describe, analyze, and implement the \emph{first sound numerical approximation algorithm for EVTs} in finite discrete- and continuous-time Markov chains.
    Our algorithm is an adaption of the known \emph{Interval Iteration (II)}~\cite{McM+05,HM14,Bai+17}.
    \item We show that computing (sound bounds on) a Markov chain's \emph{stationary distribution} reduces to EVT computations.
    The resulting algorithm significantly outperforms preexisting techniques~\cite{PRISM,Meg23} for stationary distributions.
    \item Similarly, we show how the \emph{conditional expected rewards} until reaching each of the, say, $M$ absorbing states of a Markov chain can be obtained by computing the EVTs and solving a second linear equation system --- this is in contrast to the standard approach which requires solving $M$ equation systems~\cite[Ch.~10]{BK08}.
    \item We implement our algorithm in the probabilistic model checker \storm{}~\cite{STORM} and demonstrate its scalability on various benchmarks.
\end{itemize}

\paragraph{Outline.} 
We define general notation and EVTs in \Cref{sec:prelim,sec:EVTs}, respectively.
In \Cref{sec:computingEVTs}, we present our sound iterative algorithms for computing EVTs approximately.
\Cref{sec:stationary,sec:conditional_rew} present the reductions of stationary distributions and conditional expected rewards to EVTs.
We report on the experimental evaluation of our algorithms in \Cref{sec:experiments}
and summarize related work in \Cref{sec:related}.

\section{Background}\label{sec:prelim}

Let $\mathbb N$ denote the set of non-negative integers and $\realsAndInf = \reals \cup \{\infty, -\infty\}$ the set of extended real numbers.
We equip finite sets $S \neq \emptyset$ with an arbitrary indexing $S = \{s_1,\ldots,s_n\}$ and identify functions of type $\vbf \colon S \to \realsAndInf$ and $\Abf \colon {S} \times S' \to \realsAndInf$ with (column) vectors $\vbf \in \realsAndInf^{\abs{S}}$ and matrices $\Abf \in \realsAndInf^{\abs{S} \times \abs{S'}}$, respectively.
$\Ibf$ denotes the identity matrix.
Vectors are compared component-wise, i.e., $\vbf \leq \vbf'$ iff for all $s \in S$, $\vbf(s) \leq \vbf'(s)$.
\emph{Iverson brackets} $\iverson{B}$ cast the truth value of a Boolean expression $B$ to a numerical value 1 or 0, such that $\iverson{B} = 1$ iff $B$ is true.
\begin{definition}\label{def:dtmc}
A \emph{discrete-time Markov chain (DTMC)} is defined as a triple $\dtmc = \dtmcTuple{\dtmc}$, where $S^{\dtmc}$ is a finite set of states, $\Pbf^{\dtmc}\colon S^{\dtmc} \times S^{\dtmc} \to [0,1]$ is the transition probability function satisfying $\sum_{t \in S} \Pbf^{\dtmc}(s, t)=1$ for all $s \in S$, and $\iotainit^{\dtmc} \colon S \to [0,1]$ is the initial distribution with $\sum_{s \in S} \iotainit^{\dtmc}(s)=1$. 
\end{definition}
\noindent We often omit the superscript from objects associated with a DTMC $\dtmc$ whenever this is clear from context, e.g., we write $\Pbf$ rather than $\Pbf^{\dtmc}$.
An \emph{infinite path} $\pi = s_{0}  s_{1} \cdots \in S^\omega$ in a DTMC $\dtmc = \dtmcTuple{}$ is a sequence of states  such that $\Pbf\tuple{s_{i}, s_{i+1}} >0 $ for all $i \in \nat$. 
We use $\pi[i]=s_i$ to refer to the $i$-th state. $\paths{\dtmc}$ denotes the set of all infinite paths in $\dtmc$. 
The probability measure $\PRstateMC{}{\dtmc}$ over measurable subsets of $\paths{\dtmc}$ is obtained by a standard construction: 
For finite path $\widehat{\pi}$ we set 
$\PRstateMC{}{\dtmc}(\cyl(\widehat{\pi}))= \iotainit(\widehat{\pi}[0])\cdot \prod_{ k= 0}^{\abs{\widehat{\pi}}-1} \Pbf(\widehat{\pi}[k], \widehat{\pi}[k+1])$, 
where the cylinder set $\cyl(\widehat{\pi}) = \{\pi \in \paths{\dtmc} \mid \forall i \in \{0, \dots ,\abs{\widehat{\pi}}\} \colon \pi[i] = \widehat{\pi}[i] \}$ contains all possible infinite continuations of $\widehat{\pi}$.
We write $\PRstateMC{s}{\dtmc}$ for the probability measure induced by $\dtmc$ with the initial distribution assigning probability $1$ to $s \in S$.
We use  LTL-style notation for measurable sets of infinite paths.
For $R, T \subseteq S$ and $k \in \nat$, let $R \until^{=k} T= \{ \pi \in \paths{\dtmc} \mid \pi[k] \in T \wedge \forall i < k \colon  \pi[i]\in R \}$ be the set of infinite paths that visit a state $s \in T$ in the $k$-th step while only visiting states in $R$ before.
We also define $R \until T = \bigcup_{k\geq 0} R \until^{=k} T$,  $\event T = S \until T$ and $\eventBound{k} T = S \until^{=k} T$. 

\paragraph{Expected Rewards.}
A (non-negative) \emph{random variable} over the probability space induced by $\dtmc$ is a measurable function $\genericRV \colon \paths{\dtmc} \to \rnonnegAndInf$.
Its \emph{expected value} is given by the Lebesgue integral $\EXPstateMC{\genericRV}{}{\dtmc} = \int_{\paths{\dtmc}} \genericRV \,d \PRstateMC{}{\dtmc}$.
We write $\mathbb{E}^{\dtmc}_s$ for the expectation obtained under $\PRstateMC{s}{\dtmc}$.
The \emph{total reward} \wrt{} a reward structure $\rew  \colon S \to  \rnonneg$ is defined by the random variable $\tr_{\rew} \colon \paths{\dtmc} \to \rnonnegAndInf$ with $\tr_{\rew}(\pi) = \sum_{k=0}^{\infty} \rew(\pi[k])$;
the \emph{expected total reward} is $\EXPstateMC{\tr_{\rew}}{}{\dtmc}$.

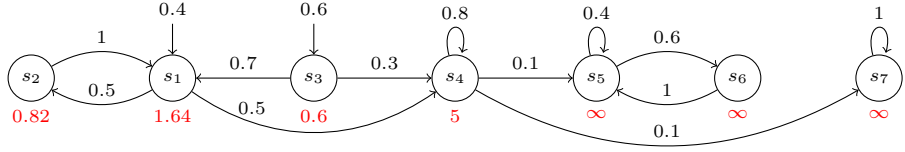
\begin{figure}[t]
    \centering
    \begin{tikzpicture}[node distance=10mm and 12.5mm]
    \node[state,label={below,red:$0.82$}] (s2) {$s_2$};
    \node[state,label={below,red:$1.64$},right=of s2,initial,initial text=$0.4$, initial where=above] (s1) {$s_1$};
    \node[state,label={below,red:$0.6$},right=of s1,initial,initial text=$0.6$, initial where=above] (s3) {$s_3$};
    \node[state,label={below,red:$5$},right=of s3] (s4) {$s_4$};
    \node[state,label={below,red:$\infty$},right=of s4] (s5) {$s_5$};
    \node[state,label={below,red:$\infty$},right=of s5] (s6) {$s_6$};
    \node[state,label={below,red:$\infty$},right=of s6] (s7) {$s_7$};
    
    \path[->] (s1) edge[bend left] node[above] {$0.5$} (s2);
    \path[->] (s1) edge[bend right=35] node[near start,above] {$0.5$} (s4);
    \path[->] (s2) edge[bend left] node[auto] {$1$} (s1);
    \path[->] (s3) edge node[auto] {$0.3$} (s4);
    \path[->] (s3) edge node[above] {$0.7$} (s1);
    \path[->] (s4) edge node[auto] {$0.1$} (s5);
    \path[->] (s4) edge[bend right] node[auto] {$0.1$} (s7);
    \path[->] (s4) edge[loop above] node[auto] {$0.8$} (s4);
    \path[->] (s5) edge[loop above] node[auto] {$0.4$} (s5);
    \path[->] (s5) edge[bend left] node[auto] {$0.6$} (s6);
    \path[->] (s6) edge[bend left] node[above] {$1$} (s5);
    \path[->] (s7) edge[loop above] node[above] {$1$} (s7);
    \end{tikzpicture}
    \caption{
        Running example DTMC. The individual EVTs are below the states.
    }
    \label{fig:runningExample}
\end{figure}

\paragraph{Connectivity in DTMCs.}
    A \emph{strongly connected component (SCC)} of a DTMC $\dtmc$ is a set of states $C \subseteq S$ 
    such that any $s, t \in C$ are mutually reachable, i.e., $\PRstateMC{s}{}(\event \{t\}) > 0$ and $\PRstateMC{t}{}(\event \{s\}) > 0$, and there is no proper subset of $C$ that satisfies this property. 
    An SCC $C$ is called \emph{bottom SCC (BSCC)} if no state outside $C$ is reachable from $C$.
In the following, $\scc{\dtmc}$ denotes the set of SCCs of $\dtmc$. 
We call a DTMC \emph{absorbing} if all its BSCCs are singleton sets, and \emph{irreducible} if $\scc{\dtmc} = \{S\}$.
The SCCs are ordered by a strict partial order $\hookrightarrow$ based on the topology of the DTMC, where $C' \hookrightarrow C$ if and only if $\PRstateMC{s'}{}(\event C) >0$ for some $s' \in C'$ and $C \not = C'$. 
An \emph{SCC chain} of $\dtmc$ is a sequence $\kappa = C_0 \hookrightarrow C_1  \hookrightarrow \dots \hookrightarrow C_k$ of SCCs $C_0, C_1, \ldots, C_k \in \scc{\dtmc}$, where $k \geq 0$. 
The set of all SCC chains in $\dtmc$ is denoted by $\chains^\dtmc$, and $\chainsTr^\dtmc$ denotes the set of SCC chains that do not contain a BSCC. 
The \emph{length} of SCC chain $\kappa = C_0 \hookrightarrow C_1  \hookrightarrow \dots \hookrightarrow C_k$ is $|\kappa| = k$. 

A state $s$ is called \emph{transient} if the probability that the DTMC, starting from $s$, will ever return to $s$ is strictly less than one, otherwise, $s$ is a \emph{recurrent} state.  
Thus, in a finite MC, recurrent states are precisely the states contained in the BSCCs whereas the transient states coincide with non-BSCC states.
The sets of recurrent and transient states in a DTMC are denoted by $\statesRe$ and $\statesTr$, respectively.

\begin{example}
    In the DTMC $\dtmc$ depicted in \Cref{fig:runningExample}, states $s_1,s_2,s_3,s_4$ are transient, and $s_5, s_6, s_7$ are recurrent.
    Also, $\scc{\dtmc} = \{\{s_1,s_2\}, \{s_3\}, \{s_4\}, \{s_5, s_6\}, \{s_7\}\}$, where only $\{s_5, s_6\}$ and $\{s_7\}$ are BSCCs.
    An example SCC chain is $\{s_3\} \hookrightarrow \{s_1,s_2\} \hookrightarrow \{s_4\} \hookrightarrow \{s_7\}$; its length is $3$.
\end{example}

\paragraph{Stationary distributions.}
The stationary distribution (also referred to as steady-state or long-run distribution) is a probability distribution that  specifies the fraction of time spent in each state in the long run (see, e.g.,~\cite[Def.~10.79]{BK08}).

\begin{definition}\label{definition:steadystate}
    The \emph{stationary} distribution of DTMC $\dtmc$ is given by $\stst{\dtmc} \in [0,1]^{\vert S \vert}$ with $\stst{\dtmc}(s) = \lim_{n \to \infty} \frac{1}{n} \sum_{k=1}^{n} \PRstateMC{}{}(\eventBound{k} \{s\})$. 
\end{definition}

\noindent If $\dtmc$ is irreducible, the stationary distribution is given by the unique eigenvector $\stst{\dtmc}$ satisfying $\stst{\dtmc} = \stst{\dtmc} \cdot \Pbf$ and $\sum_{s \in S} \stst{\dtmc}(s) = 1$, see, e.g.,~\cite[Thm.~4.18]{Kul20}.
If $\dtmc$ is reducible, $\stst{\dtmc}$ can be obtained by combining a reachability analysis and the eigenvector computation for each BSCC individually, see \Cref{sec:stationary}.
\begin{definition}
 A \emph{continuous-time Markov chain (CTMC)} is a quadruple $\ctmc = \ctmcTuple$, where $\dtmcTuple{}$ is a DTMC and $\rate \colon S \to \reals_{>0}$ defines exit rates. 
\end{definition} 

\noindent CTMCs extend DTMCs by assigning the rate of an exponentially distributed residence time to each state $s \in S$.
We denote the \emph{embedded DTMC} of a CTMC $\ctmc$ by $\emb({\ctmc}) = \dtmcTuple{}$. 
The semantics are defined in the usual way (see, e.g.,~\cite{Bai+003,Kwi+07}).  
An \emph{infinite timed path} in a CTMC $\ctmc = \ctmcTuple$ is a sequence $\pi = s_0 \overset{\tau_0}{\longrightarrow} s_{1} \overset{\tau_1}{\longrightarrow}\cdots$ consisting of states $s_{0}, s_{1}, \ldots \in S$ and time instances $\tau_0, \tau_1 \ldots  \in \reals_{\geq 0}$, such that $\Pbf\bigl(s_{i}, s_{i+1}\bigr) >0 $ for all $i \in \nat$. 
We denote the time $\tau_i$ spent in state $s_{i}$ by $\mathit{time}_i(\pi)$. 
Notations $\pi[i]$ and $\paths{\ctmc}$ are as for DTMCs.
The probability measure of $\ctmc$ over infinite timed paths~\cite{Bai+003} is denoted by $\PRstateMC{}{\ctmc}$. 
In a CTMC, the \emph{total reward} \wrt{} a reward structure $\rew  \colon S \to  \reals_{\geq 0}$ is the random variable $\tr_{\rew} \colon \paths{\ctmc} \to \rnonnegAndInf$ with $\tr_{\rew}(\pi) = \sum_{i=0}^{\infty} \rew(\pi[i]) \cdot \mathit{time}_i(\pi)$.

\section{Expected Visiting Times}\label{sec:EVTs}

\label{ch:EVTs_dtmc}

We provide characterizations of expected visiting times for a fixed DTMC $\dtmc = \dtmcTuple{}$.
Omitted proofs are in \Cref{ap:EVTs}. 

\begin{definition}\label{definition:EVTs}
    The \emph{expected visiting time (EVT)} of a state $s \in S$ is the expected value $\EXPstateMC{\vt_{s}}{}{\dtmc}$ of the random variable $\vt_s$ with $\vt_s(\pi) = \sum_{k=0}^{\infty} \iverson{\pi[k] = s}$.
\end{definition}

\begin{example}
    The EVTs of the DTMC from \Cref{fig:runningExample} are depicted below its states.
\end{example}
\noindent
Intuitively, random variable $\vt_s$ \emph{counts} the number of times state $s$ occurs on an infinite path.
Consequently, the EVTs of unreachable states and reachable recurrent states in a DTMC are always $0$ and $\infty$, respectively.
For this reason we focus on the EVTs of the transient states $\statesTr$. The following lemma provides an alternative characterization of EVTs in terms of expected total rewards.

\begin{restatable}{lemma}{ETRsToEVTs}\label{theorem:ETRs_to_EVTs}
    For a fixed $s \in S$ and $x \in \reals_{>0}$ the reward structure $\rew \colon S \to \rnonneg$ given by $\rew(t) = x \cdot \iverson{t = s}$ satisfies $\EXP{\vt_s} = \frac{1}{x} \cdot \EXP{\tr_{\rew}}$. 
\end{restatable}

\noindent
By \Cref{theorem:ETRs_to_EVTs}, EVTs can be obtained using existing algorithms for expected total rewards. 
This  approach is, however, inefficient for computing the EVTs of multiple states since it requires solving an equation system for each single state. 

Next, we elaborate on EVTs for multiple states as a solution of a single linear equation system. 
In~\cite[Def.~3.2.2]{KS76}, EVTs are defined using the so-called \emph{fundamental matrix} for absorbing DTMCs.
The fundamental matrix contains as its coefficients for each possible start and target state $s$ and $t$ the EVT $\EXPstateMC{\vt_t}{s}{}$.
Computing the fundamental matrix explicitly becomes infeasible for large models as it requires determining the inverse of a $|\statesTr| \times |\statesTr|$ matrix.
To obtain the vector $(\EXPstateMC{\vt_s}{}{\dtmc})_{s \in \statesTr}$ of EVTs that take the initial distribution of $\dtmc$ into account, it suffices to solve an equation system which is linear in the size of the DTMC. 
The same equation system arises by applying the dual linear program for expected rewards in MDPs~\cite[Ch.~7.2.7]{Put94} to the special case of DTMCs.

\begin{theorem}[{\cite[Cor.~3.3.6]{KS76}}]\label{theorem:expected_visiting_times_LEQ}
    $(\EXP{\vt_{s}})_{s \in \statesTr}$ is the \emph{unique} solution $(\xbf(s))_{s \in \statesTr}$ of the following equation system:
    $\forall s \in \statesTr \colon \xbf(s) = \iotainit(s)+  \sum_{t \in \statesTr} \Pbf(t, s) \cdot \xbf(t) $.  
\end{theorem}

\noindent
Intuitively, this equation system shows that a state $s$ can be visited initially and that it can receive visits from its predecessor states, i.e., the EVT is computed by considering the \emph{incoming} transitions to a state. 
As a consequence, we obtain that the EVTs of the transient states are always finite. 
In particular, if $s \in \statesTr$ is reachable, then $\EXP{\vt_{s}} \in  \reals_{>0}$, and otherwise $\EXP{\vt_{s}} =  0$.

\begin{example}
    Reconsider the DTMC from \Cref{fig:runningExample} with transient states $\statesTr = \{s_1,s_2,s_3,s_4\}$.
    The EVTs of $\statesTr$ are the unique solution $(\xbf(s))_{s \in \statesTr}$ of
    \begin{align*}
        &\xbf(s_1) = 0.4 + 1.0\cdot\xbf(s_2) + 0.7\cdot\xbf(s_3) &&\xbf(s_2) = 0.5\cdot\xbf(s_1)\\
        &\xbf(s_4) = 0.5 \cdot \xbf(s_1) + 0.3 \cdot \xbf(s_3) + 0.8 \cdot \xbf(s_4)  &&\xbf(s_3) = 0.6 
    \end{align*}
\end{example}

\paragraph{Expected visiting times in CTMCs.} Following~\cite{KS61}, we define the EVT of a state $s\in S^\ctmc$ of a CTMC $\ctmc$ as the expected value  $\EXPstateMC{\vt_s}{}{\ctmc}$ of the random variable $\vt_s \colon \paths{\ctmc} \to \rnonnegAndInf$ with  $\vt_s(\pi) = \sum_{k=0}^{\infty} \iverson{\pi[k]= s} \cdot \mathit{time}_k(\pi)$. 
Intuitively, $\vt_s$ considers the total time the system spends in state $s$.
Computing EVTs in CTMCs reduces to the discrete-time case: The EVT of state $s$ coincides with the EVT in the embedded DTMC weighted by the expected residence time $\frac{1}{\rate(s)}$ in $s$:
\begin{restatable}{theorem}{EVTsCTMC}\label{theorem:EVTs_CTMC}
    For all states $s \in S$, it holds that $\EXPstateMC{\vt_s}{}{\ctmc} = \frac{1}{\rate(s)} \cdot \EXPstateMC{\vt_s}{}{\emb(\ctmc)}$.
\end{restatable}

\noindent
\Cref{theorem:EVTs_CTMC} implies that all results and algorithms to compute ETVs in DTMCs are readily applicable to CTMCs, too.
We thus focus on DTMCs in the remainder.

\section{Accurately Computing EVTs}\label{sec:computingEVTs}

In this section, we discuss algorithms to compute EVTs approximately:
An unsound value iteration algorithm (\Cref{sec:VI}), its sound interval iteration extension (\Cref{sec:II}), and finally a topological, i.e., SCC-wise algorithm (\Cref{sec:topological_algorithms}).
Since the EVTs for recurrent states are always either $0$ or $\infty$, we focus on the EVTs of the transient states.
Omitted proofs are in \Cref{ap:computing}.

\subsection{Value Iteration}\label{sec:VI}

Value Iteration (VI) was originally introduced to approximate expected rewards in MDPs~\cite{Bel57}.
In a broader sense, VI simply refers to iterating a function $f \colon \reals^{|S|} \to \reals^{|S|}$ (called \emph{Bellman operator} in the MDP setting) from some given initial vector $\xbf^{(0)}$, i.e., to compute the sequence $\xbf^{(1)} = f(\xbf^{(0)}), \xbf^{(2)} = f(\xbf^{(1)})$, etc.
Instances of VI are usually set up such that the sequence converges to a (generally non-unique) fixed point $\xbf = f(\xbf)$.
In this paper, we only consider VI for the case where $f$ is a linear function $f(\xbf) = \Abf \xbf + \bbf$, where $\Abf$ and $\bbf$ are a matrix and a vector, respectively.
A fixed point $\xbf$ of $f$ is then a solution of the linear equation system $(\Ibf - \Abf) \xbf = \bbf$.
Other iterative methods for solving linear equation systems such as the Jacobi or Gauss-Seidel method can be considered optimized variants of VI, and are applicable in our setting as well, see~\Cref{ap:GS} (Appendix C). 

\paragraph{Value iteration for EVTs.}
For EVTs, the function iterated during VI is as follows:

\begin{definition}\label{definition:EVTs_operator}
    The \emph{EVTs-operator} $\vioperator \colon \evtsdomain{\statesTr} \to \evtsdomain{\statesTr}$ for DTMC $\dtmc$ is defined as
    \begin{align*}
        \vioperator(\xbf) = \paren{\iotainit(s) + \sum_{t \in \statesTr} \xbf(t) \cdot \Pbf(t, s)}_{s \in \statesTr}
        ~.
    \end{align*}
\end{definition}

The above definition is motivated by \Cref{theorem:expected_visiting_times_LEQ}.
The following result, which is analogous to~\cite[Thm.~6.3.1]{Put94}, means that VI for EVTs (stated explicitly as \Cref{algorithm:VI} for the sake of concreteness) works for arbitrary initial vectors.
\begin{restatable}{theorem}{VIConvergence}\label{theorem:VI_convergence}
    The EVTs-operator from \Cref{definition:EVTs_operator} has the following properties:
    \begin{enumthm}
        \item\label{theorem:VI:operator_fixpoint}  $(\EXP{\vt_s})_{s \in \statesTr}$ is the unique fixed point of $\vioperator$.
        \item\label{theorem:VI:operator_limit} For all $\xbf^{(0)} \in \evtsdomain{\statesTr}$ we have $\lim_{k \to \infty} \vioperator^{(k)}(\xbf^{(0)})= (\EXP{\vt_s})_{s \in \statesTr}$.
    \end{enumthm}
\end{restatable}

\begin{algorithm}[t]
    \Indm{}  
    \Input{DTMC $\dtmc$, $\xbf^{(0)} \in \evtsdomain{\statesTr}$, $crit \in \{abs, rel\}$, $\epsilon >0$}
    \Output{$\xbf \in \evtsdomain{\statesTr}$}
    \Indp{}    
    \For(){$k=1,2,3,\dots$}{        
%
        $\xbf^{(k)} \leftarrow \vioperator(\xbf^{(k-1)})$\label{algorithm:operator_line}
        \tcp{$\xbf^{(k)}(s) = \iotainit(s) + \sum_{t \in \statesTr} \xbf^{(k-1)}(t) \cdot  \Pbf(t, s)$, $s \in \statesTr$}
        \lIf{$\diff{crit}(\xbf^{(k-1)}, \xbf^{(k)})\leq \epsilon$\label{algorithm:VI:error} }{\Return{$\xbf^{(k)}$}}
    }
    \caption{Value iteration for EVTs without precision guarantee.}\label{algorithm:VI}
\end{algorithm}

\paragraph{When to stop VI?}
A general issue with value iteration is that even if the generated sequence converges to the desired fixed point in the limit, it is not easy to determine how many iterations are necessary to obtain an $\epsilon$-precise result.
An ad hoc solution, which is implemented in probabilistic model checkers such as \storm{}~\cite{STORM}, \prism~\cite{PRISM}, and \mcsta~\cite{HH14}, is to stop the iteration once the difference between two consecutive approximations is small, i.e., the number of iterations is the smallest $k > 0$ such that $\diff{}(\xbf^{(k)}, \xbf^{(k-1)}) < \epsilon$ for some predefined fixed $\epsilon > 0$.
Common choices for the distance $\diff{}$ between vectors $\xbf,\ybf \in \reals^{|S|}$ are the \emph{absolute difference} 
$\diff{abs}(\xbf,\ybf) = \max_{s \in S}  \abs{\xbf(s)- \ybf(s)}$, and the \emph{relative difference}
\begin{align*}
    \diff{rel}(\xbf,\ybf) = \max_{s \in S} \abs*{\frac{ \xbf(s) - \ybf(s) }{\ybf(s)}} ~,
\end{align*}
where by convention $0/0 = 0$ and $a/0 = \infty$ for $a \neq 0$. 
As pointed out by various authors~\cite{Wim+08,HM14,Bai+17,McM+05,QK18}, there exist instances where the iteration terminates with a result which vastly differs from the true fixed point, even if $\epsilon$ is small (e.g. $\epsilon=10^{-6}$).
An example of this for the EVT variant of VI is given in \Cref{ap:VI_unsound} (Appendix B). 

\subsection{Interval Iteration}\label{sec:II}

\emph{Interval iteration (II)}~\cite{McM+05,HM14,Bai+17} is an extension of VI that formally guarantees $\epsilon$-close results for all possible inputs.
The general idea of II is to construct \emph{two} sequences of vectors $(\lbf^{(k)})_{k \in \nat}$ and $(\ubf^{(k)})_{k \in \nat}$ such that for all $k \in \nat$ we have $\lbf^{(k)} \leq \xbf \leq \ubf^{(k)}$, where $\xbf$ is the desired fixed point solution.
II can be stopped with precision guarantee $\epsilon$ once it detects that $\diff{}(\lbf^{(k)}, \ubf^{(k)}) \leq \epsilon$.

\paragraph{Initial bounds for II.}
In general, II requires initial vectors $\lbf^{(0)}$ and $\ubf^{(0)}$ which are already sound (but perhaps very crude) lower and upper bounds on the solution.
In the case of EVTs, we can use $\lbf^{0} = \0$.
Finding an upper bound $\ubf^{(0)} \geq (\EXP{\vt_{s}})_{s \in \statesTr}$ is more involved since EVTs may be unboundedly large in general.
We solve this issue using a technique from~\cite{Bai+17}.
\paragraph{II for EVTs.}
In \Cref{theorem:II_convergence} below we show that once we have found initial bounds $\lbf^{(0)}$ and $\ubf^{(0)}$, we can readily perform a sound II for EVTs by simply iterating the operator $\vioperator$ from \Cref{definition:EVTs_operator} on $\lbf^{(0)}$ and $\ubf^{(0)}$ in parallel.
Inspired by~\cite{Bai+17}, we propose the following optimization to speed up convergence:
Whenever $\vioperator$ \emph{decreases} the current lower bound in some entries, we retain the old values for these entries (and similar for upper bounds).
The next definition formalizes this.
\begin{definition}\label{definition:II_operator} 
    The \emph{Max and Min EVTs-operators} $\maxvioperator, \minvioperator\colon \evtsdomain{\statesTr} \to \evtsdomain{\statesTr}$ are defined by 
    $\maxvioperator(\xbf) = \max \left\{\xbf, \vioperator(\xbf)\right\} = (\max \{\xbf(s) , (\vioperator(\xbf))(s)\})_{s\in \statesTr}$
    and $\minvioperator(\xbf) = \min\{\xbf, \vioperator(\xbf)\} = (\min \{ \xbf(s) , (\vioperator(\xbf))(s)\})_{s\in \statesTr}$.
\end{definition}

\noindent
The following result is analogous to~\cite[Lem.~3.3]{Bai+17}: 
\begin{restatable}{lemma}{IIConvergence}
    \label{theorem:II_convergence}
    Let $\ubf,\lbf \in \evtsdomain{\statesTr}$ with $\lbf \leq (\EXP{\vt_s})_{s \in \statesTr} \leq \ubf$. Then, 
    \begin{enumthm}
        \item $\maxvioperator^{(k)}(\lbf) \leq \maxvioperator^{(k+1)}(\lbf) \quad \text{and} \quad \minvioperator^{(k)}(\ubf) \geq \minvioperator^{(k+1)}(\ubf)$ for all $k \in \nat$.\label{theorem:II:seqmono} 
        \item $\vioperator^{(k)}(\lbf) \leq  \maxvioperator^{(k)}(\lbf) \leq \left(\EXP{\vt_s}\right)_{s\in \statesTr} \leq \minvioperator^{(k)}(\ubf) \leq \vioperator^{(k)}(\ubf)$ for all $k \in \nat$.\label{theorem:II:seqbound}
        \item $ \lim_{k\to \infty} \maxvioperator^{(k)}(\lbf) = \lim_{k \to \infty} \minvioperator^{(k)}(\ubf) = \left(\EXP{\vt_s}\right)_{s\in \statesTr}$.\label{theorem:II:seqconv}
    \end{enumthm}
\end{restatable}

\noindent
The resulting II algorithm for EVTs is presented as \Cref{alg:II}.
Note the following additional optimization:
The algorithm stops as soon as $\diff{crit}(\ubf^{(k)}, \lbf^{(k)}) \leq 2 \epsilon$ and returns the mean of $\ubf^{(k)}$ and $\lbf^{(k)}$, ensuring that the absolute or relative difference between $(\EXP{\vt_s})_{s \in \statesTr}$ and the returned result is at most $\epsilon$.

\begin{example}
    We illustrate a run of \Cref{alg:II} on the DTMC from \Cref{fig:runningExample}, with $crit = abs$ and $\epsilon=0.05$ (the following numbers are rounded to 4 decimal digits):
    \begin{center}
        {\setlength{\tabcolsep}{5pt}
        \renewcommand{\arraystretch}{1}
        \begin{tabular}{r c c c}
            $k$ & $\lbf^{(k)}$ & $\ubf^{(k)}$ & $\diff{abs}(\lbf^{(k)}, \ubf^{(k)})$ \\[-1pt]
            \midrule
            $0$ & $(0.000,0.000,0.000,0.000)$ & $(2.000,2.000,1.000,5.000)$ & $5.000$ \\
            $1$ & $(0.400,0.000,0.600,0.000)$ & $(2.000,1.000,0.600,5.000)$ & $5.000$ \\
            $2$ & $(0.820,0.200,0.600,0.380)$ & $(1.820,1.000,0.600,5.000)$ & $4.602$ \\
            & $\cdots$ & $\cdots$ & $\cdots$ \\
            $22$ & $(1.639, 0.819, 0.600, 4.899)$ & $(1.640, 0.820, 0.600, 5.000)$ & $0.101$\\
            $23$ & $(1.639, 0.819, 0.600, 4.919)$ & $(1.640, 0.820, 0.600, 5.000)$ & $\mathbf{0.081}$ 
        \end{tabular}}
    \end{center}
    After $k=23$ iterations, the algorithm stops as $\diff{abs}(\lbf^{(k)}, \ubf^{(k)}) = 0.081 \leq 2 \cdot \epsilon$ and outputs the mean $\tfrac{1}{2}(\lbf^{(23)} + \ubf^{(23)})$.
\end{example}

\begin{algorithm}[t]
  \Indm{} 
    \Input{DTMC $\dtmc$, $\lbf^{(0)} \leq (\EXP{\vt_{s}})_{s \in \statesTr }, \ubf^{(0)} \geq (\EXP{\vt_{s}})_{s \in \statesTr }$, $crit \in \{abs, rel\}, \epsilon >0$}
    \Output{$\xbf \in \evtsdomain{\statesTr}$ with $\diff{crit}(\xbf, (\EXP{\vt_{s}})_{s \in \statesTr}) \leq \epsilon$} 
  \Indp{}
        
      \For(){$k=1,2,\dots$}{
            
            $\lbf^{(k)}\leftarrow  \maxvioperator(\lbf^{(k-1)})$~;\quad $\ubf^{(k)}\leftarrow \minvioperator(\ubf^{(k-1)})$

\lIf{$\diff{crit}(\ubf^{(k)}, \lbf^{(k)}) \leq 2\cdot\epsilon$\label{alg:II:error}}{    \Return{$\frac{1}{2}(\lbf^{(k)}+\ubf^{(k)})$} 
}

    }

    \caption{Interval iteration for EVTs with precision guarantee.}\label{alg:II}
\end{algorithm}

\begin{restatable}[{Correctness of \Cref{alg:II}}]{theorem}{IICorrectness}\label{theorem:II_correctness}  
    Given an input DTMC $\dtmc$, initial vectors $\lbf^{(0)}$, $\ubf^{(0)}$ with $\lbf^{(0)} \leq (\EXP{\vt_s})_{s \in \statesTr } \leq \ubf^{(0)}$, $crit\in \{abs,rel\}$, and a threshold $\epsilon>0$, 
    \Cref{alg:II} terminates and returns a vector $\xbf^{res} \in \evtsdomain{\statesTr}$ satisfying $\diff{crit} (\xbf^{res},(\EXP{\vt_s})_{s \in \statesTr}) \leq \epsilon$. 
\end{restatable}

\begin{remark}
The monotonicity of the sequences $(\maxvioperator^{(k)}(\lbf))_{k \in \nat}$ and $(\minvioperator^{(k)}(\ubf))_{k \in \nat}$ (see \Cref{theorem:II_convergence} \Cref{theorem:II:seqmono}) is not used in the proof of \Cref{theorem:II_correctness}. 
By \Cref{theorem:II_convergence} \Cref{theorem:II:seqbound}, we can replace $\maxvioperator$ and $\minvioperator$ with $\vioperator$ in \Cref{alg:II} and still obtain $\epsilon$-sound results. 
However, using $\maxvioperator$ and $\minvioperator$ instead of $\vioperator$ can lead to faster convergence. 
\end{remark}

\subsection{Topological Algorithm}\label{sec:topological_algorithms}

To increase the efficiency of VI for the analysis of rewards and probabilities in MDPs, several authors have proposed topological VI~\cite{DG07,Dai+11}, which is also known as blockwise VI~\cite{Cie+08}. 
The idea is to avoid the analysis of the complete model at once and instead consider the strongly connected components (SCCs)  sequentially based on the order relation $\hookrightarrow$.  
We lift this approach to EVTs and in particular consider error propagations when approximative methods are used.

\paragraph{SCC Restrictions.}
To formalize the topological approach, we introduce the \emph{SCC restriction} $\dtmcRestr{C}[\xbf]$ of DTMC $\dtmc$ to $C \in \scc{\dtmc}$ with parameters $\xbf \in \reals^{|\statesTr|}$.
Intuitively, $\dtmcRestr{C}[\xbf]$ is a DTMC-like model obtained by restricting $\dtmc$ to the states $C$ and assigning each $s \in C$ the ``initial value'' $\iotainit^{\dtmc}(s) + \sum_{s' \in S\setminus C} \Pbf^{\dtmc}(s',s) \cdot \xbf(s')$.
The idea is that $\xbf$ is an approximation of the EVTs of the predecessor SCCs $C' \hookrightarrow C$.
We also define $\dtmcRestr{C} = \dtmcRestr{C}[\xbf]$ with $\xbf(s) = \EXPstateMC{\vt_s}{}{\dtmc}$ (i.e., the \emph{exact} EVT) for each state $s \in S$ that can reach $C$ with positive probability in one step. 
See \Cref{definition:param_SCC_restriction} in \Cref{ap:topo} for more formal details. 
\newline 
\begin{wrapfigure}[7]{r}{0.42\linewidth}
    \raggedleft
    \begin{tikzpicture}[node distance=10mm and 12.5mm]
        \node[state] (s2) {$s_2$};
        \node[state,right=of s2,initial,initial text=$0.4 + 0.7 \cdot \textcolor{red}{0.6}$, initial where=above] (s1) {$s_1$};
        \node[gray,dashed,state,label={below,xshift=3mm:$\xbf(s_3){=}\textcolor{red}{0.6}$},right=of s1] (s3) {$s_3$};
        \node[below=5mm of s3] (ghost) {};
        
        \path[->] (s1) edge[bend left] node[above] {$0.5$} (s2);
        \path[->] (s1) edge[dashed,gray,bend right=15] node[below left] {$0.5$} (ghost);
        \path[->] (s2) edge[bend left] node[auto] {$1$} (s1);
        \path[->] (s3) edge[gray,dashed] node[above] {$0.7$} (s1);
    \end{tikzpicture}
    \caption{SCC restriction.}
    \label{fig:sccRestrExample}
\end{wrapfigure}
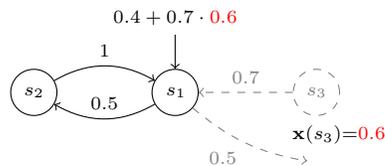

\begin{example}
    The SCC restriction of the DTMC $\dtmc$ from \Cref{fig:runningExample} to the SCC $C = \{s_1,s_2\}$ with parameters $\xbf = (\EXPstateMC{\vt_s}{}{\dtmc})_{s \in S}$ is depicted in \Cref{fig:sccRestrExample}.
    Note that the initial values depend only on the initial distribution $\iotainit^\dtmc$ and the $\xbf$-values of the states that can reach the SCC $C$ in one step.
\end{example}

\begin{restatable}{lemma}{EVTsInSCCRestriction}\label{theorem:EVTs_in_SCC_restriction}
    For non-bottom SCC $C$ and $s \in C$ we have $\EXPstateMC{\vt_s}{}{\dtmc} = \EXPstateMC{\vt_s}{}{\dtmcRestr{C}}$. 
\end{restatable}

\begin{remark}\label{remark:sccrestr}
    Since a parametric SCC restriction $\dtmcRestr{C}[\xbf]$ is defined for arbitrary vectors $\xbf\in \evtsdomain{\statesTr}$, the initial values do not necessarily form a probability distribution. 
    Strictly speaking, this means that $\dtmcRestr{C}[\xbf]$ is not a DTMC in general. 
    Thus, by abuse of notation, we define the ``EVTs'' in the parametric SCC restriction $\dtmcRestr{C}[\xbf]$ as the unique solution 
    of the following linear equation system: 
    For all $s \in C$: $\xbf(s) = \iotainit^{\dtmcRestr{C}[\xbf]}(s)+  \sum_{t \in \statesTr} \Pbf^{\dtmcRestr{C}}(t, s) \cdot \xbf(t)$. 
    To solve this system, we can still apply the methods described in \Cref{sec:VI,sec:II} without further ado. 
\end{remark}

\begin{algorithm}[t]
    \Indm{} 
    \Input{DTMC $\dtmc$, $\epsilon \geq 0$}
    \Output{$\xbf \in \evtsdomain{\statesTr}$ with $\diff{rel}(\xbf, (\EXP{\vt_{s}})_{s \in \statesTr}) \leq \epsilon$}    
    \Indp{}
    $\{C_1, \dots ,C_n\} \leftarrow \sccTr{\dtmc}$\label{alg:topo_exact:line:SCC} \tcp*{obtain the non-bottom SCCs, $C_i \hookrightarrow C_j$ implies $i<j$}
    
    $\xbf \leftarrow \0$\label{alg:topo_exact:line:init}
    
    $\delta \leftarrow \sqrt[L+1]{1+\epsilon}-1$ \tcp{$L$ is the length of a longest non-bottom SCC chain}
    
    \For{$j = 1$ \KwTo{} $n$}
    { 
        
        compute $\widehat{\xbf}$ such that $\diff{rel}((\EXPstateMC{\vt_s}{}{\dtmcRestr{C_j}[\xbf]})_{s \in C_j}, \widehat{\xbf}) \leq \delta$ \tcp{use e.g.~Alg.~\ref{alg:II}}\label{alg:topo_exact:line:compute} 
        
        \lFor{$s \in C_j$}
        { 
            ${\xbf(s)} \leftarrow \widehat{\xbf}(s)$\label{alg:topo_exact:line:update}
        }
    }
    \Return{$\xbf$}
    \caption{Topological EVT algorithm with relative precision.}\label{alg:topo_exact}
\end{algorithm}

We now describe an algorithm for computing the EVTs SCC-wise in topological order. 
The desired precision $\epsilon \geq 0$ is an input parameter ($\epsilon = 0$ is possible).
Due to space limitations we only discuss relative precision, see~\Cref{ap:topo_abs}for an algorithm with absolute precision. 
The idea is to solve the linear equation systems tailored to the parametric SCC restrictions, each of which is constructed based on the analysis of the preceding SCCs.
\Cref{alg:topo_exact} outlines the procedure. 

The algorithm first decomposes the input DTMC $\dtmc$ into its non-bottom SCCs: 
The function $\sccTr{\dtmc}$ called in \Cref{alg:topo_exact:line:SCC} returns the set $\{C_1,\dots ,C_n\}$ of SCCs of $\dtmc$ which consist of transient states, i.e., the non-bottom SCCs. 
Furthermore, we assume that the SCCs are indexed such that $C_i \hookrightarrow C_j$ implies $i<j$. 
The algorithm considers each SCC in topological order.
We assume that a ``black box'' can obtain the (approximate) EVTs in \Cref{alg:topo_exact:line:compute}. 
Then, the vector $\xbf \in \evtsdomain{\statesTr}$
is updated such that it contains the approximations of the EVTs in $\dtmc$ upon termination. 
For the analysis of an SCC $C_j$ for some $j \in \{1, \dots ,n\}$, the algorithm considers the parametric SCC restriction $\dtmcRestr{C_j}[\xbf]$, which is based on the result $\xbf(t)$ for states $t$ in SCCs $C_i$ that are topologically before $C_j$. 
For each parametric SCC restriction $\dtmcRestr{C_j}[\xbf]$, the algorithm computes the vector $(\EXPstateMC{\vt_s}{}{\dtmcRestr{C_j}[\xbf]})_{s \in C}$ in \Cref{alg:topo_exact:line:compute}. 
Then, the algorithm updates the corresponding entries in $\xbf$ in \Cref{alg:topo_exact:line:update}. 
After each non-bottom SCC has been considered, the algorithm terminates and returns the vector $\xbf$ that contains the updated value for each transient state. 
The following lemma provides an upper bound on the error that accumulates during the topological computation.

\begin{restatable}{lemma}{topoBoundRel}\label{theorem:topo_bound_rel}
    Let $\epsilon \in [0,1)$ and let  
    $\xbf \in \evtsdomain{\statesTr}$ such that for every non-bottom SCC $C$, 
    $(\xbf(s))_{s\in C}$ satisfies $\diff{rel}\left((\xbf(s))_{s\in C}, (\EXPstateMC{\vt_s}{}{\dtmcRestr{C}[\xbf]})_{s \in C} 
      \right) \leq \epsilon$.
    Then
    \begin{align*} 
      \diff{rel}\left(\xbf, (\EXPstateMC{\vt_s}{}{\dtmc})_{s \in \statesTr} \right) 
      \leq  (1+\epsilon)^{L+1} -1, 
    \end{align*}
    where $L= \max_{\kappa \in \chainsTr^\dtmc}\abs{\kappa}$ is the largest length of a chain of non-bottom SCCs. 
\end{restatable}

\begin{restatable}[{Correctness of \Cref{alg:topo_exact}}]{theorem}{topoCorrectnessRel}\label{theorem:topo_correctness_rel}
\Cref{alg:topo_exact} returns a vector $\xbf^{res} \in \evtsdomain{\statesTr}$ such that $\diff{rel}(\xbf^{res}, (\EXPstateMC{\vt_s}{}{\dtmc})_{s \in \statesTr}) \leq \epsilon$.
\end{restatable}

\section{Stationary Distributions via EVTs}\label{sec:stationary}
We show that EVTs can be used to determine \emph{sound} approximations of the stationary distribution $\stst{\dtmc}$ of a (reducible) DTMC $\dtmc$. 
It is known that $\stst{\dtmc}$ can be computed as follows (see, e.g.,~\cite[Thm.~4.23]{Kul20}). 
For each BSCC $B$:
\begin{itemize}
    \item Compute the reachability probability $\PRstateMC{}{\dtmc}(\event B)$ to $B$. 
    \item Determine the stationary distribution $\stst{\dtmcRestr{B}}$ of the DTMC restricted to $B$.
\end{itemize}
Then, we obtain the stationary distribution for the recurrent states $s$ in the BSCC $B$: 
$\stst{\dtmc}(s) = \PRstateMC{}{\dtmc}(\event B) \cdot \stst{\dtmcRestr{B}}(s)$. 
For the remaining transient states, we have $\stst{\dtmc}(s) = 0$. 
We show how both, $\PRstateMC{}{\dtmc}(\event B)$ and $\stst{\dtmcRestr{B}}$, can be computed efficiently for every BSCC $B$ using EVTs.
We also elaborate on how relative errors propagate through the computation, allowing us to derive sound lower- and upper bounds for the stationary distribution $\stst{\dtmc}$.
Omitted proofs are in~\Cref{ap:stationary}. 

\paragraph{Computing BSCC reachability probabilities.}\label{sec:reachability_probabilities}
The \emph{absorption probabilities}, i.e., the probability of reaching a singleton BSCC can be computed using EVTs~\cite{KS76}. 
By collapsing the BSCCs into a single state, a slightly generalised result is obtained: 
\begin{theorem}[{\cite[Thm.~3.3.7]{KS76}}]\label{theorem:EVTs_absorPr}
    For any BSCC $B$ of a DTMC $\dtmc$ it holds that
    $\PRstateMC{}{}(\event B) =  \sum_{s \in B}  \left(\iotainit(s) + \sum_{s' \in \statesTr} \Pbf(s', s)  \cdot \EXPstateMC{\vt_{s'}}{}{} \right)$.
\end{theorem}
Applying \Cref{theorem:EVTs_absorPr}, we can compute the EVTs \emph{once} to derive the reachability probabilities for \emph{every} BSCC.
Further, when using interval iteration from \Cref{sec:II} to obtain $\xbf \in \evtsdomain{S}$ with $\diff{rel}\left(\xbf, (\EXPstateMC{\vt_s}{}{})_{s \in S} \right) \leq \epsilon$ for some $\epsilon \in (0,1)$, the relative error does not increase when deriving the reachability probabilities, i.e.,
$      \diff{rel}\left(\sum_{s \in B}  \left(\iotainit(s) + \sum_{s' \in \statesTr} \Pbf(s', s)  \cdot \xbf(s) \right),\PRstateMC{}{}(\event B) \right)  
      \leq \epsilon.
$

\paragraph{Computing the stationary distribution within a BSCC.}
Next, we leverage EVTs to compute the stationary distribution $\stst{\mathcal{B}}$ of an irreducible DTMC $\mathcal B = (S, \Pbf, \iotainit)$.
This method can be applied to derive the stationary distribution $\stst{\dtmcRestr{B}}$ of $\dtmcRestr{B}$ since the latter is an irreducible DTMC for each BSCC $B$. 
Let $v \in S$ be an arbitrary state. 
We construct the DTMC ${\dtmcBv{}}$ in which all incoming transitions of state $v$ are redirected to a fresh absorbing state $\hat{v}$.
Thus, its only BSCC is $\{\hat{v}\}$, all other states are transient. 
Formally, ${\dtmcBv{}} = (S \uplus \{\hat{v}\}, \hat{\Pbf}, \hatiotainit)$, where $\hat{\Pbf}(\hat{v}, \hat{v})=1$, 
$\hat{\Pbf}(s,\hat{v}) = \Pbf(s,v)$ and $\hat{\Pbf}(s,v) = 0$ for all $s \in S$, $\hat{\Pbf}(s,t) = \Pbf(s,t)$ for all $s \in S, t \in S \setminus \{v\}$, and $\hatiotainit(v)= 1$.
\begin{restatable}{theorem}{stationaryViaEVTs}\label{theorem:stationary_via_EVTs}
    The stationary distribution $\stst{\mathcal{B}}$ of an irreducible DTMC $\mathcal{B}$ is given by $\stst{\mathcal{B}}(s) = \frac{\EXPstateMC{\vt_s}{}{{\dtmcBv{}}}}{\sum_{t\in S} \EXPstateMC{\vt_t}{}{{\dtmcBv{}}}}$.
    Further, if $\diff{rel}\left(\xbf, (\EXPstateMC{\vt_s}{}{{\dtmcBv{}}})_{s \in S} 
      \right) \leq \epsilon$ for $\xbf \in \evtsdomain{S}$ and $\epsilon \in (0,1)$, then 
$      \diff{rel}\left(\tfrac{\xbf}{\sum_{s\in S} \xbf(s)}, \stst{\mathcal{B}}\right)  
      \leq \frac{2\epsilon}{1-\epsilon}
$.
\end{restatable}
\noindent The first part of~\Cref{theorem:stationary_via_EVTs} can also be established by considering the renewal processes embedded in $\mathcal{B}$ (see, e.g.,~\cite[Theorem 2.2.3]{Vol05}).

\paragraph{Combining both steps}. \Cref{theorem:EVTs_absorPr,theorem:stationary_via_EVTs} and the interval iteration method yield approximations $p_B$ and $d_s$ for $\PRstateMC{}{\dtmc}(\event B)$ and $\stst{\dtmcRestr{B}}(s)$, respectively, where $B$ is a BSCC $s \in B$, and $\epsilon_1, \epsilon_2 \in (0,1)$ such that
$\diff{rel}\left(p_B, \PRstateMC{}{\dtmc}(\event B)\right) \le \epsilon_1$ and $\diff{rel}\left(d_s, \stst{\dtmcRestr{B}}(s)\right) \le \epsilon_2$.
The product $p_B \cdot d_s$ approximates $\stst{\dtmc}(s) = \PRstateMC{}{\dtmc}(\event B) \cdot \stst{\dtmcRestr{B}}(s)$ such that
$\diff{rel}\left(p_B \cdot d_s, \stst{\dtmc}(s)\right) \le \epsilon_1 + \epsilon_2 + \epsilon_1  \epsilon_2$.\newline 
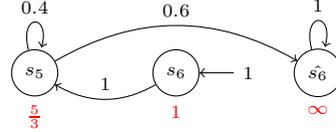
\begin{wrapfigure}[9]{r}{0.38\linewidth}
    \raggedleft
    \begin{tikzpicture}[node distance=10mm and 12.5mm]
        \node[state,label={below,red:$\tfrac 5 3$}] (s5) {$s_5$};
        \node[state,initial, initial where=right, initial text=$1$,label={below,red:$1$},right=of s5] (s6) {$s_6$};
        \node[state,label={below,red:$\infty$},right=of s6] (s6hat) {$\hat{s_6}$};
        
        \path[->] (s5) edge[loop above] node[auto] {$0.4$} (s5);
        \path[->] (s5) edge[bend left] node[auto] {$0.6$} (s6hat);
        \path[->] (s6) edge[bend left] node[above] {$1$} (s5);
        \path[->] (s6hat) edge[loop above] node[above] {$1$} (s6hat);
    \end{tikzpicture}
    \caption{The DTMC $\mathcal{B}{\Rsh}^{\hat{s_6}}$.}
    \label{fig:stationary}
\end{wrapfigure}

\begin{example}
    We compute the stationary distribution of the running example DTMC $\dtmc$ from \Cref{fig:runningExample}.
    Its only non-trivial BSCC is $\mathcal B = \dtmcRestr{\{s_5,s_6\}}$.
    The DTMC $\mathcal{B}{\Rsh}^{\hat{s_6}}$ along with its (exact) EVTs is depicted in \Cref{fig:stationary}.
    We conclude that $\stst{\mathcal B}$ is proportional to $(\tfrac{3}{5}, 1)$, i.e., $\stst{\mathcal B} = (\tfrac{5}{8}, \tfrac{3}{8})$.
    Since the two BSCCs $\{s_5, s_6\}$ and $\{s_7\}$ are both reached with probability $\tfrac 1 2$, it follows that the stationary probabilities of the three recurrent states $s_5, s_6, s_7$ are $\tfrac{5}{16}, \tfrac{3}{16}$, and $\tfrac 1 2$, respectively.
\end{example}

\section{Conditional Expected Rewards}\label{sec:conditional_rew}

\newcommand{\CEXPstateMCcond}[4]{\ensuremath{\mathbb{E}^{#3}_{#2}[#1 | #4]}}
\newcommand{\rhoinit}{\rho_{\text{init}}}

\Cref{theorem:EVTs_absorPr} states that the EVTs of the transient states of a DTMC $\dtmc$ can be used to compute the probability to reach each individual BSCC of $\dtmc$.
We now generalize this result and show that the EVTs can also be used to compute the \emph{total expected rewards conditioned on reaching each BSCC}.

The total expected reward conditioned on reaching a set $T \subseteq S$ of states with $\PRstateMC{}{\dtmc}(\event T) > 0$ is defined as:
\begin{align*}
    \CEXPstateMCcond{\tr_\rew}{}{\dtmc}{\event T}
    =
    \frac{1}{\PRstateMC{}{\dtmc}(\event T)} \int_{\event T} \tr_{\rew} \,d \PRstateMC{}{\dtmc}
\end{align*}
Our next result asserts that given the EVTs of $\dtmc$, all the values $\{\CEXPstateMCcond{\tr_\rew}{}{\dtmc}{\event B} \mid B \text{ a BSCC of } \dtmc \}$ can be computed by solving a \emph{single} linear equation system (the standard approach is to solve one linear equation system \emph{per BSCC}~\cite[Ch.~10]{BK08}).
For simplicity, we state the result only for BSCCs $\{r\}$ with a single absorbing state $r$, and for reward functions that assign zero reward to all recurrent (BSCC) states;
this is w.l.o.g.\ as larger BSCCs can be collapsed, and positive reward in a (reachable) BSCC causes the conditional expected reward \wrt{} this BSCC to be $\infty$, rendering numeric computations unnecessary. 
\begin{restatable}{theorem}{condExpRew}\label{theorem:cond_exp_rew}
    Let $\rew \colon S \to \Rationals_{\geq 0}$ with $\rew(\statesRe) = \{0\}$.
    Then the equation system
    \begin{align*}
        \forall s \in \statesTr \colon
        \qquad
        \ybf(s) = \rew(s) \cdot \EXPstateMC{\vt_s}{}{\dtmc} + \sum_{t \in \statesTr} \Pbf(t, s) \cdot \ybf(t)
    \end{align*}
    has a unique solution $(\ybf(s))_{s \in \statesTr}$ and for all absorbing $r \in S$, $\PRstateMC{}{\dtmc}(\event \{r\}) > 0$, 
    \begin{align*}
        \CEXPstateMCcond{\tr_\rew}{}{\dtmc}{\event \{r\}}
        =
        \frac{\sum_{t \in \statesTr} \Pbf(t,r) \cdot \ybf(t)}{\iotainit(r) + \sum_{t \in \statesTr} \Pbf(t,r) \cdot \EXPstateMC{\vt_t}{}{\dtmc}}
        ~.
    \end{align*}
\end{restatable}

\noindent
\Cref{theorem:cond_exp_rew} assumes rational rewards as required in our proof in~\Cref{ap:conditional_reward}. 

\section{Experimental Evaluation}\label{sec:experiments}

\paragraph{Implementation details.}
We integrated the presented algorithms for EVTs and stationary distributions in the model checker \storm~\cite{STORM}. 
The implementation is part of \storm's main release available at \url{https://stormchecker.org}. 
It uses explicit data structures such as sparse matrices and vectors. 
When computing EVTs, we can use value iteration (\method{VI}) or interval iteration (\method{II}) as presented in \Cref{algorithm:VI,alg:II}.
Alternatively, the corresponding linear equation systems can be solved using \method{LU} factorization (a direct method implemented in the \tool{Eigen} library~\cite{eigenweb}) or \method{gmres} (a numerical method implemented in \tool{gmm++}~\cite{getfemweb}).
Each EVT approach can be used in combination with the topological algorithm (\method{topo})  from \Cref{sec:topological_algorithms}.
We use double precision floating point numbers.
For \method{II}, the propagation of (relative) errors is respected in a way that the error of the end result does not exceed a user-defined threshold (here: $\epsilon = 10^{-3}$).
Implementing \method{II} with safe rounding modes as in~\cite{Har22} is left for future work.
The methods \method{gmres} and \method{VI} are  configured with a fixed relative precision parameter (here: $\epsilon = 10^{-6}$).
For \method{LU}, floating point errors are the only source of inaccuracies.
We also consider an exact configuration \method{LU$^X$} that uses rational arithmetic instead of floats.

Stationary distributions can be computed in \storm{} using the approaches \method{Classic}, \method{EVTreach}, and \method{EVTfull}.
The \method{Classic} approach computes each BSCC reachability probability separately and the stationary distributions within the BSCCs are computed using the standard equation system~\cite[Thm.~4.18]{Kul20}. 
\method{EVTreach} and \method{EVTfull} implement our approaches from \Cref{sec:stationary}, where \method{EVTreach} only considers EVTs for  BSCC reachability and \method{EVTfull} also derives the BSCC distributions from EVTs.
As for EVTs, we use \method{LU$^{(X)}$} or \method{gmres} to solve linear equation systems.  For the BSCC reachability probabilities, a topological algorithm can be enabled as well.
Using \method{EVTfull} with \method{II} yields sound approximations.

\paragraph{Experimental setup.}
The experiments ran on an Intel$^{\mbox{\scriptsize\textregistered}}$ Xeon$^{\mbox{\scriptsize\textregistered}}$ Platinum 8160 Processor limited to 4 cores and 12 GB of memory. The time timeout was set to 30 minutes. Our implementation does not use multi-threading.

\subsection{Verifying the Fast Dice Roller}
\begin{wrapfigure}[16]{r}{0.475\linewidth}
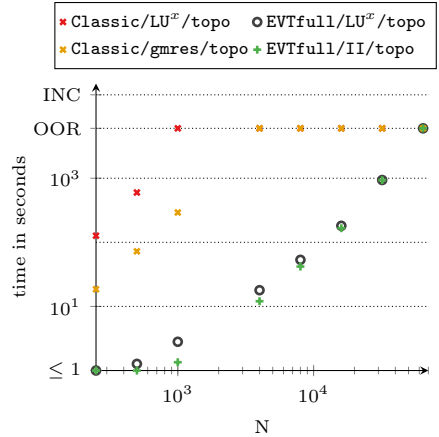
 
    \centering
	\vspace{-0.8cm}
	\setlength{\numberruntimeplotheight}{0.9\scatterplotsize}
	\renewcommand{\legendcols}{2}
	\renewcommand{\legendstyle}{at={(0.45,1.05)},anchor=south}
	{\numberruntimeplot{stationary_csv/lumbroso/wallclock-time-1800-INC-relative-error-max-norm-value-texpgf-scatter.csv"}%
	{
	storm.exact.topological-classic-sparselu.ignored,
	storm.exact.topological-evt-e-sparselu.ignored,
	storm.sparse-rel.topological-classic-g-gmres.1e-06,
	storm.sound-rel.topological-evt-n-ii.0.001}
	{
	{\method{Classic/LU$^x$/topo}}, 
	{\method{EVTfull/LU$^x$/topo}}, 
	{\method{Classic/gmres/topo}},
	{\method{EVTfull/II/topo}}}
	{N}
	{time in seconds}{Parameters}{70000}
	}
	\caption{Fast Dice Roller Results.}\label{fig:runtime_lumbroso}
\end{wrapfigure}
Recall Lumbroso's Fast Dice Roller~\cite{Lum13} from \Cref{sec:intro}.
For a given parameter $N\ge1$, we verify that the resulting distribution is indeed uniform by computing the stationary distribution of the corresponding DTMC which, for this model family, coincides with the individual BSCC reachability probabilities as each BSCC consists of a single state.
We conducted our experiments with an equivalent state-reduced variant of the Fast Dice Roller which we obtained automatically using the technique from~\cite{DBLP:conf/vmcai/WinklerLK22}, i.e., for every given $N$, our variant has fewer states than the original algorithm from~\cite{Lum13}.
The plot in \Cref{fig:runtime_lumbroso} shows for different values of $N$ (x-axis) the runtime (y-axis) of the approaches \method{Classic} (using \method{gmres} or \method{LU$^x$}) and \method{EVTfull} (using \method{II} or \method{LU$^x$}), all using topological algorithms. 
Our novel \method{EVTfull} approach is significantly faster than the \method{Classic} method, enabling us to verify large instances with up to 4\,800\,255 states ($N=32\,000$) within the time limit. 
In particular, we can compute the values using exact arithmetic as \method{LU$^x$} has runtimes similar to those of \method{II} in the \method{EVTfull} approach.

\subsection{Performance Comparison}
To evaluate the various approaches, we computed EVTs and stationary distributions for all applicable finite models of type DTMC or CTMC of the Quantitative Verification Benchmark Set (QVBS) \cite{Har+19}: 
We excluded 2 model families for which none of the tested parameter valuations allowed any algorithm to complete using the available resources  
and for the EVT computation we excluded 8 models that do not contain any transient states. 
In addition to QVBS, we included Lumbroso's Fast Dice Roller~\cite{Lum13} and the handcrafted models (\method{branch} and \method{loop}) introduced in~\cite{Meg23}. 
We considered multiple parameter valuations 
yielding a total of 62 instances (including 12 CTMCs) for computing EVTs  and 
79 instances (including 28 CTMCs) for computing the stationary distribution. 

Our experiments for stationary distributions also include the implementations of the naive and guided sampling approaches (\method{ap-naive} and \method{ap-sample}) from~\cite{Meg23} as well as the implementation of the \method{classic} approach in \prism~\cite{PRISM}\footnote{We consider the \method{explicit} engine of \prism v4.8 with the \method{Jacobi} method (default).} as external baselines.
These tools do not support the \jani models and can not compute EVTs. 

We measured the runtime of the respective computation including the time for model construction. 
In cases where the exact results are known (using exact computations via \method{LU$^X$}), we consider results as \emph{incorrect} if the relative difference to the exact value is greater than $10^{-3}$.
Results provided by the tool from~\cite{Meg23} and by \prism{} are not checked for correctness.
We set the relative termination threshold of \method{gmres} and \method{VI} to $\epsilon = 10^{-6}$ to compensate for inaccuracies of the unsound methods. 
When using \method{II} or \prism, a relative precision of $\epsilon = 10^{-3}$ was requested. 
For the implementation of~\cite{Meg23} --- which exclusively support absolute precision --- we set the threshold to $\epsilon = 10^{-3}$.
See~\Cref{app:evaluation} for more experiments. 

\begin{figure}[t]
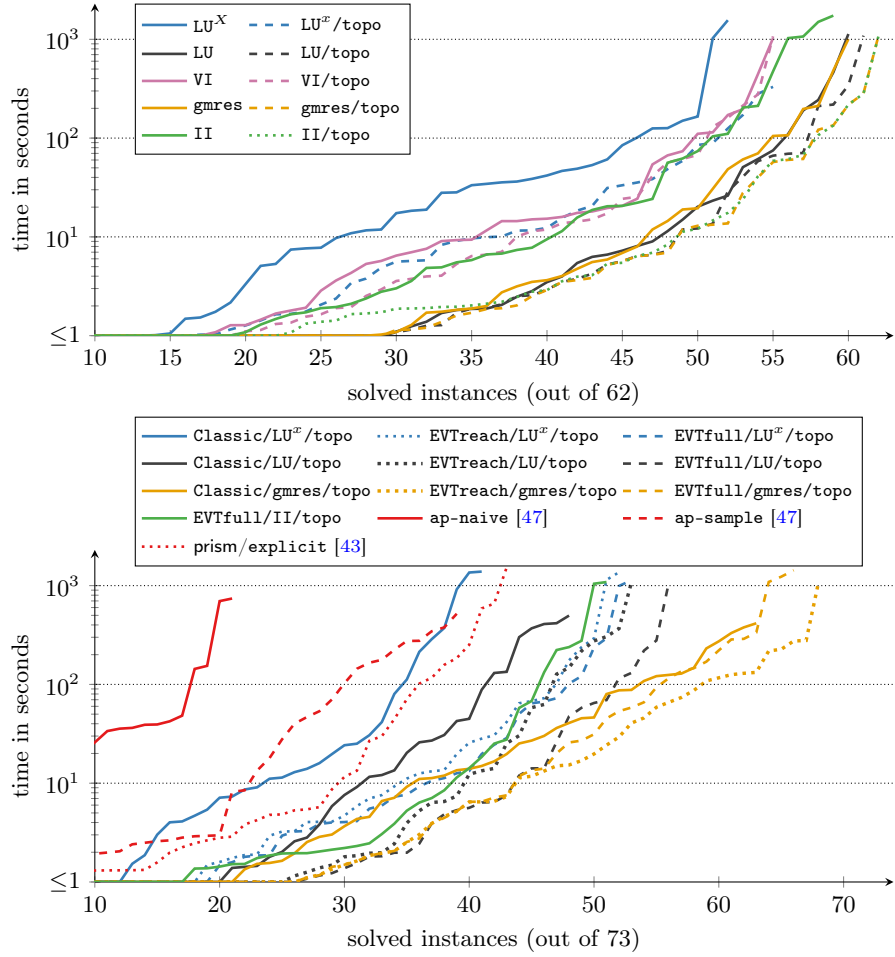

	\renewcommand{\plotwidth}{\linewidth}
	\renewcommand{\plotheight}{6cm}
	\renewcommand{\legendstyle}{at={(0.05,1)},anchor=north west}
	\renewcommand{\legendcols}{2}
	\quantileplot{evts_csv/wallclock-time-INC-relative-error-max-norm-value-quantile.csv"}%
	{storm.exact.e-sparselu.ignored,
	storm.exact.topological-e-sparselu.ignored,
	storm.sparse.e-sparselu.ignored,
	storm.sparse.topological-e-sparselu.ignored,
	storm.sparse-rel.n-power.1e-06,
	storm.sparse-rel.topological-n-power.1e-06,
	storm.sparse-rel.g-gmres.1e-06,
	storm.sparse-rel.topological-g-gmres.1e-06,
	storm.sound-rel.n-ii.0.001,
	storm.sound-rel.topological-n-ii.0.001}%
	{{\method{LU$^X$}},
	{\method{LU$^x$/topo}},
	{\method{LU}}, 
	{\method{LU/topo}},
	{\method{VI}}, 
	{\method{VI/topo}},
	{\method{gmres}}, 
	{\method{gmres/topo}},
	{\method{II}},
	{\method{II/topo}}}%
	{62}

	\renewcommand{\legendstyle}{at={(0.05,0.95)},anchor=south west}
	\renewcommand{\legendcols}{3}
	\quantileplot{stationary_csv/wallclock-time-INC-relative-error-max-norm-value-quantile.csv"}%
		{storm.exact.topological-classic-sparselu.ignored,
		storm.exact.topological-eqsys-e-sparselu.ignored,
		storm.exact.topological-evt-e-sparselu.ignored,
		storm.sparse.topological-classic-sparselu.ignored,
		storm.sparse.topological-eqsys-e-sparselu.ignored,
		storm.sparse.topological-evt-e-sparselu.ignored,
		storm.sparse-rel.topological-classic-g-gmres.1e-06,
		storm.sparse-rel.topological-eqsys-g-gmres.1e-06,
		storm.sparse-rel.topological-evt-g-gmres.1e-06,
		storm.sound-rel.topological-evt-n-ii.0.001,
		megg.default-abs.ap-naive.0.001,
		megg.default-abs.ap-sample.0.001,
		prism.explicit-rel.default.0.001}
		{{\method{Classic/LU$^x$/topo}}, 
		{\method{EVTreach/LU$^x$/topo}},
		{\method{EVTfull/LU$^x$/topo}},
		{\method{Classic/LU/topo}}, 
		{\method{EVTreach/LU/topo}}, 
		{\method{EVTfull/LU/topo}}, 
		{\method{Classic/gmres/topo}}, 
		{\method{EVTreach/gmres/topo}}, 
		{\method{EVTfull/gmres/topo}},
		{\method{EVTfull/II/topo}},
		{\method{ap-naive}~\cite{Meg23}},
		{\method{ap-sample}~\cite{Meg23}},
		{\prism/\method{explicit}~\cite{PRISM}}}%
		{73}
	\caption{Runtime comparison for EVTs (top) and stationary distributions (bottom).}
	\label{fig:runtime_quantile_evt}
	\label{fig:runtime_quantile_stationary}
\end{figure}

\begin{figure}[t]
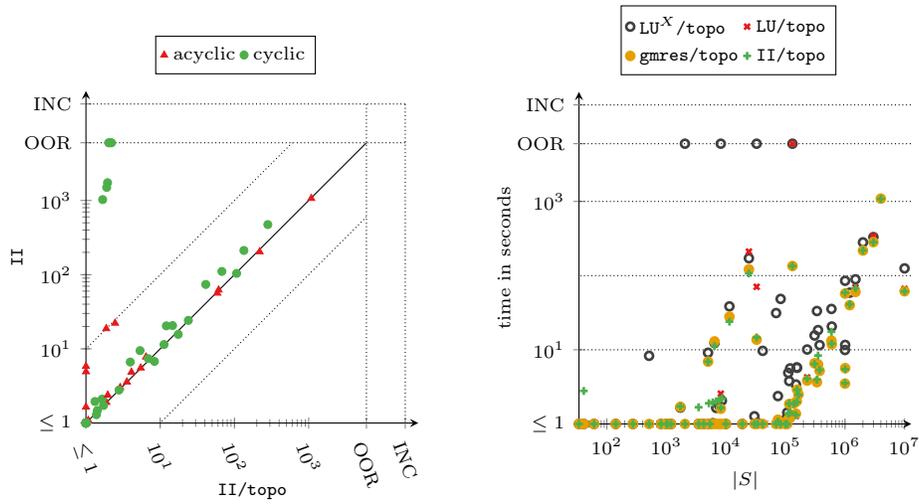

			\centering
			\renewcommand{\legendcols}{2}
			\renewcommand{\legendstyle}{at={(0.45,1.05)},anchor=south}
			\scatterplotstormcyclic{evts_csv/wallclock-time-1800-INC-relative-error-max-norm-value-texpgf-scatter.csv}
			{storm.sound-rel.topological-n-ii.0.001}{{\method{II/topo}}}
			{storm.sound-rel.n-ii.0.001}{{\method{II}}}
			{true}%
            \hfill
			\numberruntimeplot{evts_csv/wallclock-time-1800-INC-relative-error-max-norm-value-texpgf-scatter.csv"}%
				{storm.exact.topological-e-sparselu.ignored,
				storm.sparse.topological-e-sparselu.ignored,
				storm.sparse-rel.topological-g-gmres.1e-06,
				storm.sound-rel.topological-n-ii.0.001}%
				{{\method{LU$^X$/topo}},
				{\method{LU/topo}},
				{\method{gmres/topo}},
				{\method{II/topo}}}%
				{$\vert S \vert$}%
				{time in seconds}{states}{13007079}
			\caption{Scatter plots showing the EVT computation runtime for standard and topological \method{II} (left) as well as the runtime of the EVT (right) approaches for different state space sizes $|S|$.}
			\label{fig:scatter_plots}
\end{figure}

\paragraph{Computing EVTs.}
The quantile plot at the top of \Cref{fig:runtime_quantile_evt} indicates the time that is required for computing the EVTs in 62 models for different approaches. 
A point at position $(x,y)$ indicates that the corresponding method solved the $x^{th}$ fastest instance in $y$ seconds, where only correctly solved instances are considered. 
The unsound methods \method{VI} and \method{gmres} produced 7 and 2 incorrect results, respectively. 
Furthermore, as errors accumulated, the topological variations of \method{VI} and \method{gmres} more frequently exceeded the threshold of $10^{-3}$. 
The variants of \method{II} always produced correct results. 

The plot indicates that (topological) \method{LU} and \method{gmres} outperform \method{VI} and \method{II} for easier instances. 
However, \method{II} catches up for the more intricate instances as it is more scalable than \method{LU} and always yields correct results. 
The exact method \method{LU$^X$} is significantly slower compared to the other methods. 
We also observe that the topological algorithms are superior to the non-topological variants. 
This is confirmed by the leftmost scatter plot in \Cref{fig:scatter_plots} which compares the topological and non-topological variant of \method{II}.
Here, each point $(x,y)$ indicates an instance for which the methods on the $x$-axis and the $y$-axis respectively required $x$ and $y$ seconds to compute the EVTs. 
Instances that contain only singleton SCCs are depicted as triangles (\mytriangle{myred}), whereas cycles (\mycircle{mygreen}) represent the remaining instances.
No incorrect results (INC) were obtained. 
The scatter plot in the middle in \Cref{fig:scatter_plots} indicates that the iteartive methods are more scalable than exact \method{LU$^X$}.
We also see that models with millions of states can be solved in reasonable time.

\paragraph{Computing stationary distributions.}
The quantile plot at the bottom of \Cref{fig:runtime_quantile_stationary} summarizes the runtimes for the different stationary distribution approaches\footnote{Six \jani{}~\cite{Bud+17} models that \method{ap-naive}, \method{ap-sample}, and \prism do not support are omitted.}.
We only consider the topological variants of the approaches of \storm{} as they were consistently faster.
No incorrect results were observed in this experiment.

The plot indicates that the guided sampling method (\method{ap-sample}) from~\cite{Meg23} and \prism perform significantly better 
than \method{ap-naive}. 
However, all algorithms provided by \storm{} \textemdash{} except for \method{Classic/LU$^x$} \textemdash{} outperform the other implementations.
For \method{LU} and \method{GMRES}, we observe that the \method{EVTreach} and \method{EVTfull} variants are significantly faster than the \method{Classic} approach. 
The \method{EVTreach} approach using \method{gmres} provides the fastest configuration, but is also the least reliable one in terms of accuracy.
For the sound methods, we observe that the \method{EVTfull} approach with \method{II} is outperformed by 
\method{LU} combined with either \method{EVTreach} or \method{EVTfull}, where the latter shows the better performance. 

\section{Related Work}\label{sec:related}

\paragraph{Computing stationary distributions.}
Other methods for the sound computation of the stationary distribution have been proposed in, e.g.,~\cite{FQ12,BF11,Bre+20}. 
In contrast to our work, they consider only subclasses of Markov chains: 
\cite{FQ12,BF11} introduce an (iterative) algorithm applicable to Markov chains with positive rows while~\cite{Bre+20} presents a technique limited to time-reversible Markov chains.
The recent approach from~\cite{Meg23} can also handle general Markov chains. 
Our technique ensures soundness with respect to both absolute and relative differences, whereas the approaches of~\cite{Meg23} only consider absolute precision. 

\paragraph{Other applications of EVTs.}
The authors of \cite{KS76} have suggested to use EVTs for the \emph{expected time to absorption}, i.e., the expected number of steps until the chain reaches an absorbing state. 
Indeed, this quantity is given by the sum of the vector $(\EXP{\vt_s})_{s \in \statesTr}$~\cite[Thm.~3.3.5]{KS76}. 
For acyclic DTMCs the EVT of a state coincides with the probability of reaching this state. 
This is relevant in the context of Bayesian networks~\cite{SK20,SK21} since inference queries in the network can be reduced to reachability queries by translating Bayesian networks into tree-like DTMCs. 
Existing procedures for multi-objective model checking of MDPs employ linear programming methods relying on the EVTs of state-action pairs~\cite{For+11,Ete+07,Cha+06,DKQR20}.  
EVTs are also employed in an algorithm proposed in~\cite{Ben+13} for LTL model checking of \emph{interval Markov chains}. 
Moreover, EVTs have been leveraged for minimizing and learning DTMCs~\cite{Bac+17,Bac+21}. 
Further recent applications of EVTs to MDPs include verifying \emph{cause-effect dependencies}~\cite{Bai+22}, as well as an abstraction-refinement procedure that measures the importance of states based on the EVTs under a fixed policy~\cite{Jun+22}. 
\cite{DBLP:journals/corr/abs-2309-01107} employs EVTs in the context of policy iteration in reward-robust MDPs.
\section{Conclusion}\label{sec:conclusion}
We elaborated on the computation of EVTs in DTMCs and CTMCs:
The EVTs in DTMCs can be determined by solving a linear equation system, while computing EVTs in CTMCs reduces to the discrete-time setting. 
We developed an iterative algorithm based on the value iteration~\cite{Put94} algorithm lacking assurance of precision. 
Building on interval iteration~\cite{HM14,Bai+17} --- an algorithm for the sound computation of reachability probabilities and expected rewards --- we developed an algorithm for approximating EVTs with accuracy guarantees. 
To enhance efficiency, we adapted a topological algorithm~\cite{DG07,Dai+11,Cie+08} to compute EVTs SCC-wise in topological order.
We showed that EVTs enable the sound approximation of the stationary distribution and the efficient computation of conditional expected rewards.
For future work, we want to extend our implementation provided in the model checker \storm{}~\cite{STORM} by symbolic computations. 
Another direction is to combine EVT-based computations with approximate verification approaches based on partially exploring relevant parts of the system~\cite{KM19,Meg23}. 
We conjecture that EVTs serve as a good heuristic to identify significant sub-regions within the state space.

\paragraph{Data availability statement.} The datasets generated and analyzed in this study and code to regenerate them are available in the accompanying artifact \cite{zenodo}.

\clearpage

\bibliographystyle{splncs04}
\bibliography{ref}

\clearpage

\appendix
\allowdisplaybreaks 
\section{Omitted Proofs}

\subsection{Proofs of \Cref{sec:EVTs}}\label{ap:EVTs}

\ETRsToEVTs*
\begin{proof}
    Let $s \in S$ be a state of $\dtmc$ and a let $\rew$ be a reward structure $\rew \colon S \to \rnonneg$ with  $\rew(t) = x \cdot \iverson{t = s}$ for some $x \in \reals_{>0}$. 
    We have $\vt_s(\pi) = \frac{1}{x} \cdot \tr_{\rew}(\pi)$ holds for all paths $\pi \in \paths{\dtmc}$.
    Thus, we obtain that 
    \begin{align*}
    \EXP{\vt_s}
       & = \int_{\paths{\dtmc}} \vt_s \,d \PR
        && \text{by \Cref{definition:EVTs}} \\ 
        &  = \int_{\paths{\dtmc}} \frac{1}{x} \cdot \tr_{\rew} \,d \PR \\
        &  = \frac{1}{x} \cdot \EXP{\tr_{\rew}}. 
    \end{align*}
\end{proof}

\EVTsCTMC*
\begin{proof}
    Let $s \in S$ be an arbitrary but fixed state in the CTMC $\ctmc$. Consider the reward structure $\rew$ with $\rew(t) = \iverson{s = t}$ for all $t \in S$. 
    It follows that for all paths $\pi \in \paths{\ctmc}$, the visiting time $\vt_s(\pi)$ is equal to the total reward $\tr_{\rew}(\pi)$ of $\pi$. 
    Thus,
    \begin{align*}
        \EXPstateMC{\vt_s}{}{\ctmc}  
        & = \int_{\paths{\ctmc}} \vt_s \,d \PRstateMC{}{\ctmc}  \\
        & = \int_{\paths{\ctmc}} \tr_{\rew} \,d \PRstateMC{}{\ctmc}\\
        & = \EXPstateMC{\tr_{\rew}}{}{\ctmc}. 
    \end{align*}
    Consider the reward structure for the DTMC $\emb(\ctmc)$ with 
    $\rew^{\emb(\ctmc)}(t)  = \frac{\rew(t)}{\rate({t})}$ 
    for all $t\in S$, or equivalently 
    \begin{align*}
        \rew^{\emb(\ctmc)}(t)  
        &= \frac{\iverson{s = t}}{\rate({t})} = \frac{\iverson{s = t}}{\rate({s})}. 
    \end{align*} 
    It holds that $\EXPstateMC{\tr_{\rew}}{}{\ctmc} = \EXPstateMC{\tr_{\rew^{\emb({\ctmc})}}}{}{\emb({\ctmc})}$ (see, e.g.,~\cite[Page 31]{Kwi+07}). 
    Consequently,
    \begin{align*}
        \EXPstateMC{\vt_s}{}{\ctmc} 
        & = \EXPstateMC{\tr_{\rew}}{}{\ctmc} \\
        & = \EXPstateMC{\tr_{\rew^{\emb(\ctmc)}}}{}{\emb(\ctmc)} \\
        & = \frac{1}{\rate(s)}\cdot \EXPstateMC{\vt_s}{}{\emb(\ctmc)} && \text{by \Cref{theorem:ETRs_to_EVTs}}. 
    \end{align*}
\end{proof}

\subsection{Proofs of \Cref{sec:computingEVTs} (Accurately Computing EVTs)}\label{ap:computing}

\subsubsection{Proofs of \Cref{sec:VI} (Value Iteration)}\label{ap:VI}

\VIConvergence*
\begin{proof}
To prove claim \Cref{theorem:VI:operator_fixpoint}, we write the linear equation system from \Cref{theorem:expected_visiting_times_LEQ} in matrix notation: 
    \[ \tauinit + \Qbf^T \cdot \xbf = \xbf  \quad \Leftrightarrow  \quad  \vioperator(\xbf) = \xbf. \]
    
    By \Cref{theorem:expected_visiting_times_LEQ}, the vector $(\EXP{\vt_{s}})_{s\in \statesTr}$ is the unique solution $\xbf$ of the equation system above.
    Thus, the unique fixed point of $\vioperator$ is $(\EXP{\vt_{s}})_{s\in \statesTr}$.  
    
    To show claim \Cref{theorem:VI:operator_limit}, we prove that $\vioperator$ is a contraction mapping on a suitable normed vector space and apply Banach's fixed point theorem (see, e.g.,~\cite[Theorem 6.2.3]{Put94}). 
    First, we describe the normed vector space. 
    The matrix $\Qbf$ satisfies $\lim_{k\to \infty} \Qbf^k = \0$ and thus the spectral radius $\rho(\Qbf)$ is strictly smaller than 1 (see, e.g.,~\cite[Theorem 5.6.12]{HC12}). 
    Thus, there is a matrix norm $\normM{\cdot}_{\Qbf}$ satisfying $\normM{\Qbf}_{\Qbf} = \rho(\Qbf) <1$ (see, e.g.,~\cite[ Lemma 5.6.10]{HC12}). 
    It follows that there is a vector norm $\norm{\cdot}_{\Qbf}$ on $\evtsdomain{\statesTr}$ such that for all $\xbf \in \evtsdomain{\statesTr}$ we have $\normQ{\Qbf  \cdot \xbf} \leq \normM{\Qbf}_{\Qbf} \cdot \normQ{\xbf}$ (see, e.g.,~\cite[Theorem 5.7.13]{HC12}). 
    As Cauchy sequences are the only sequences that converge in a finite-dimensional normed vector space, we have that $(\evtsdomain{\statesTr}, \normQ{\cdot})$ is a complete normed vector space (see, e.g.,~\cite[Theorem 5.4.10]{HC12}). 
    The contraction property of $\vioperator$ follows straightforwardly from the definition of $\normQ{\cdot}$: 
    For any $\xbf,\ybf\in \evtsdomain{\statesTr}$ we have 
        \begin{align*}
            \normQ{\vioperator(\xbf) - \vioperator(\ybf)} & = \normQ{\Qbf (\xbf - \ybf)} \\ 
            & \leq \underbrace{\normQ{\Qbf}}_{<1} 
            \cdot \normQ{\xbf - \ybf} ~. && \text{\cite[Theorem 5.7.13]{HC12}}
        \end{align*}
    We have shown that $\vioperator$ is a contraction mapping on $(\evtsdomain{\statesTr}, \normQ{\cdot})$ and thus Banach's fixed point theorem (see, e.g.,~\cite[Theorem 6.2.3]{Put94}) implies that $(\vioperator^{(k)}(\xbf))_{k \in \nat}$ converges to the unique fixed point $(\EXP{\vt_s})_{s \in \statesTr}$ of $\vioperator$ with respect to the norm $\normQ{\cdot}$.  
    Lastly, $\lim_{k\to \infty} \vioperator^{(k)}(\xbf)$ holds with respect to any norm on $\evtsdomain{\statesTr}$ (see, e.g.,~\cite[Corollary 5.4.6.]{HC12}), which concludes the proof. 
\end{proof} 
\begin{theorem}\label{theorem:VI_termination}
    For input DTMC $\dtmc$, stopping criterion $crit \in \{abs,rel\}$, and any termination threshold $\epsilon>0$, 
    \Cref{algorithm:VI} terminates. 
\end{theorem}
\begin{proof}
    We first proof termination for the absolute termination criterion. 
    The sequence $(\xbf^{(k)})_{k\in \nat}$ converges to $(\EXP{\vt_s})_{s \in \statesTr}$ in $\evtsdomain{\statesTr}$, by \Cref{theorem:VI_convergence}. 
    Thus, for $\epsilon' = \frac{\epsilon}{2} >0$ there exists $N \in \nat$ such that for all $k >N$
     \[
       \max_{s \in \statesTr} \left\{ \Bigl\vert \xbf^{(k)}(s) - \EXP{\vt_s} \Bigr\vert \right\} < \epsilon'.
    \]
    As the inequality holds for all $k>N$ we conclude that 
    \[
        \max_{s \in \statesTr}  \left\{ \Bigl\vert \xbf^{(k+1)}(s) - \EXP{\vt_s} \Bigr\vert \right\}  < \epsilon'.
    \]
    By the triangular inequality and homogeneity of the max norm, we have that 
    \begin{align*}
        \max_{s \in \statesTr}  \left\{\Bigl\vert \xbf^{(k)}(s) -  \xbf^{(k+1)}(s) \Bigr\vert \right\} 
        &\leq 
        \max_{s \in \statesTr}  \left\{\Bigl\vert \xbf^{(k)}(s) -  \EXP{\vt_s} \Bigr\vert \right\} \\ 
        & \qquad + 
        \max_{s \in \statesTr}  \left\{ \Bigl\vert \xbf^{(k+1)}(s) - \EXP{\vt_s}\Bigr\vert \right\} \\
        & \leq 2 \cdot  \epsilon' = \epsilon. 
\end{align*}
    Thus, the  absolute stopping criterion is satisfied after a finite number of iterations. 
    
    Next, we show that the relative stopping criterion is also satisfied for all $\epsilon >0$ within a finite number of iterations. 
    Recall that $\statesTr$ only contains reachable states, i.e., $(\EXP{\vt_s})_{s \in \statesTr}>0$. 
    Let $s \in \statesTr$. 
    Since $(\xbf^{(k)}(s))_{k\in \nat}$ converges to $\EXP{\vt_s}$ in $\evtsdomain{\statesTr}$, for $\epsilon' = \frac{1}{4} \cdot \EXP{\vt_s} \cdot \epsilon  >0$ there exists $N \in \nat$ such that for all $k > N$,
    \begin{align*}
        \Bigl\vert  \xbf^{(k)}(s) - \EXP{\vt_s} \Bigr\vert <  \epsilon'.
    \end{align*}
    Again, by triangular inequality and homogeneity, 
    we obtain that 
    for all $k > N$
    \begin{align*}
        \Bigl\vert  \xbf^{(k)}(s) - \xbf^{(k+1)}(s) \Bigr\vert  <  2 \cdot \epsilon' = 
        \frac{1}{2} \cdot \EXP{\vt_s} \cdot \epsilon, 
    \end{align*}
    or, equivalently 
    \begin{align}
      \label{eqdistxx}
        2 \cdot \abs*{\frac{\xbf^{(k)}(s) - \xbf^{(k+1)}(s)}{\EXP{\vt_s}}}  < \epsilon. 
    \end{align}
    Next, let 
    $\epsilon_s = \frac{1}{2}\cdot \EXP{\vt_s} > 0$.  
    By convergence of $(\xbf^{(k)}(s))_{k \in \nat}$, 
    there exists $N_s \in \nat$ such that 
    for all $k > N_s$, 
    \begin{align}
      \label{eqdistancexe}
        \Bigl\vert  \xbf^{(k)}(s) - \EXP{\vt_s} \Bigr\vert  < \epsilon_s = \frac{1}{2} \cdot \EXP{\vt_s}.
    \end{align}
    Thus, 
    \begin{align*}
         \xbf^{(k)}(s) & =  
         \Bigl\vert   \EXP{\vt_s} - \bigl(\EXP{\vt_s} - \xbf^{(k)}(s)\bigr) \Bigr\vert   && \text{since  $\xbf^{(k)}(s)\geq 0$}\\
        & \geq  \Bigl\vert  \EXP{\vt_s}  \Bigr\vert  - 
        {\Bigl\vert \EXP{\vt_s} -\xbf^{(k)}(s) \Bigr\vert } \\ 
        & > \EXP{\vt_s}  -  \frac{1}{2} \cdot \EXP{\vt_s} && \text{by \Cref{eqdistancexe}} \\
        & = \frac{1}{2} \cdot \EXP{\vt_s}.
    \end{align*}
    Consequently, for all $k > \max \{ N,(N_s)_{s \in \statesTr}\}$ it holds that $\xbf^{(k+1)}(s) \geq \frac{1}{2}\cdot{\EXP{\vt_s}} >0$ and by using \Cref{eqdistxx} we obtain  
    \begin{align*}
         \abs*{\frac{\xbf^{(k)}(s) - \xbf^{(k+1)}(s)}{\xbf^{(k+1)}(s)}}  
         \leq 2 \cdot \abs*{\frac{\xbf^{(k)}(s) - \xbf^{(k+1)}(s)}{\EXP{\vt_s}}}  
         < \epsilon. 
    \end{align*}
    By considering all states in $\statesTr$, there exists $N' \in \nat$ such that for all $k> N'$
    \begin{align*}
        \max_{s \in \statesTr} \left\{ \abs*{ \frac{\xbf^{(k)}(s) - \xbf^{(k+1)}(s)}{\xbf^{(k+1)}(s)} } \right\}  <  \epsilon. 
    \end{align*} 
    Thus, the relative stopping criterion is satisfied after a finite number of iterations.  
\end{proof}

\subsubsection{Proofs of \Cref{sec:II} (Interval Iteration)}\label{ap:II}

\IIConvergence*
\begin{proof}
    The proof is analog to the proof of~\cite[Lemma 3.3]{Bai+17}. 
    Claim \Cref{theorem:II:seqmono} follows directly from the definition of the operators $\maxvioperator$ and $\minvioperator$. 
    
    To show claim \Cref{theorem:II:seqbound}, we first prove that $\vioperator^{(k)}(\lbf) \leq \maxvioperator^{(k)}(\lbf)$ holds for all $k \in \nat$, by induction over $k \in \nat$. 
    For $k=0$, the claim holds trivially. 
    For the induction step, assume $\vioperator^{(k)}(\lbf) \leq \maxvioperator^{(k)}(\lbf)$ holds for an arbitrary but fixed $k>0$.
    By the monotonicity of $\vioperator$ follows that 
    \begin{align*}
        \vioperator\paren{\vioperator^{(k)}(\lbf)} &\leq \vioperator\paren{\maxvioperator^{(k)}(\lbf)} && \text{using $\vioperator^{(k)}(\lbf) \leq \maxvioperator^{(k)}(\lbf)$ }  \\
        &\leq 
        \max\left\{\maxvioperator^{(k)}(\lbf), \vioperator\paren{\maxvioperator^{(k)}(\lbf) }\right\} \\
        & = \maxvioperator^{(k+1)}(\lbf).
    \end{align*}
    Secondly, we prove $\maxvioperator^{(k)}(\lbf) \leq \left(\EXP{\vt_s}\right)_{s\in \statesTr}$ by induction over $k\in \nat$. 
    For the induction base $\maxvioperator^{(0)}(\lbf)=\lbf \leq \left( \EXP{\vt_s}\right)_{s \in \statesTr}$ holds by assumption. 
    For the induction step, assume $\maxvioperator^{(k)}(\lbf) \leq \left( \EXP{\vt_s}\right)_{s \in \statesTr}$.  
    Then, by the monotonicity of $\vioperator$,
    \begin{align*}
        \vioperator\paren{\maxvioperator^{(k)}(\lbf)}  & \leq \vioperator\paren{ \left(\EXP{\vt_s}\right)_{s \in \statesTr}}. 
    \end{align*}
    Since $\left(\EXP{\vt_s}\right)_{s \in \statesTr}$ is the unique fixed point of $\vioperator$, we get 
    \begin{align*}
        \maxvioperator^{(k+1)}(\lbf) = \max\left\{\maxvioperator^{(k)}(\lbf), \vioperator\paren{\maxvioperator^{(k)}(\lbf)}\right\} 
          \leq  \left(\EXP{\vt_s}\right)_{s \in \statesTr}.
    \end{align*}
    The reasoning for $\minvioperator^{(k)}(\ubf) \leq \vioperator^{(k)}(\ubf)$ works analogously. 

    In order to prove \Cref{theorem:II:seqconv}, exploit that $\vioperator^{(k)}(\lbf) \leq \maxvioperator^{(k)}(\lbf) \leq  (\EXP{\vt_s})_{s\in \statesTr}$ holds for all $k \in \nat$ by claim \Cref{theorem:II:seqbound}. 
    Additionally, $\left(\EXP{\vt_s}\right)_{s\in \statesTr} =\lim_{k\to \infty} \vioperator^{(k)}(\lbf)$ holds by \Cref{theorem:VI_convergence}. 
    Claim \Cref{theorem:II:seqconv} follows as $\left(\EXP{\vt_s}\right)_{s\in \statesTr} = \lim_{k\to \infty} \vioperator^{(k)}(\lbf) \leq \lim_{k\to \infty} \maxvioperator^{(k)}(\lbf) \leq \left(\EXP{\vt_s}\right)_{s\in \statesTr}$.
    Again, the reasoning for the sequence of upper bounds is analogous. 
\end{proof}

\IICorrectness*
\begin{proof}
    First, we show that the algorithm terminates after a finite number of iterations for any $\epsilon>0$ using the absolute and relative stopping criterion. 
    From \Cref{theorem:II_convergence} \Cref{theorem:II:seqconv}, it follows that the sequences $(\lbf^{(k)})_{k \in \nat}$ and $(\ubf^{(k)})_{k \in \nat}$ converge to $(\EXP{\vt_s})_{s \in \statesTr}$. 
    Thus, we have that for all $\epsilon>0$ there exist $N_l, N_u \in \nat$ such that for all $k_l>N_l$ and $k_u>N_u$, it holds that 
    \[
      \max_{s \in \statesTr} \left\{ \vert \lbf^{(k)}(s) - \EXP{\vt_s} \vert  \right\}< \epsilon 
        \qquad \text{and} \qquad
      \max_{s \in \statesTr} \left\{ \vert \ubf^{(k)}(s) - \EXP{\vt_s} \vert \right\} < \epsilon.
    \] 
    Thus, for all $k> \max\{N_l, N_u\}$ we have 
    \begin{align} 
      \label{eqdistul}
        \max_{s \in \statesTr} \left\{ \vert \lbf^{(k)}(s) - \ubf^{(k)}(s) \vert  \right\} < 2 \cdot  \epsilon. 
    \end{align}
    Using \Cref{eqdistul}, we can conclude that for all $k>  \max\{N_l, N_u\}$   
    Thus the absolute stopping criterion is satisfied after finitely many iterations. 

    Second, we consider the relative stopping criterion for any $\epsilon>0$. 
    Let $s\in \statesTr$. 
    Since $\statesTr$ only contains reachable states, we have that $\EXP{\vt_s}>0$. 
    By a similar reasoning as in the proof of \Cref{theorem:VI_termination}, we obtain that there exists $N_s \in \nat$ such that for all $k > N_s$
    \begin{align}
        \vert \lbf^{(k)}(s) \vert >\frac{1}{2} \cdot \EXP{\vt_s}>0. 
    \end{align}
    Analogously to the proof of \Cref{theorem:VI_termination}, it follows that for all $s\in \statesTr$ and $\epsilon>0$, there exists $N \in \nat$ such that for all $k>N$, 
    \begin{align*}
        \abs*{\frac{ \lbf^{(k)}(s) - \ubf^{(k)}(s)}{\lbf^{(k)}(s)}}  <  2 \cdot \epsilon.
    \end{align*}
    By considering all transient states, we obtain that there exists $N'\in \nat_{>0}$ such for all $k> N'$, 
    \begin{align*}
        \max_{s \in \statesTr} \left\{ \abs*{\frac{ \lbf^{(k)}(s) - \ubf^{(k)}(s)}{\lbf^{(k)}(s)}}  \right\} <  2 \cdot \epsilon. 
    \end{align*} 
    Thus, the relative stopping criterion is satisfied and the algorithm terminates after a finite number of iterations. 
    
    Next, we prove $\epsilon$-soundness of \Cref{alg:II} with respect to the absolute and relative difference. 
    The algorithm returns $\xbf^{res}= \frac{1}{2}(\lbf + \ubf)$, with $\lbf, \ubf \in \evtsdomain{\statesTr}$ such that 
    $\diff{crit}(\ubf, \lbf) \leq 2 \cdot \epsilon$ for $crit \in \{abs,rel\}$. 
    Let $\xbf^{true} = (\EXP{\vt_s})_{s \in \statesTr}$. 
    First, we consider the absolute difference $\diff{abs}$. 
    For all $s \in \statesTr$, we have 
    \begin{align*}
        \vert \xbf^{res}(s) - \xbf^{true}(s) \vert 
        & = \left| \frac{1}{2} \cdot ({\lbf(s) + \ubf(s)})- \xbf^{true}(s) \right| \\
        & =   \frac{1}{2} \cdot \vert   {\lbf(s) + \ubf(s)- 2 \cdot \xbf^{true}(s)}\vert \\
        & \leq  \frac{1}{2} \cdot  \vert {\lbf(s) + \ubf(s) - 2 \cdot \lbf(s)} \vert 
        && \text{$\0 \leq \lbf(s) \leq \xbf^{true}(s) \leq \ubf(s)$} \\ 
        & = \frac{1}{2} \cdot \underbrace{\vert (\ubf(s) - \lbf(s))\vert}_{\leq 2 \cdot \epsilon} \\
        & \leq  \epsilon && \text{termination criterion.}
    \end{align*}
    This implies that $\diff{abs}(\xbf^{res},\xbf^{true}) = \max_{s\in \statesTr} \left\{ \vert \xbf^{res}(s) - \xbf^{true}(s) \vert \right\} \leq \epsilon$. 
    Second, we prove that the returned result is $\epsilon$-sound if the relative difference is used. 
    Let $s\in \statesTr$ be a transient state. 
    The restriction of $\statesTr$ to reachable states guarantees that $\xbf^{true} = (\EXP{\vt_s})_{s \in \statesTr} > \0$. 
    This yields $\ubf_s \not =  0$ by \Cref{theorem:II_convergence} \Cref{theorem:II:seqbound}. 
    Thus, due to the termination criterion, it holds that $\lbf(s) \not = 0$. 
    Otherwise, we would have that $\diff{rel}(\ubf, \lbf) = \infty \not \leq \epsilon$.
    As a result, we obtain 
        \begin{align*}
            \abs*{\frac{\xbf^{res}(s) - \xbf^{true}(s)}{ \xbf^{true}(s)} }
            &=\abs*{\frac{   \frac{1}{2} \cdot ({\lbf(s) + \ubf(s)}) - \xbf^{true}(s)}{\xbf^{true}(s)} }\\
            & = \abs*{\frac{ \lbf(s) + \ubf(s) - 2 \cdot \xbf^{true}_s}{2 \cdot \xbf^{true}(s)}} \\
            & \leq \abs*{ \frac{\lbf(s) + \ubf(s) - 2 \cdot \lbf(s)}{2 \cdot \lbf(s)} }
            && \text{$\0 < \lbf(s) \leq \xbf^{true}(s) \leq \ubf(s)$} \\
            & =  \frac{1}{2} \cdot \underbrace{\abs*{\frac{\ubf(s) - \lbf(s) }{\lbf(s)}}}_{\leq 2 \cdot \epsilon} 
             \leq \epsilon && \text{termination criterion.}
        \end{align*}
    Consequently,
    \begin{align*}
        \diff{rel}(\xbf^{res},\xbf^{true}) \overset{\xbf^{true}>\0}{=} \max_{s \in \statesTr} \left\{ \abs*{\frac{\xbf^{res}(s) - \xbf^{true}(s)}{\xbf^{true}(s)}}  \right\} \leq \epsilon.
    \end{align*} 
\end{proof}

\subsubsection{Proofs of \Cref{sec:topological_algorithms} (Topological Algorithm)}\label{ap:topo}

To show \Cref{theorem:EVTs_in_SCC_restriction}, we provide an auxiliary lemma which is a generalization of \Cref{theorem:EVTs_absorPr}, originally presented by Kemeny and Snell in~\cite[Theorem 3.3.7]{KS76}: 
The probability $\PR\Bigl((S \setminus C )\until \{s\}\Bigr)$ of reaching a transient state $s$ in an SCC $C$, and only visiting states in $S \setminus C$ prior to reaching $s$, depends on the EVTs of the predecessors of $s$ that lie outside the SCC $C$:
\begin{lemma}\label{theorem:EVTs_constrPr}
  Let $s \in \statesTr$ be a transient state in $\dtmc$ and let $C$ be the unique SCC of $s$ in $\dtmc$. 
  Then,
  \begin{align*}
  \PR\Bigl((S \setminus C )\until \{s\}\Bigr) = \iotainit(s)+ \sum_{\substack{s' \in C'\\ C'\hookrightarrow C} } \Pbf(s',s) \cdot \EXP{\vt_{s'}}. 
  \end{align*}
\end{lemma}

\begin{proof}
The proof is similar to~\cite[Theorem 3.3.7]{KS76}, but lifted to transient states. 

Consider a transient state $s \in \statesTr$ of $\dtmc$ that is contained in SCC $C$. 
Let $S_{?} = \{s' \in C' \mid C' \in \scc{\dtmc} \wedge {C' \hookrightarrow C } \}$. 
Since each state in $C$ can reach $s$, $S_{?}$ contains exactly those states that can reach $s$ and lie outside $C$. 
Thus for all states $s'\in S \setminus (S_{?} \cup \{s\})$, 
we have $\PRstateMC{s'}{}((S \setminus C) \until \{s\}) = 0$. 

If $S_{?}=\emptyset$, we conclude that $\Pbf(s',s)= 0$ for all states $s' \in S \setminus C$. We directly obtain from~\cite[Theorem 10.15]{BK08} that 
\begin{align*}
  \PR\Bigl((S \setminus C) \until \{s\}\Bigr) 
  & = \iotainit(s) \\
  & = \iotainit(s) + \sum_{\substack{s' \in C' \\ C'\hookrightarrow C} }  
  \underbrace{\Pbf(t,s)}_{=0} \cdot \EXP{\vt_{s'}}. 
\end{align*}

We now assume $S_{?} \not = \emptyset$. 
By definition of $S_?$, we have that a state in $S_{?}$ can only be reached by other states in $S_{?}$, i.e., $\PRstateMC{t}{}(\event \{s\}_{?}) = 0$ holds for all $t \not\in S_{?}$. 
Consequently, it is sufficient to consider only states in $S_{?}$ for calculating the probability to reach a state in $S_{?}$. 
In particular, for the matrix $\Abf = (\Pbf(s', t'))_{s',t' \in S_{?}}$ containing the transition probabilities within $S_?$, it holds that $\Abf^{k}(s',t')$ equals $\Pbf^{k}(s', t')$, i.e., the probability to be in state $t'$ at the $k^{th}$ step from $s'$ for all $s', t' \in S_{?}$ and $k \in \nat$ by~\cite[Remark 10.22]{BK08}. 

Since $s$ is a transient state and all states $s'$ in $S_{?}$ can reach $s$, we conclude that $S_{?}$ does not contain recurrent states. 
Thus, $\lim_{k \to \infty} \Abf^k = \0$ by~\cite[Theorem 10.27]{BK08}. 
From~\cite[Theorem 3.2.4]{KS76}, it follows that $(\Ibf - \Abf)$ is invertible, where $(\Ibf -\Abf)^{-1} = (\EXPstateMC{\vt_{s'}}{t'}{})_{s',t'\in S_{?}}$. 

\cite[Theorem 10.15]{BK08} implies that the vector $\left(\PRstateMC{s'}{}\bigl((S \setminus C) \until \{s\}\bigr)\right)_{s' \in S_{?}}$ coincides with the vector $(\Ibf-\Abf)^{-1} \cdot (\Pbf(t', s))_{t' \in S_{?}}$.
Thus, by taking the initial distribution into account, we obtain 
\begin{align*}
  \PR\Bigl((S\setminus C ) \until \{s\} \Bigr) 
 & = \iotainit(s) + \sum_{s' \in S_?} \iotainit(s') \cdot \Pr_{s'}\left( (S\setminus C ) \until \{s\}\right) \\
 & = \iotainit(s) + \left((\iotainit(s'))_{s' \in S_?}\right)^T \cdot (\Ibf-\Abf)^{-1} \cdot (\Pbf(s', s))_{s' \in S_{?}} \\
 & = \iotainit(s) + \left((\Ibf-\Abf^T)^{-1} 
  \cdot (\iotainit(s')){t' \in S_?}  \right)^T \cdot (\Pbf(s', s))_{s' \in S_{?}} \\
 & = \iotainit(s) + (\EXP{\vt_{s'}})_{s' \in S_?} \cdot (\Pbf(s', s))_{s' \in S_{?}} ~.
\end{align*}
\end{proof}

The following definition was omitted in the main text:
\begin{restatable}[SCC restriction]{definition}{paramSCCRestriction}\label{definition:param_SCC_restriction}
    For $C \in \scc{\dtmc}$ and $\xbf\in \evtsdomain{\statesTr}$, the \emph{parametric SCC restriction $\dtmcRestr{C}[\xbf]$ of $\dtmc$ to $C$} is defined as $\dtmcRestr{C}[\xbf] = 
    \dtmcRestrTuple{\dtmcRestr{C}}{[\xbf]}$, where 
    \begin{itemize}
        \item $S^{\dtmcRestr{C}} = C$ if $C$ is a BSCC, and $S^{\dtmcRestr{C}} = C \uplus \{s_{rec}\}$ otherwise,
        \item $\Pbf^{\dtmcRestr{C}} \colon S^{\dtmcRestr{C}} \times S^{\dtmcRestr{C}} \to [0,1]$, where
        \begin{align*}
        \Pbf^{\dtmcRestr{C}}(s,t) = 
        \begin{cases}
        \Pbf^{\dtmc}(s,t) & \text{ if $s, t\in C$}, \\
        \sum_{s'  \in S\setminus C} \Pbf^{\dtmc}(s,s'), &\text{ if $s\in C$ and $t = s_{rec}$, }\\
        0 & \text{ if $s = s_{rec}$ and $t\in C$,} \\
        1 & \text{ if $s=t= s_{rec}$,} \\
        \end{cases}
        \end{align*}
        \item $\iotainit^{\dtmcRestr{C}[\xbf]} \colon S^{\dtmcRestr{C}} \to \reals_{\geq 0}$, where      \begin{align*}
        \iotainit^{\dtmcRestr{C}[\xbf]}(s) = 
        \begin{cases}
        \iotainit^{\dtmc}(s)+ \sum_{s' \in S\setminus C} \Pbf^{\dtmc}(s',s) \cdot \xbf(s') & \text{if }s \in C, \\[10pt]
        1- \sum_{t \in C} \iotainit^{\dtmcRestr{C}[\xbf]}(t)& \text{if } s = s_{rec}. 
        \end{cases} 
        \end{align*}
    \end{itemize}
    Further, we define the DTMC $\dtmcRestr{C} = \dtmcRestrTuple{\dtmcRestr{C}}{}$ as the parametric SCC restriction $\dtmcRestr{C}[\xbf]$ where $\xbf$ is chosen such that $\iotainit^{\dtmcRestr{C}[\xbf]} = \iotainit^{\dtmcRestr{C}}$ with 
    \begin{align*}
    \iotainit^{\dtmcRestr{C}}(s) = 
    \begin{cases}
    \PRstateMC{}{\dtmc}\Bigl((S \setminus C) \until  \{s\} \Bigr) & \text{ if $s\in C$,}\\
    1- \sum_{t \in C} \PRstateMC{}{\dtmc}\Bigl((S \setminus C) \until \{t\}\Bigr) & \text{ if $s=s_{rec}$.}
    \end{cases}
    \end{align*}
\end{restatable}

\EVTsInSCCRestriction*
\begin{proof}
    By \Cref{theorem:VI_convergence} \Cref{theorem:VI:operator_fixpoint}, the vector $(\EXPstateMC{\vt_s}{}{\dtmc})_{s \in \statesTr}$ is the  
    unique fixed point of the EVTs operator $\vioperator$. 
    In particular, for all $s\in \statesTr$,
    \begin{align}\label{EQ:topo_fixpoint}
    \EXPstateMC{\vt_s}{}{\dtmc} = \iotainit^{\dtmc}(s) + \sum_{t \in \statesTr} \Pbf^{\dtmc}(t, s) \cdot\EXPstateMC{\vt_t}{}{\dtmc} ~. 
    \end{align} 
    Let $C$ be a non-bottom SCC and let $\dtmcRestr{C} = \dtmcRestrTuple{\dtmcRestr{C}}{}$ be the restriction of $\dtmc$ to $C$. 
    By \Cref{EQ:topo_fixpoint} and \Cref{theorem:EVTs_constrPr}, we obtain for all states $s \in C$,
    \begin{align*}
    \EXPstateMC{\vt_s}{}{\dtmc} 
    &=  \underbrace{\iotainit^{\dtmc}(s) + \sum_{t \in S \setminus C} \Pbf^{\dtmc}(t, s) \cdot\EXPstateMC{\vt_t}{}{\dtmc}}_{\PRstateMC{}{\dtmc}\Bigl((S\setminus C) \until \{s\}\Bigr)} 
    + \sum_{t \in C} \Pbf^{\dtmc}(t, s) \cdot\EXPstateMC{\vt_t}{}{\dtmc} \\
    &= \iotainit^{\dtmcRestr{C}}(s) + \sum_{t \in C} \Pbf^{\dtmc}(t, s) \cdot\EXPstateMC{\vt_t}{}{\dtmc} \qquad  \qquad \text{by \Cref{definition:param_SCC_restriction}. }
    \end{align*}
    Thus, 
    $(\EXPstateMC{\vt_s}{}{\dtmc})_{s \in C}$ is a fixed point of the function 
    $\vioperator^{\dtmcRestr{C}} \colon \evtsdomain{\abs{C}} \to \evtsdomain{\abs{C}}$ with 
    \begin{align*}
    \vioperator^{\dtmcRestr{C}}(\xbf) = \left(\iotainit^{\dtmcRestr{C}}(s) + 
    \sum_{t \in {C}} \Pbf^{\dtmcRestr{C}}(t, s) \cdot \xbf(t)\right)_{s \in C}. 
    \end{align*}
    The function $\vioperator^{\dtmcRestr{C}}$ coincides with the EVTs-operator tailored to $\dtmcRestr{C}$ and from \Cref{theorem:VI_convergence}, it follows that 
    $(\EXPstateMC{\vt_s}{}{\dtmcRestr{C}})_{s \in C}$ is the \emph{unique} fixed point of 
    $\vioperator^{\dtmcRestr{C}}$. 
    Consequently, $\EXPstateMC{\vt_s}{}{\dtmc} = \EXPstateMC{\vt_s}{}{\dtmcRestr{C}}$.
\end{proof}

\topoBoundRel*
\begin{proof}
  Let $\epsilon \in (0,1)$ and $\xbf \in \evtsdomain{\statesTr}$ as specified in \Cref{theorem:topo_bound_rel}. 
  First, we show that for all non-bottom SCCs $C$ and $s \in C$ we have 
  \begin{align*}
    \xbf(s) \leq \EXPstateMC{\vt_s}{}{\dtmc} \cdot (1+ \epsilon)^{k+1}, 
  \end{align*}
  where $ k \in \nat$ is the length of the longest SCC chain $C_0 \hookrightarrow \dots \hookrightarrow C_k \in \chainsTr^\dtmc$ with $C_k = C$. 
  We prove the claim by induction over the chain length $k \in \nat$. 
  
  Consider an SCC $C$ with $k=0$. 
  Analogously to the proof of \Cref{theorem:topo_bound_abs}, we obtain that $\dtmcRestr{C}[\xbf]$ is equivalent to $\dtmcRestr{C}$. 
  Consequently, for every $s \in C$, 
  \begin{align*}
    \xbf(s) & \leq \EXPstateMC{\vt_s}{}{\dtmcRestr{C}[\xbf]} \cdot (1+ \epsilon)\\
    &= \EXPstateMC{\vt_s}{}{\dtmc}  \cdot (1+ \epsilon). 
  \end{align*} 

  Next, let $k \in \nat_{>0}$. 
  Let $C$ be a non-bottom SCC such that the longest SCC chain ending in $C$ has length $k$. 
  Consider a state $s \in C$. 
  By the same reasoning as in \Cref{EQ:topo_abs} in the proof of \Cref{theorem:topo_bound_abs}, we obtain 
  \begin{align*}
    \EXPstateMC{\vt_s}{}{\dtmcRestr{C}[\xbf]} & = \sum_{t \in C} \sum_{i=0}^{\infty} \Pbf^i(t,s) \cdot \Bigl( \iotainit^{\dtmc}(t) + \sum_{\substack{t' \in C'  \\ C'\hookrightarrow C}} \Pbf(t', t) \cdot \xbf(t') \Bigr). 
  \end{align*}
   By the induction hypothesis, we have 
   $\xbf(t') \leq \EXPstateMC{\vt_{t'}}{}{\dtmc} \cdot (1+ \epsilon)^{k}$ for all states $t' \in C'$ 
   in an SCC $C' \not = C$ with $C' \hookrightarrow C$. 
   This yields
   \begin{align*}
    \EXPstateMC{\vt_s}{}{\dtmcRestr{C}[\xbf]} 
    & \leq \sum_{t \in C} \sum_{i=0}^{\infty} \Pbf^i(t,s) \cdot 
    \Bigl( \iotainit^{\dtmc}(t) + \sum_{\substack{t' \in C' \\ C'\hookrightarrow C}} \Pbf(t', t) \cdot \EXPstateMC{\vt_{t'}}{}{\dtmc} \cdot (1+ \epsilon)^{k} \Bigr). 
  \end{align*}
  We have $(\EXPstateMC{\vt_s}{t}{\dtmcRestr{C}})_{t,s \in C} = (\Ibf - (\Qbf^{\dtmcRestr{C}})^T)^{-1}$ by \cite[Theorem 3.2.4]{KS76}. 
  The inverse if $\Ibf - (\Qbf^{\dtmcRestr{C}})^T$ is given by the converging Neumann series $(\Ibf - (\Qbf^{\dtmcRestr{C}})^T)^{-1} = \sum_{i=0}^{\infty} (\Qbf^{\dtmcRestr{C}})^i$ (see~\cite{Neu77}). Thus, we obtain 
  $\EXPstateMC{\vt_s}{t}{\dtmcRestr{C}} = \sum_{i=0}^{\infty} \Pbf^i(t,s)$. 
  Continuing the reasoning above and applying~\cite[Lemma 3.5.]{Bai+17}, results in 
  \begin{align*}
    \EXPstateMC{\vt_s}{}{\dtmcRestr{C}[\xbf]} 
    & \leq \sum_{t \in C} \EXPstateMC{\vt_s}{t}{\dtmcRestr{C}} \cdot 
    \Bigl( \iotainit^{\dtmc}(t) + 
    \sum_{\substack{t' \in C' \\ C'\hookrightarrow C}} \Pbf(t', t) \cdot \EXPstateMC{\vt_{t'}}{}{\dtmc} \cdot (1+ \epsilon)^{k} \Bigr) \\
    &=  \sum_{t \in C} \frac{ \PRstateMC{t}{\dtmcRestr{C}}(\event \{s\}) }{1-\PRstateMC{s}{\dtmcRestr{C}}(\nextLTL \event \{s\})}   \\
    &\phantom{={}} \cdot \Bigl( \iotainit^{\dtmc}(t) + 
      \sum_{\substack{t' \in C' \\ C'\hookrightarrow C}} \Pbf(t', t) \cdot \EXPstateMC{\vt_{t'}}{}{\dtmc} \cdot (1+ \epsilon)^{k} \Bigr)  \\
  \end{align*}%
 Since $\epsilon$ is larger than $0$, we infer $(1+\epsilon)^{k} \geq 1$. 
 Together with \Cref{theorem:EVTs_constrPr}, we have 
 \begin{align}\label{EQ:topo_rel} 
  \EXPstateMC{\vt_s}{}{\dtmcRestr{C}[\xbf]}  
  & \leq \sum_{t \in C} \frac{ \PRstateMC{t}{\dtmcRestr{C}}(\event \{s\}) }{1-\PRstateMC{s}{\dtmcRestr{C}}(\nextLTL \event \{s\})}  \\
  &\phantom{={}} 
    \cdot \Bigl( \iotainit^{\dtmc}(t) + \sum_{\substack{t' \in C' \\ C' \hookrightarrow C}} \Pbf(t',t) 
    \cdot  \EXPstateMC{\vt_{t'}}{}{\dtmc} \Bigr) \cdot (1+\epsilon)^{k} \\
  & 
  = \sum_{t \in C}  
  \frac{ \PRstateMC{t}{\dtmcRestr{C}}(\event \{s\}) }{1-\PRstateMC{s}{\dtmcRestr{C}}(\nextLTL \event \{s\})} \cdot \PRstateMC{}{\dtmc}\Bigl((S \setminus C) \until \{t\} \Bigr) \cdot (1+\epsilon)^{k} \nonumber
\end{align}
By applying the definition of the initial distribution $\iotainit^{\dtmcRestr{C}}$ of the SCC restriction $\dtmcRestr{C}$ (see \Cref{definition:param_SCC_restriction}), we obtain that 
\begin{align*}
  \EXPstateMC{\vt_s}{}{\dtmcRestr{C}[\xbf]} & 
  \leq 
  \sum_{t \in C} \frac{ \PRstateMC{t}{\dtmcRestr{C}}(\event \{s\})}{1-\PRstateMC{s}{\dtmcRestr{C}}(\nextLTL \event \{s\})} \cdot  \iotainit^{\dtmcRestr{C}}(t) \cdot (1+\epsilon)^{k} \\
  &= \frac{ \PRstateMC{}{\dtmcRestr{C}}(\event \{s\}) }{1-\PRstateMC{s}{\dtmcRestr{C}}(\nextLTL \event \{s\})} \cdot (1+\epsilon)^{k} && \text{by~\cite[Theorem 10.15]{BK08}} \\
  &= \EXPstateMC{\vt_s}{}{\dtmcRestr{C}} \cdot (1+\epsilon)^{k} && \text{by~\cite[Lemma 3.5.]{Bai+17}.} 
\end{align*} 
  Since the relative difference between $(\xbf(s))_{s\in C}$ and 
  $(\EXPstateMC{\vt_s}{}{\dtmcRestr{C}[\xbf]})_{s \in C}$ is bounded by $\epsilon$, we have 
  \begin{align*}
    \xbf(s) &\leq \EXPstateMC{\vt_s}{}{\dtmcRestr{C}[\xbf]} \cdot (1+\epsilon)  \\
    & \leq \EXPstateMC{\vt_s}{}{\dtmcRestr{C}} \cdot (1+\epsilon)^{k} \cdot (1+\epsilon) \\
    & \leq \EXPstateMC{\vt_s}{}{\dtmc} \cdot (1+\epsilon)^{k+1} && \text{by \Cref{theorem:EVTs_in_SCC_restriction},}
  \end{align*}
  which proves that the claim holds for all $k \in \nat$. 
 
  Next, let $s \in \statesTr$ and let $C$ be the SCC containing $s$. 
  $L$ is the length of the longest non-bottom SCC chain in $\dtmc$. 
  In particular, this implies that the longest SCC chain ending in $C$ has length at most $L$. 
  Thus, 
  \begin{align*} 
    & & \xbf(s) \leq  \EXPstateMC{\vt_s}{}{\dtmc} \cdot(1+\epsilon)^{L+1} &= \EXPstateMC{\vt_s}{}{\dtmc} + \EXPstateMC{\vt_s}{}{\dtmc} \cdot \left((1+\epsilon)^{L+1} -1 \right) \\
    & \Leftrightarrow &   \xbf(s) -   \EXPstateMC{\vt_s}{}{\dtmc}  &\leq \EXPstateMC{\vt_s}{}{\dtmc} \cdot \left((1+\epsilon)^{L+1} -1 \right) \\
    & \Leftrightarrow & \frac{\xbf(s)- \EXPstateMC{\vt_s}{}{\dtmc}}{\EXPstateMC{\vt_s}{}{\dtmc}} & \leq \left(1+\epsilon\right)^{L+1} -1 \tag{since $\EXPstateMC{\vt_s}{}{\dtmc}>0$.} 
  \end{align*}  
  We have that $\left(1+\epsilon\right)^{L+1} -1>0$ because $\epsilon>0$. 
  Thus, if $x(s) > \EXPstateMC{\vt_s}{}{\dtmc}$,
\begin{align}\label{upperboundx}
  \abs*{\frac{\xbf(s)- \EXPstateMC{\vt_s}{}{\dtmc}}{\EXPstateMC{\vt_s}{}{\dtmc}}} & \leq \left(1+\epsilon\right)^{L+1} -1, 
\end{align}
which yields an upper bound on the 
relative difference in case 
$x(s) > \EXPstateMC{\vt_s}{}{\dtmc}$. 

Next, consider $\xbf(s) < \EXPstateMC{\vt_s}{}{\dtmc}$. 
We derive a lower bound of $\xbf$. The reasoning works analogously as above, except that we exploit the fact that $(1-\epsilon)<1$ holds by choice of $\epsilon <1$ to establish the analog to \Cref{EQ:topo_rel}. 
As a result, for all states $s\in \statesTr$, 
\begin{align*} 
  & &  \xbf(s) \geq \EXPstateMC{\vt_s}{}{\dtmc} \cdot  (1-\epsilon)^{L+1} &= \EXPstateMC{\vt_s}{}{\dtmc} +  \EXPstateMC{\vt_s}{}{\dtmc} \cdot \left((1-\epsilon)^{L+1} -1 \right) \\
  & \Leftrightarrow & 
  {\xbf(s)- \EXPstateMC{\vt_s}{}{\dtmc}} & \geq \EXPstateMC{\vt_s}{}{\dtmc} 
  \cdot \left((1-\epsilon)^{L+1} -1\right) \\
  & \Leftrightarrow & \frac{\xbf(s)- \EXPstateMC{\vt_s}{}{\dtmc}}{\EXPstateMC{\vt_s}{}{\dtmc}} & \geq \left(1-\epsilon\right)^{L+1} -1 \tag{since $\EXPstateMC{\vt_s}{}{\dtmc}>0$.} 
\end{align*} 
We have $\left(1-\epsilon\right)^{L+1} -1 <0$ because $\epsilon<1$. 
For $\xbf(s) < \EXP{\vt_s}$, we conclude
\begin{align}\label{lowerboundx}
  \begin{split}
  \abs*{\frac{\xbf(s)- \EXPstateMC{\vt_s}{}{\dtmc}}{\EXPstateMC{\vt_s}{}{\dtmc}}} & \leq - \left( \left(1-\epsilon\right)^{L+1} -1 \right)\\
  & = 1- \left(1-\epsilon\right)^{L+1}. 
  \end{split}
\end{align} 

By combining \Cref{upperboundx,lowerboundx}, we infer that 
\begin{align*} 
  \abs*{\frac{\xbf(s)- \EXPstateMC{\vt_s}{}{\dtmc}}{\EXPstateMC{\vt_s}{}{\dtmc}}} \leq 
   \begin{cases}
    \left(1+\epsilon \right)^{L+1} -1 & \text{if } \xbf(s) > \EXPstateMC{\vt_s}{}{\dtmc}, \\
    1- \left(1-\epsilon\right)^{L+1} & \text{if } \xbf(s) < \EXPstateMC{\vt_s}{}{\dtmc},\\
    0 & \text{otherwise.}
   \end{cases}
 \end{align*} 
 This implies that 
  \begin{align*} 
    \diff{rel} \left( \xbf, (\EXPstateMC{\vt_s}{}{\dtmc})_{s \in \statesTr} \right) 
    &= \max_{s \in \statesTr} 
    \left\{ \abs*{ \frac{\xbf(s) - \EXPstateMC{\vt_s}{}{\dtmc}}{\EXPstateMC{\vt_s}{}{\dtmc}} } \right\}  \\
    & \leq \max\left\{1-(1-\epsilon)^{L+1}, (1+\epsilon)^{L+1} -1\right\}.  
  \end{align*} 

  Lastly, we show that  
  $1-(1-\epsilon)^{L+1} \leq (1+\epsilon)^{L+1} -1$ holds for all 
  $\epsilon \in (0,1)$ and $L \in \nat$ by induction on $L$. 

  For $L=0$ we have 
  \begin{align*}
    1-(1-\epsilon)^{1} = \epsilon = (1+\epsilon)^1 -1. 
  \end{align*}
  By assuming that the claim holds for fixed $L \in \nat$, we get 
  \begin{align*}
    1-(1-\epsilon)^{L+2} 
    &=  \left(1-(1-\epsilon)^{L+1} \right) \cdot (1-\epsilon) + \epsilon \\
    &\leq \left((1+\epsilon)^{L+1} -1 \right) \cdot (1-\epsilon) + \epsilon && \text{by the induction hypothesis} \\
    & \leq \left((1+\epsilon)^{L+1} -1 \right) \cdot (1+\epsilon) + \epsilon \\
    & =  (1+\epsilon)^{L+2}  - (1+\epsilon) + \epsilon \\
    & =  (1+\epsilon)^{L+2} - 1, 
  \end{align*}
  which concludes the proof of \Cref{theorem:topo_bound_rel}. 
\end{proof} 

\topoCorrectnessRel*
\begin{proof}
  First, termination of the algorithm is obvious assuming that the procedure used to approximate the EVTs (\Cref{alg:topo_exact:line:compute}) terminates.
  
  Second, we show $\epsilon$-soundness regarding the relative difference.  
  We have $\sigma  =  \sqrt[L+1]{1+\epsilon}-1$, where $L$ is the length of the longest SCC chain in $\chainsTr{\dtmc}$. 
  By assumption, $\epsilon$ is in $(0,1)$ which guarantees that $\sigma \in (0,1)$. 
  Consequently, we can apply \Cref{theorem:topo_bound_rel}, which yields 
  \begin{align}
    \diff{rel}\left(\xbf^{res}, (\EXPstateMC{\vt_s}{}{\dtmc})_{s \in \statesTr} \right) \leq   (1+\sigma)^{L+1} -1 . 
  \end{align}
  By substituting $\sigma  =  \sqrt[L+1]{1+\epsilon}-1$, we obtain that 
    \begin{align*}
      \diff{rel}\left(\xbf^{res}, (\EXPstateMC{\vt_s}{}{\dtmc})_{s \in \statesTr} \right) 
      & \leq (1+\sigma)^{L+1} -1 \\
      & = \left(1+ \sqrt[L+1]{1+\epsilon}-1\right)^{L+1} -1 \\
      &= \epsilon,
    \end{align*}
  which shows that the vector $\xbf^{res}$ is an $\epsilon$-sound approximation of $(\EXPstateMC{\vt_s}{}{\dtmc})_{s \in \statesTr}$. 
\end{proof}

\subsection{Proofs of \Cref{sec:stationary} (Stationary Distributions)}
\label{ap:stationary}

We establish some basics concerning the propagation of relative errors for (non-negative) values.
Claims made about the error propagations in \Cref{sec:stationary} follow straightforwardly.
Let $x \in \rnonneg$ be an approximation of $x^* \in \rnonneg \setminus \{0\}$ with relative precision $\epsilon \in (0,1)$, i.e., $\diff{rel}(x,x^*) = |\delta_x| \le \epsilon$, where $\delta_x \coloneq \frac{x-x^*}{x^*}$.
It holds that $x = x^* \cdot (1+\delta_x)$.
Let $y$, $y^*$, and $\delta_y$ be similar.

We first argue that the sum $x+y$ is an approximation of $x^* + y^*$ with relative precision $\epsilon$.
Using $-\epsilon \le \delta_x,\delta_y \le \epsilon$ and $x+y = x^* \cdot (1+\delta_x) + y^* \cdot (1+\delta_y)$ we get
\begin{align*}
 (x^* + y^*) \cdot (1-\epsilon) ~\le~x+y ~\le~ (x^* + y^*) \cdot (1+\epsilon).
\end{align*}
It follows that there is a $\delta^+ \in [-\epsilon,\epsilon]$ with $x+y = (x^* + y^*) \cdot (1+\delta^+)$, i.e.,
\[
\diff{rel}(x+y,~x^*+y^*) ~=~ |\delta^+| \le \epsilon.
\]
Next, we consider multiplication. We have
\begin{align*}
x\cdot y
~=~ x^* \cdot (1+\delta_x) \cdot y^* \cdot (1+\delta_y)
~=~ x^* \cdot y^* \cdot (1+\delta_x + \delta_y + \delta_x \cdot \delta_y).
\end{align*}
With $\delta^\bullet \coloneq \delta_x + \delta_y + \delta_x \cdot \delta_y$ it follows that
\[
\diff{rel}(x\cdot y,~x^* \cdot y^*) ~=~ |\delta^\bullet | ~\le~  2 \epsilon + \epsilon^2.
\]
For division, we observe that
\[
\frac{x}{y} ~=~ \frac{x^*}{y^*} \cdot \frac{1+\delta_x}{1+\delta_y} ~=~ \frac{x^*}{y^*} \cdot \frac{1+\delta_x+\delta_y - \delta_y}{1+\delta_y} = \frac{x^*}{y^*} \cdot \left(1+ \frac{\delta_x- \delta_y}{1+\delta_y}\right).
\]
With $\delta^{/} \coloneq \frac{\delta_x- \delta_y}{1+\delta_y}$ we get
\[
\diff{rel}\left(\frac{x}{y},~\frac{x^*}{y^*}\right) ~=~ |\delta^{/} | ~\le~  \frac{2 \epsilon}{1-\epsilon}.
\]

\stationaryViaEVTs*
\begin{proof}
  See, e.g.,~\cite[Theorem 2.2.3]{Vol05} for a proof of the first part in the context of renewal theory.
  Herein, we present an alternative proof of the statement. 
  The equation system for computing the EVTs in ${\dtmcBv{}}$ is: 
  \begin{align*}
      \text{For all } s \in S \setminus \{v\}: \xbf(s) &= \sum_{r \in S} \Pbf(r,s) \cdot \xbf(r) \\
      \xbf(v) &= 1. 
  \end{align*}
  Since $\{\hat{v}\}$ is the only BSCC and thus is reached with probability 1, the unique solution $(\EXPstateMC{\vt_s}{}{{\dtmcBv{}}})_{s \in S}$ to this equation system satisfies 
  \begin{align*}
      \EXPstateMC{\vt_v}{}{{\dtmcBv{}}} &=1 \\
      & = \PRstateMC{}{{\dtmcBv{}}}(\event \{\hat{v}\}) \\
      & = \sum_{r \in S} \hat{\mathbf{P}}(r,\hat{v}) \cdot \EXPstateMC{\vt_r}{}{{\dtmcBv{}}} && \text{by \Cref{theorem:EVTs_absorPr}} \\
      & = \sum_{r \in S} \mathbf{P}(r,v) \cdot \EXPstateMC{\vt_r}{}{{\dtmcBv{}}}. 
  \end{align*}
  We have thus shown that the equality 
  \begin{align*}
      \EXPstateMC{\vt_s}{}{{\dtmcBv{}}}
      =
      \sum_{r \in S} \mathbf{P}(r,s) \cdot \EXPstateMC{\vt_r}{}{{\dtmcBv{}}}
  \end{align*}
  holds for \emph{all} $s \in S$, i.e., the expected visiting times $(\EXPstateMC{\vt_s}{}{{\dtmcBv{}}})_{s \in S}$ constitute a left eigenvector of the probability transition matrix $\Pbf$ of the original irreducible DTMB $\mathcal{B}$.
  Normalising the values so that they sum up to 1 yields: 
  \begin{align*}
      \frac{\EXPstateMC{\vt_s}{}{{\dtmcBv{}}}}{\sum_{t \in S} \EXPstateMC{\vt_t}{}{{\dtmcBv{}}}} = \sum_{r \in S} \mathbf{P}(r,s) \cdot \frac{\xbf(r)}{\sum_{t \in S} \EXPstateMC{\vt_t}{}{{\dtmcBv{}}}}. 
  \end{align*}
  Since $\mathcal{B}$ is irreducible, \cite[Thm.~4.18]{Kul20} implies that the 
  steady state distribution $\stst{\mathcal{B}}$ is uniquely defined by 
  $\stst{\mathcal{B}}(s) = \frac{\EXPstateMC{\vt_s}{}{{\dtmcBv{}}}}{\sum_{t \in S} \EXPstateMC{\vt_t}{}{{\dtmcBv{}}}}$.

The second claim of the proof follows directly from the propagation of relative errors from the beginning of this section.
\end{proof}

\subsection{Proofs of \Cref{sec:conditional_rew} (Conditional Expected Rewards)}
\label{ap:conditional_reward}

\condExpRew*

\begin{proof}
  Throughout the proof we use matrix-vector notation for the sake of brevity and assume that $\dtmc = \dtmcTuple{}$ is a DTMC with
  $\Pbf = \big(\begin{smallmatrix}
  \Qbf & \Rbf \\
  \0 & \Ibf\\
  \end{smallmatrix}\big)$,
  i.e., $\Qbf$ is the transition matrix from transient to transient states, $\Rbf$ is the transition matrix from transient to recurrent states.
  Furthermore, we assume that all recurrent states are absorbing (this is why there is an identity matrix in the lower right corner of $\Pbf$).

  First assume that $\rew(\statesTr) = 1$, i.e., all transient states have unit reward.
  Let $r \in \statesRe$.
  We can rewrite the integral in the definition of $\CEXPstateMCcond{\tr_\rew}{}{\dtmc}{\event \{r\}}$ as follows:
  \begin{align*}
      & \int_{\event \{r\}} \tr_{\rew} \,d \PRstateMC{}{\dtmc}  \\
      =~& \sum_{k \geq 0} \int_{S \setminus \{r\} \until^{=k} \{r\}} \tr_{\rew} \,d \PRstateMC{}{\dtmc} && \text{(because $\event \{r\} = \biguplus_{k \geq 0} S \setminus \{r\} \until^{=k} \{r\}$)}\\
      =~& \sum_{k \geq 0} k \cdot \PRstateMC{}{\dtmc}(S \setminus \{r\} \until^{=k} \{r\}) && \text{(by the assumptions about $\rew$)}
      ~.
  \end{align*}
  Recall that the transition matrix of $\dtmc$ is $\Pbf = \big(\begin{smallmatrix}
  \Qbf & \Rbf \\
  \0 & \Ibf\\
  \end{smallmatrix}\big)$.
  Define
  $\tilde{\Pbf} = \big(\begin{smallmatrix}
  \Qbf & \Rbf \\
  \0 & \0 \\
  \end{smallmatrix}\big)$
  and note that for all inegers $k \geq 1$ we have
  $\tilde{\Pbf}^k = \big(\begin{smallmatrix}
  \Qbf^k & \Qbf^{k-1}\Rbf \\
  \0 & \0 \\
  \end{smallmatrix}\big)$.
  Further, it can be shown that
  \begin{align*}
      \PRstateMC{}{\dtmc}(S \setminus \{r\} \until^{=k} \{r\})
      =
      (\iotainit \cdot \tilde{\Pbf}^k)_r
      ~.
  \end{align*}
  Assume that $\iotainit$ is the row vector $\iotainit = \begin{pmatrix}\tauinit & \rhoinit \end{pmatrix}$ (i.e., $\tauinit$ and $\rhoinit$ are the initial probabilities of the transient and the recurrent states, respectively).
  Then
  \begin{align*}
      & \int_{\event \{r\}} \tr_{\rew} \,d \PRstateMC{}{\dtmc} \\
      =~& \sum_{k \geq 0} k \cdot (\iotainit \cdot \tilde{\Pbf}^k)_r && \text{(see above)}\\
      =~& \sum_{k \geq 1} k \cdot (\tauinit \cdot \Qbf^{k-1} \cdot \Rbf)_r && \text{(equation for $\tilde{\Pbf}^k$ above)} \\
      =~& \sum_{k \geq 1} k \cdot \sum_{t \in \statesTr} (\tauinit \cdot \Qbf^{k-1})_t \cdot \Rbf_{t,r} && \text{(rewriting)} \\
      =~& \sum_{t \in \statesTr} \Big( \sum_{k \geq 1} k \cdot (\tauinit \cdot \Qbf^{k-1})_t  \Big) \cdot \Rbf_{t,r} && \text{(rewriting)} \tag{$\star$}\label{eq:condexprewproof1}
  \end{align*}
  Now let $t \in \statesTr$.
  The next goal is to derive an equation system for the terms $\sum_{k \geq 1} k \cdot (\tauinit \cdot \Qbf^{k-1})_t$:
  \begin{align*}
      &\quad \sum_{k \geq 1} k \cdot (\tauinit \cdot \Qbf^{k-1})_t \\
      &= \sum_{k \geq 1} (\tauinit \cdot \Qbf^{k-1})_t + \sum_{k \geq 1} (k-1) \cdot (\tauinit \cdot \Qbf^{k-1})_t && \text{(rewriting)} \\
      &= \sum_{k \geq 0} (\tauinit \cdot \Qbf^{k})_t + \sum_{k \geq 0} k \cdot (\tauinit \cdot \Qbf^{k})_t && \text{(index shifts)} \\
      &= \Big( \tauinit \cdot  \sum_{k \geq 0} \Qbf^{k} \Big)_t + \sum_{k \geq 0} k \cdot \sum_{t' \in \statesTr} (\tauinit \cdot \Qbf^{k-1})_{t'} \cdot \Qbf_{t',t} && \text{(rewriting)}\\
      &= \EXPstateMC{\vt_t}{}{\dtmc} + \sum_{t' \in \statesTr} \Big( \sum_{k \geq 1} k \cdot  (\tauinit \cdot \Qbf^{k-1})_{t'} \Big) \cdot \Qbf_{t',t} && \text{(by \Cref{theorem:expected_visiting_times_LEQ}; rewriting)}
  \end{align*}
  In matrix-vector notation, the above equations can be written as
  \begin{align*}
      &\ybf = \EXPstateMC{\mathbf {vt}}{}{\dtmc} + \ybf \cdot \Qbf \\
      \iff\quad & \ybf \cdot (\Ibf - \Qbf) = \EXPstateMC{\mathbf{vt}}{}{\dtmc}
  \end{align*}
  where $\ybf = \Big(\sum_{k \geq 1} k \cdot (\tauinit \cdot \Qbf^{k-1})_t \Big)_{t \in \statesTr}$ and $\EXPstateMC{\mathbf{vt}}{}{\dtmc} = (\EXPstateMC{\vt_t}{}{\dtmc})_{t \in \statesTr}$ are interpreted as row vectors.
  Since $\Ibf - \Qbf$ is invertible~\cite{KS76}, $\ybf$ is the unique solution of this equation system.
  \Cref{theorem:cond_exp_rew} with $\rew(\statesTr) = \{1\}$ now follows with~\eqref{eq:condexprewproof1} (numerator) and \Cref{theorem:EVTs_absorPr} (denominator).
  
  Now we discuss the case where the reward function $\rew$ assigns arbitrary non-negative rational rewards to the transient states $\statesTr$.
  
  First assume that $\rew$ takes positive integer values for the transient states, i.e., $\rew(\statesTr) \subseteq \nat_{\geq 1}$.
  We proceed by induction on $\max_{s \in \statesTr}\rew(s)$.
  The base case where all states have reward $1$ was already treated above.
  For the inductive case let $\max_{s \in \statesTr}\rew(s) = n > 1$ and consider all $s \in \statesTr$ with $\rew(s) = n$.
  We construct new a DTMC $\dtmc'$ with additional states $s'$ for each such $s$ and reward function $\rew'$ such that $\rew'(s') = 1$, $\rew'(s) = n-1$ by redirecting all incoming transitions of $s$ to $s'$ and adding a single transition with probability $1$ from $s'$ to $s$.
  Intuitively, we simulate the reward $n$ of $s$ by inserting a unique predecessor with reward $1$ right before $s$, and decreasing the reward of $s$ to $n-1$.
  It is clear that conditional expected rewards in $\dtmc$ and $\dtmc'$ are the same.
  Now we can apply the induction hypothesis to $\dtmc'$ (the induction hypothesis is just the statement of \Cref{theorem:cond_exp_rew}), i.e., 
  \begin{align*}
      \ybf(s) &= (n-1) \cdot \EXPstateMC{\vt_s}{}{\dtmc'} + \ybf(s') \\
      \ybf(s') &= \EXPstateMC{\vt_{s'}}{}{\dtmc'} + \sum_{t \in \statesTr} \Pbf(t,s') \cdot \ybf(t)
  \end{align*}
  But since $s$ is a direct successor of $s'$, $\EXPstateMC{\vt_{s'}}{}{\dtmc'} = \EXPstateMC{\vt_{s}}{}{\dtmc'}$.
  Combining the above two equations yields $\ybf(s) = n \cdot \EXPstateMC{\vt_s}{}{\dtmc} + \sum_{t \in \statesTr} \Pbf(t,s) \cdot \ybf(t)$ as desired.
 
  Now assume that $\rew(\statesTr) \subseteq \nat$ (including 0).
  We can make a similar construction as above for transient states $s \in \statesTr$ with $\rew(s) = 0$.
  We ``eliminate'' (see, e.g.,~\cite{DBLP:journals/iandc/BaierHHJKK20}) each such $s$ in $\dtmc$ to obtain a modified DTMC $\dtmc'$ with fewer states where each state has reward in $\nat_{\geq 1}$.
  This is achieved by first scaling the outgoing transition of $s$ by the factor $\frac{1}{1-\Pbf(s,s)}$, and then ``gluing'' each pair if incoming and outgoing transition of $s$ together by multiplying their probabilities.
  The initial probability mass of $\iotainit(s)$ is distributed to the successors of $s$.
  Then $s$ is not reachable anymore and can be removed from $\dtmc'$.
  However, it can be shown that the effect of this elimination procedure on the transitions of $\dtmc'$ is really equivalent to simply adding the equation
  \begin{align*}
      \ybf(s) = \sum_{t \in \statesTr} \Pbf(t,s) \cdot \ybf(t) = 0 \cdot \EXPstateMC{\vt_s}{}{\dtmc} + \sum_{t \in \statesTr} \Pbf(t,s) \cdot \ybf(t)
  \end{align*}
  to the system; indeed, eliminating the variable $\ybf(s)$ with Gaussian elimination is the same as altering the transitions as described above~\cite{DBLP:journals/iandc/BaierHHJKK20}.
  
  Finally, if $\rew(\statesTr) \subseteq \Rationals_{\geq 0}$, we can consider $\rew'\colon S \to \nat$ where $\rew'(s) = M \cdot \rew(s)$ for all $s \in S$, $M$ the greatest common integer multiple of the denominators of the rational numbers occurring in $\rew$.
  Let the vector $\ybf'$ be the unique solution of
  \begin{align*}
      \ybf' = (\rew'(s) \cdot \EXPstateMC{\vt_s}{}{\dtmc})_{s \in \statesTr} + \ybf' \cdot \Qbf
      ~,
  \end{align*}
  i.e., $\ybf' = (\rew'(s) \cdot \EXPstateMC{\vt_s}{}{\dtmc})_{s \in \statesTr} \cdot (\Ibf - \Qbf)^{-1} $.
  Then
  \begin{align*}
      \frac{1}{M} \ybf' = (\rew(s) \cdot \EXPstateMC{\vt_s}{}{\dtmc})_{s \in \statesTr} \cdot (\Ibf - \Qbf)^{-1}= \ybf
      \tag{$\dagger$}\label{eq:condexprewproof2}
      ~.
  \end{align*}
  Now
  \begin{align*}
      &\CEXPstateMCcond{\tr_\rew}{}{\dtmc}{\event \{r\}} \\
      =~& \frac{1}{M} \cdot \CEXPstateMCcond{\tr_{\rew'}}{}{\dtmc}{\event \{r\}} && \text{(linearity of expected value)} \\
      =~& \frac{1}{M} \cdot \frac{\sum_{t \in \statesTr} \Pbf(t,r) \cdot \ybf'(t)}{\iotainit(r) + \sum_{t \in \statesTr} \Pbf(t,r) \cdot \EXPstateMC{\vt_t}{}{\dtmc}} && \text{(\Cref{theorem:cond_exp_rew} for integer rewards)} \\
      =~& \frac{\sum_{t \in \statesTr} \Pbf(t,r) \cdot \ybf(t)}{\iotainit(r) + \sum_{t \in \statesTr} \Pbf(t,r) \cdot \EXPstateMC{\vt_t}{}{\dtmc}} && \text{by~\eqref{eq:condexprewproof2}}
  \end{align*}
  which concludes the proof.
\end{proof}

\begin{remark}
    As an (immediate) consequence of \Cref{theorem:cond_exp_rew}, if $\lbf \leq \EXPstateMC{\vt}{}{\dtmc} \leq \ubf$, and $\ybf_{\lbf}$ and $\ybf_{\ubf}$ are the unique solutions of $(\Ibf - \Qbf^T) \cdot \ybf_{\lbf} = \lbf$ and $(\Ibf - \Qbf^T) \cdot \ybf_{\ubf} = \ubf$, respectively, then
    \begin{align*}
    \frac{(\Rbf^T \cdot \ybf_{\lbf})_r}{\iotainit(r) + (\Rbf^T \cdot \ubf)_r}
    \leq
    \CEXPstateMCcond{\tr_\rew}{}{\dtmc}{\event \{r\}}
    \leq
    \frac{(\Rbf^T \cdot \ybf_{\ubf})_r}{\iotainit(r) + (\Rbf^T \cdot \lbf)_r}
    ~.
    \end{align*}
\end{remark}
\newpage
\section{Value Iteration is not $\epsilon$-Sound for EVTs}\label{ap:VI_unsound}

We show that the absolute or relative difference between the returned result and the true EVTs cannot be bounded by a predefined $\epsilon>0$, in general. 

\begin{figure}[t]
    \centering
    \begin{tikzpicture}[>=stealth', auto, semithick, node distance=3cm]
        \node[] (init)  {};
        \node[state] (q0) [right=1cm of init]  {$s_1$};
        \node[state] (q1) [right=3cm of q0]  {$s_2$};
        \node[state] (q2) [right=3cm of q1] {$s_{3}$};
    
        \path[->] (init)  edge node {1} (q0); 
        \path[->]  (q0)  edge   [bend left=20] node {$0.1$} (q1);
        \path[->]  (q1)  edge   [bend left=20] node {$1-p$} (q2);
        \draw [->] (q1) edge[loop above] node  {$p$} (q1);
        \draw [->] (q2) edge[loop above] node {1} (q2);
        \path[->]  (q0)  edge   [bend right=20] node  [swap] {$0.9$} (q2);
  \end{tikzpicture}
  \caption{A DTMC for which VI does not provide $\epsilon$-sound results.}\label{figure:VI_unsound}
\end{figure}
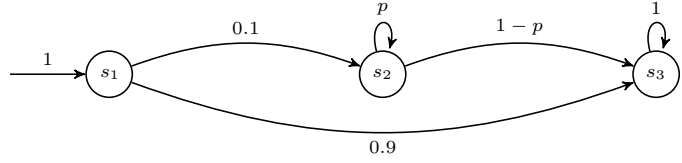

\begin{example}\label{example:VI_unsound}
    Assume that we use the standard VI algorithm with $\xbf^{(0)}= \0$ and $\epsilon = 0.1$ to compute the EVTs for the states $s_1$ and $s_2$ of the DTMC depicted in \Cref{figure:VI_unsound}, where $p \in [0,1)$. 
    It can be verified that the EVT of state $s_1$ equals $1$, while the EVT of state $s_2$ equals $0.1 \cdot \frac{1}{1-p}$. 
    The algorithm computes 
    $\xbf^{(1)} = 
        \begin{pmatrix}
            1,
            0 
        \end{pmatrix}^T$ 
    and 
    $ \xbf^{(2)} = \begin{pmatrix}
        1,
        0.1 
    \end{pmatrix}^T$ 
    in the first and second iteration. 
    Thus, the absolute difference between the vectors $\xbf^{(1)}$ and $\xbf^{(2)}$ equals 0.1 and the algorithm terminates. 
    Hence, the absolute difference between the returned vector and the true result equals 
    \[
    \diff{abs}\left(
        \begin{pmatrix}
        1 \\
        0.1 
    \end{pmatrix},  \begin{pmatrix}
        1 \\
        0.1 \cdot \frac{1}{1-p}
    \end{pmatrix}\right) = \abs*{ 0.1 - 0.1 \cdot \frac{1}{1-p}}.
    \] 
    In particular, the difference cannot be bounded by any $\epsilon>0$ since $\vert \frac{1}{1-p}\vert$ can increase arbitrarily as $p$ is getting closer to 1. 
    
    Using the relative termination criterion, the algorithm stops 
    after at most 11 iterations for any $p<1$. 
    The reason for this is that $\xbf^{(10)}(s_2) = 0.1 \cdot \sum^{8}_{k=0} p^k$ and $\xbf^{(11)}(s_2) = 0.1 \cdot \sum^{9}_{k=0} p^k$.
    Thus,
    \begin{align*}
         \diff{rel}(\xbf^{10}, \xbf^{11}) 
         & = \frac{ 0.1 \cdot \sum^{8}_{k=0} p^k  -  0.1 \cdot \sum^{9}_{k=0} p^k}{ 0.1 \cdot \sum^{9}_{k=0} p^k} \\
         &= \frac{p^9}{\sum^{9}_{i=0} p^k} \\
         & \overset{p<1}{\leq} \frac{p^9}{10 \cdot p^9}  = 0.1. 
    \end{align*}
    By a similar argument as above, we obtain that this criterion cannot guarantee $\epsilon$-soundness: 
    The true EVTs for $s_2$ can become arbitrarily high by setting $p$ closer to 1. 
    However, the returned value for $s_2$ equals $\xbf^{(11)}(s) = 0.1 \cdot \sum^{9}_{k=0} p^k $, which is smaller than $1$ for any $p<1$. 
    In this case, it can be verified that the relative difference is $p^{10}$, which implies that we can find $p <1$ for every $\epsilon\in (0,1)$ such that the relative difference is greater than $\epsilon$. 
\end{example}

\newpage
\section{Gauss-Seidel Extensions}\label{ap:GS}

\subsection{Gauss-Seidel Value Iteration}\label{ap:GSVI}

We lift an established optimization of the classic VI approach, called Gauss-Seidel VI~\cite{Put94}, to EVTs. 
The VI algorithm, as described in \Cref{algorithm:VI}, uses the vector $\xbf^{(k-1)}$ of the previous iteration to perform the update of the current approximation $\xbf^{(k)}=\vioperator(\xbf^{(k-1)})$. 
The idea of the Gauss-Seidel optimization is to compute the value of a state $s_i \in \statesTr$ by considering the already updated values for states $s_j \in \statesTr$ with $i<j$. 

We consider an arbitrary but fixed linear order relation $\preceq$   
on the set of transient states and assume that the states are indexed such that $s_i \preceq s_j $ if 
and only if $i \leq j$ holds for all $s_i, s_i \in \statesTr$. 
\begin{definition}\label{definition:GSEVTs_operator}
    The \emph{Gauss-Seidel EVTs-operator} 
    $\Gamma \colon \evtsdomain{\statesTr} \to \evtsdomain{\statesTr}$ is defined as $\Gamma(\xbf) = (\gamma_{s_i}(\xbf))_{s_i \in \statesTr}$ for all $\xbf \in \evtsdomain{\statesTr}$, where 
    \begin{align*}
        \gamma_{s_i}(\xbf)=
        \iotainit(s_i) +  \sum_{j = 1}^{i-1} \Pbf(s_j, s_i) 
        \cdot \gamma_{s_j}(\xbf) + \sum_{j = i}^{\vert \statesTr \vert } \Pbf(s_j, s_i) \cdot  
        \xbf(s_j).
    \end{align*}
\end{definition} 
The Gauss-Seidel variant of \Cref{algorithm:VI} is obtained by replacing $\vioperator$ with the Gauss-Seidel EVTs-operator $\Gamma$ in \Cref{algorithm:operator_line}. 
Then, given the input vector $\xbf^{(0)} \in \evtsdomain{\statesTr}$, the algorithm computes the vector $\Gamma^{(k)}(\xbf^{(0)})$ for increasing $k \in \nat$.  
The convergence results from standard VI also apply to the Gauss-Seidel variant: 
\begin{theorem}\label{theorem:GSVI_convergence}
    It holds that 
    \begin{enumthm}
        \item $\left( \EXP{\vt_s}\right)_{s \in \statesTr}$ is the unique 
        fixed point of $\Gamma$ in $\evtsdomain{\statesTr}$ and\label{gstheorem:VI:operator_fixpoint}
        \item for all $\xbf \in \evtsdomain{\statesTr}$, 
        we have $\lim_{k\to \infty} \Gamma^{(k)}(\xbf)= (\EXP{\vt_s})_{s \in \statesTr}$.\label{gsvi:operator_limit}
    \end{enumthm}
\end{theorem}

\begin{proof}
  The proof of \Cref{theorem:GSVI_convergence} is based on the proof provided by Puterman to establish the classical Gauss-Seidel VI approach~\cite[Theorem 6.3.4]{Put94}. 
  To be more precise, the statements \Cref{gstheorem:VI:operator_fixpoint} and \Cref{gsvi:operator_limit} can be shown by formulating Gauss-Seidel VI in terms of a regular splitting (see e.g.,~\cite[Chapter 6.3]{Var99})
  The linear system $(\Ibf-\Qbf^T) \cdot \xbf = \tauinit$ from \Cref{theorem:expected_visiting_times_LEQ} can be solved by splitting $(\Ibf-\Qbf^T)$ into $(\Ibf-\Qbf^T) = \Lbf - \Ubf$. 
  In particular, the splitting is represented by $\Lbf = (\Ibf-\Qbf^{T}_L)$ and $\Ubf= (\Qbf^{T}_U)$, where $\Qbf^{T}_L$ and $\Qbf^{T}_U$ are the strictly lower triangular and upper triangular part of $\Qbf^T$, respectively.
  Thus, for $n = \vert \statesTr \vert$, we have 
  \[
      \Qbf^T_{\Lbf} = 
       {\begin{pmatrix}
          0&0&\cdots &0\\
          \Qbf^T_{21}& 0 &\cdots & 0\\
          \vdots &\vdots &\ddots &\vdots \\
          \Qbf^T_{n1}& \Qbf^T_{n2}&\cdots & 0
      \end{pmatrix}} 
      \quad \text{and} \quad 
      \Qbf^T_{\Ubf}=
      {
      \begin{pmatrix}
          \Qbf^T_{11}& \Qbf^T_{12}&\cdots & \Qbf^T_{1n}\\
          0&\Qbf^T_{22}&\cdots & \Qbf^T_{2n}\\
          \vdots &\vdots &\ddots &\vdots \\
          0 & 0 & \cdots &\Qbf^T_{nn}
      \end{pmatrix}}. 
  \]
  From the proof of \Cref{theorem:VI_convergence}, we recall that $\lim_{k \to \infty} ({\Qbf^{k}})^k= \0$. 
  It follows that $\lim_{k \to \infty} ({\Qbf^{T}_L})^k= \0$ and thus, the spectral radius $\rho(\Qbf^{T}_L)$ is strictly smaller than 1. 
  Consequently, the inverse of $\Lbf$ is well-defined by the Neumann series $\Lbf^{-1} = \left(\Ibf - \Qbf^{T}_L\right)^{-1} = \sum_{k =0}^{\infty} \left({\Qbf^{T}_L}\right)^k$ (see~\cite{Neu77}). 
  In particular, we have that $\Qbf^{T}_L \geq \0$, which implies that the infinite sum converges to a non-negative matrix. 
  Since $\Ubf = \Qbf^T_U$ is also non-negative, we can conclude that the splitting above is \emph{regular}. 
  Since $(\Ibf-\Qbf^T)$ is invertible, we obtain that the iterative scheme 
  \begin{align*}
      \xbf^{(k)} = \begin{cases}
          \xbf  & \text{if $k=0$}, \\
          \Lbf^{-1}  \Ubf \cdot \xbf^{(k-1)} + \Lbf^{-1}   \cdot \tauinit& \text{otherwise}
      \end{cases}
  \end{align*} 
  defines a sequence $(\xbf^{(k)})_{k \in \nat}$ that converges to the unique solution of the linear system 
  $(\Ibf-\Qbf^T) \cdot \xbf = \tauinit$ for any starting vector $\xbf\in \evtsdomain{\statesTr}$. 
  In fact, the unique solution is equal to the vector $(\EXP{\vt_s})_{s \in \statesTr}$ by \Cref{theorem:expected_visiting_times_LEQ}. 

  It remains to show that the sequence $(\xbf^{k})_{k \in \nat}$ defined by the iterative scheme indeed coincides with the sequence $(\Gamma^{(k)}(\xbf))_{k \in \nat}$ generated by iteratively applying the operator $\Gamma$ to the vector $\xbf$. 
  The proof is conducted by induction over $k \in \nat$:  
  For $k=0$, the claim trivially holds. 
  Assuming $\Gamma^{(k)}(\xbf) = \xbf^{(k)}$ holds for a fixed $k>0$, we obtain $\Gamma^{(k+1)}(\xbf) = \Gamma(\xbf^{(k)})$. 
  In particular, for all $s_i \in \{1, \dots, \statesTr\}$, we obtain 
  \begin{align*}
      \gamma_{s_i}\bigl(\xbf^{(k)}\bigr) & = \iotainit({s_i}) +  
      \sum_{j = 1}^{i-1} \Pbf(s_j, s_i) \cdot \gamma_{s_j}\left(\xbf^{(k)}(s_j)\right) 
      + \sum_{j = i}^{\vert \statesTr \vert } \Pbf(s_j, s_i) \cdot \xbf^{(k)}(s_j) \\
      & = \tauinit({s_i}) 
      + \sum_{j = 1}^{i-1} \underbrace{\Qbf^T(s_i, s_j)}_{=-\Lbf_{ij}} \cdot \gamma_{s_j}\left(\xbf^{(k)}(s_j)\right) 
      + \sum_{j = i}^{\vert \statesTr \vert } \underbrace{\Qbf^T(s_i, s_j)}_{=\Ubf_{ij}} \cdot \xbf^{(k)}(s_j) \\
      & = \tauinit({s_i})  
      - \sum_{j = 1}^{i-1} \Lbf_{ij} \cdot \gamma_{s_j}\left(\xbf^{(k)}(s_j)\right) 
      + \sum_{j = i}^{\vert \statesTr \vert } \Ubf_{ij} \cdot \xbf^{(k)}(s_j). 
  \end{align*}
  This can equivalently be written as 
  \begin{align*}
      \Bigr(\gamma_{s_i}\bigl(\xbf^{(k)}\bigr)\Bigl)_{s_i \in \statesTr} = 
      \Lbf^{-1} \Ubf \cdot \Gamma^{(k)}(\xbf) + \Lbf^{-1} \cdot \tauinit,  
  \end{align*}
  which concludes the proof of \Cref{theorem:GSVI_convergence}. 
\end{proof}
One advantage of this modification is that the updates are performed in-place. As a result, the implementation can be realized using only one storage vector. In contrast, the standard VI algorithm requires two vectors: one for the values of the previous iteration and another one for the current iteration. 
Moreover, compared to standard VI, the Gauss-Seidel version can lead to faster convergence because more recent values are used for the update of the approximations. 
However, the Gauss-Seidel version of the VI algorithm also uses the absolute or relative stopping criterion. 
Thus, applying Gauss-Seidel VI for obtaining the EVTs of the DTMC considered in \Cref{example:VI_unsound} can also yield inaccurate results without any guarantees on the accuracy. 

\subsection{Gauss-Seidel Interval Iteration}\label{ap:GSII}
Analogous to the Gauss-Seidel version of VI, we can apply the Gauss-Seidel optimization to the II algorithm, where the update of the upper and lower bounds of EVTs of a state $s_i \in \statesTr$ 
are performed based on the already updated values for the states $s_1, \ldots, s_{i-1}$. 
Again, we assume that the states are indexed such that $s_i \preceq s_j $ holds for all $s_i, s_i \in \statesTr$, if and only if $i \leq j$ for an arbitrary but fixed linear order relation $\preceq$ on the set of transient states. 
\begin{definition}
    The \emph{Gauss-Seidel Max EVTs-operator} $\gsmaxvioperator \colon \evtsdomain{\statesTr} \to \evtsdomain{\statesTr}$ is defined as $\gsmaxvioperator(\xbf) = \max\left\{\xbf, \bigl(\smallgsmaxvioperator_{s_i}(\xbf)\bigr)_{s_i \in \statesTr}\right\} $ for all $\xbf \in \evtsdomain{\statesTr}$, where 
    \begin{align*}
        \smallgsmaxvioperator_{s_i}(\xbf) = \iotainit(s_i) +  \sum_{j = 1}^{i-1} \Pbf(s_j, s_i) \cdot \bigl(\gsmaxvioperator(\xbf)\bigr)(s_j) + \sum_{j = i}^{n} \Pbf(s_j, s_i) \cdot  \xbf(s_j), 
    \end{align*}%
    and the \emph{Gauss-Seidel Min EVTs-operator} $\gsminvioperator\colon \evtsdomain{\statesTr} \to \evtsdomain{\statesTr}$ is given by $\gsminvioperator(\xbf) = \min \left\{\xbf, \bigl(\smallgsminvioperator_{s_i}(\xbf)\bigr)_{s_i \in \statesTr}\right\},$
    where
    \begin{align*}
        \smallgsminvioperator_{s_i}(\xbf) = \iotainit(s_i) +  \sum_{j = 1}^{i-1} \Pbf(s_j, s_i) \cdot \bigl(\gsminvioperator(\xbf)\bigr)(s_j) + \sum_{j = i}^{n} \Pbf(s_j, s_i) \cdot \xbf(s_j).
    \end{align*}%
\end{definition} 
The Gauss-Seidel II algorithm corresponds to \Cref{alg:II}, where the operators $\gsmaxvioperator$ and $\gsminvioperator$ are used instead of $\maxvioperator$ and $\minvioperator$. 
The correctness results stated in \Cref{theorem:II_correctness} can be directly lifted to the Gauss-Seidel II variant of \Cref{alg:II} by using the following analog to \Cref{theorem:II_convergence}. 
\begin{theorem}\label{theorem:GSII_convergence}
    Let $\ubf,\lbf \in \evtsdomain{\statesTr}$ such that 
    $\lbf \leq (\EXP{\vt_s})_{s \in \statesTr} \leq \ubf$. 
    For the operators $\gsmaxvioperator$ and $\gsminvioperator$ it holds that 
   \begin{enumthm}
       \item  $\gsmaxvioperator^{(k)}(\lbf) \leq \gsmaxvioperator^{(k+1)}(\lbf) 
       \quad \text{and} \quad  \gsminvioperator^{(k)}(\ubf) \geq \gsminvioperator^{(k+1)}(\ubf)$ for all $k \in \nat,$ 
       \item $\vioperator^{(k)}(\lbf) \leq  \gsmaxvioperator^{(k)}(\lbf) \leq 
       \left(\EXP{\vt_s}\right)_{s\in \statesTr} 
       \leq \gsminvioperator^{(k)}(\ubf) \leq \vioperator^{(k)}(\lbf)$ for all $k \in \nat$, and 
       \item $ \lim_{k\to \infty} \gsmaxvioperator^{(k)}(\lbf)= 
       \left(\EXP{\vt_s}\right)_{s\in \statesTr} \quad \text{and} \quad  \lim_{k \to \infty} \ubf^{(k)}= \left(\EXP{\vt_s}\right)_{s\in \statesTr}$.
   \end{enumthm}
\end{theorem}
\begin{proof}
    Using the monotonicity of $\Gamma$ the proof works analogously to the proof of \Cref{theorem:II_convergence}, where the occurrences of $\maxvioperator$ and $\minvioperator$ are replaced by $\gsmaxvioperator$ and $\gsminvioperator$, respectively. 
\end{proof}

\newpage
\section{Absolute $\epsilon$-Sound Topological Computation}\label{ap:topo_abs}

The total error can exceed $\epsilon$ if an imprecise method is used in \Cref{alg:topo_exact:line:compute} of \Cref{alg:topo_exact} even if it guarantees $\epsilon$-sound approximations. 
In the following, we calculate an upper bound on the accumulated error with respect to the absolute difference. 

To study the extent of the error, we consider three factors that have an impact on the error accumulation; 
the incoming transitions from different SCCs, the recurrence probability, and the length of the SCC chains. 
Below, we discuss the intuition behind each cause. For simplicity, we assume that the parametric SCC restriction $\dtmcRestr{C}[\xbf]$ is a well-defined DTMC for each SCC $C$. 
\begin{itemize}
  \item \emph{Incoming transitions}: 
  As elaborated above, the SCC restriction $\dtmcRestr{C}[\xbf]$ for the non-bottom SCC $C$ is based on approximations $\xbf(t)$ of the EVTs of states $t$ in $\dtmcRestr{C'}[\xbf]$ for SCCs $C' \hookrightarrow C$. 
  We assume the vector $\xbf$ satisfies $\xbf(t) = \EXPstateMC{\vt_t}{}{\dtmc}+ \epsilon_{t}$ for each $t \in C'$, $C' \hookrightarrow C$, and $\abs{\epsilon_{t}} \leq \epsilon$. 
  For a state $s \in C$, the initial distribution of $\dtmcRestr{C}[\xbf]$ (see \Cref{definition:param_SCC_restriction}) is specified by 
  \begin{align*}
    \iotainit^{\dtmcRestr{C}[\xbf]}(s) = 
    \underbrace{\iotainit^{\dtmc}(s)+\sum_{t \in \statesTr \setminus C} \Pbf^{\dtmc}(t,s) \cdot \EXPstateMC{\vt_t}{}{\dtmc}}_{= \iotainit^{\dtmcRestr{C}}(s)} +
     \sum_{t \in \statesTr \setminus C} \Pbf^{\dtmc}(t,s) \cdot \epsilon_{t}.
  \end{align*}
  Hence, the error in the initial distribution is given by $\sum_{t \in \statesTr \setminus C} \Pbf^{\dtmc}(t,s) \epsilon_{t}$, which can generally exceed $\epsilon$. 
  \item \emph{Recurrence probability}: To study the effect of recurrence probability, by~\cite[Lemma 3.5.]{Bai+17} which states that EVTs in $\dtmcRestr{C}[\xbf]$ equal 
  \begin{align*}
    \EXPstateMC{\vt_s}{}{\dtmcRestr{C}[\xbf]} & = 
    \PRstateMC{}{\dtmcRestr{C}[\xbf]}(\event \{s\}) 
    \cdot \frac{1}{1-\PRstateMC{s}{\dtmcRestr{C}[\xbf]}(\nextLTL \event \{s\})}.  
  \end{align*}
  Assuming the initial distribution of $\dtmcRestr{C}[\xbf]$ is erroneous yields that the reachability probability $\PRstateMC{}{\dtmcRestr{C}[\xbf]}(\event \{s\})$ in the parametric SCC restriction $\dtmcRestr{C}[\xbf]$ differs from the corresponding probability $\PRstateMC{}{\dtmcRestr{C}}(\event \{s\})$ in the exact SCC restriction $\dtmcRestr{C}$. 
  It follows that the 
  error in the reachability probability 
  $\PRstateMC{}{\dtmcRestr{C}[\xbf]}(\event \{s\})$ 
  is weighted by $\frac{1}{1- \PRstateMC{s}{\dtmcRestr{C}}(\nextLTL \event \{s\})}$. 
  Thus, the difference between the EVTs in $\dtmcRestr{C}[\xbf]$ and 
  the EVTs of the DTMC $\dtmc$ 
  increases as the recurrence probability tends to 1. 
  \item \emph{Chain length}: The analysis of an 
  SCC $C$ in \Cref{alg:topo_exact} is based on a potentially erroneous 
  model $\dtmcRestr{C}[\xbf]$. 
  Consequently, the EVTs in $\dtmcRestr{C}[\xbf]$ might differ 
  from the EVTs in the original model $\dtmc$ as described above. 
  Moreover, the vector 
  $(\widehat{\xbf}(s))_{s \in C}$ that is computed 
  in \Cref{alg:topo_exact:line:compute} of \Cref{alg:topo_exact} 
  is $\epsilon$-sound with respect to the EVTs in $\dtmcRestr{C}[\xbf]$ and 
  the approximated EVT $\xbf(s)$ of the state $s\in C$ 
  can differ by an additional error $\epsilon$ 
  from the \textemdash{} already erroneous  \textemdash{} EVT of $s$ in $\dtmcRestr{C}[\xbf]$. 
  Thus, the error is propagated further since 
  the values $\xbf(s)$ are considered in the 
  analysis of the successor SCCs of $C$. 
  In particular, the error accumulates 
  up to $L$ times if the longest chain in $\dtmc$ 
  that ends in a non-bottom SCC has length $L$. 
\end{itemize}

\begin{lemma}\label{theorem:topo_bound_abs}
    Let $\epsilon>0$ and 
    let $\xbf \in \evtsdomain{\statesTr}$ be a vector such that 
    for every non-bottom SCC $C$, the vector 
    $(\xbf(s))_{s\in C}$ is an $\epsilon$-sound approximation of 
    the EVTs  
    $(\EXPstateMC{\vt_s}{}{\dtmcRestr{C}[\xbf]})_{s \in C}$ with respect to the absolute difference, 
    i.e., 
    \begin{align*}
      \diff{abs}\left((\xbf(s))_{s\in C}, (\EXPstateMC{\vt_s}{}{\dtmcRestr{C}[\xbf]})_{s \in C} \right) \leq \epsilon. 
    \end{align*}
    Then, the EVTs in $\dtmc$ satisfy
    \begin{align*}
      \diff{abs}\left(\xbf, (\EXPstateMC{\vt_s}{}{\dtmc})_{s \in \statesTr} \right) \leq \epsilon \cdot \Bigl(1+ \frac{1}{q}  \cdot \sum_{i=1}^{L} T^i \Bigr), 
    \end{align*}
    where \begin{itemize}
      \item $T$ is the maximal number of incoming transitions into an SCC, i.e., 
      \begin{align*}
      T = \max_{C \in \sccTr{\dtmc}}
      \left\{ \abs*{ \{(s,t)  \mid s \in  S \setminus C , t \in C \text{ and }\Pbf^{\dtmc}(s,t) >0 \}} \right\}, 
    \end{align*}
      \item For each $s \in \statesTr$, $q$ under-approximates the probability of starting in $s$ and never revisit $s$ from the next step on: $0< q \leq \min_{s \in \statesTr} \left\{ 1-\PRstateMC{s}{\dtmc}(\nextLTL \event \{s\}) \right\}$,  
      \item $L$ is the length of the longest chain ending in a non-bottom SCC, i.e., 
      \begin{align*}
          L = \max_{\kappa \in \chainsTr{\dtmc}} \left\{\abs{\kappa} \right\}. 
      \end{align*}
    \end{itemize}
\end{lemma}
\begin{proof}
  Consider $\epsilon >0$ and $\xbf \in \evtsdomain{\statesTr}$ as described above. 
  Let $C$ be a non-bottom SCC and let $\dtmcRestr{C} = \dtmcRestrTuple{\dtmcRestr{C}}{}$ be the corresponding SCC restriction. 
  We show $\xbf(s) \leq \EXPstateMC{\vt_s}{}{\dtmc} +  \epsilon \cdot  \left( 1+ \frac{1}{q} \cdot \sum_{i=1}^{L} T^i \right)$ holds for all $s \in \statesTr$. 
  Proving that $\xbf(s) \geq \EXPstateMC{\vt_s}{}{\dtmc} -  \epsilon \cdot  \left( 1+ \frac{1}{q} \cdot \sum_{i=1}^{L} T^i \right)$ works analogously. 

  We start by showing that for all $s \in C$, it holds that $\xbf(s) \leq \EXPstateMC{\vt_s}{}{\dtmc} + err(s)$, where 
  \begin{align}\label{errdef}
    err(s) = 
    \epsilon \cdot \left(1+ 
    \frac{1}{1-\PRstateMC{s}{\dtmcRestr{C}}(\nextLTL \event \{s\})} \cdot \sum_{i=1}^{k} T^i \right)
  \end{align} and 
  $k$ is 
  the length of the longest SCC chain 
  $C_0 \hookrightarrow \dots \hookrightarrow C_k \in \chainsTr{\dtmc}$ such that $C_k =C$. 
  The proof is by induction over the chain length $k \in \nat$. 

  If $k=0$, we know that $C$ is not reachable via other SCCs. 
  Thus, for the SCC restriction $\dtmcRestr{C} = \dtmcRestrTuple{C}{}$ and the parametric SCC restriction $\dtmcRestr{C}[\xbf] =  \dtmcRestrTuple{\dtmcRestr{C}}{[\xbf]}$ we have for $s \in C$
   \begin{align*}
    \iotainit^{\dtmcRestr{C}}(s) &= 
   \iotainit^{\dtmc}(s)+ \sum_{\substack{t \in C' \not = C \\ C'\hookrightarrow C} } 
   \Pbf(t,s) \cdot \EXPstateMC{\vt_s}{}{\dtmc} = \iotainit^{\dtmc}(s) =  \iotainit^{\dtmcRestr{C}[\xbf]}(s). 
   \end{align*}
  We conclude that $\dtmcRestr{C}[\xbf]$ is equal to the DTMC $\dtmcRestr{C}$. 
  Hence, for every $s \in C$, we have 
  \begin{align*}
    \xbf(s) &\leq \EXPstateMC{\vt_s}{}{\dtmcRestr{C}[\xbf]}+ \epsilon && \text{by assumption on $\xbf$}\\
    & =  \EXPstateMC{\vt_s}{}{\dtmcRestr{C}}+ \epsilon \\
      & = \EXPstateMC{\vt_s}{}{\dtmc} +  err(s) && \text{by \Cref{theorem:EVTs_in_SCC_restriction}}. 
  \end{align*} 
  Let $C$ be a non-bottom SCC, 
  where $k > 0$ is the length of the longest SCC chain ending sin $C$. 
  We assume that the claim holds for all SCCs that are reachable by a chain that has at most length $k-1$. 
  
  Furthermore, let $\dtmcRestr{C}[\xbf]$ be the parametric restriction to $C$. 
  By \Cref{remark:sccrestr}, we have  
  \begin{align*}
    \bigl(\EXPstateMC{\vt_s}{}{\dtmcRestr{C}[\xbf]}\bigr)_{s \in C} 
    &= \Bigl(\Ibf - (\Qbf^{\dtmcRestr{C}})^T\Bigr)^{-1} \cdot 
    \tauinit^{\dtmcRestr{C}[\xbf]}. 
  \end{align*} 
  We know that the inverse of $\Ibf - (\Qbf^{\dtmcRestr{C}})^T$ is given by the following converging Neumann series (see~\cite{Neu77}): 
  $\bigl(\Ibf - (\Qbf^{\dtmcRestr{C}})^T\bigr)^{-1} = \sum_{i=0}^{\infty} \bigl(\Qbf^{\dtmcRestr{C}}\bigr)^{i}$. 
  Thus, for a single state $s \in C$, we have 
  \begin{align}
    \label{EQ:topo_abs}
    \EXPstateMC{\vt_s}{}{\dtmcRestr{C}[\xbf]} 
    &= 
    \sum_{t \in C} \sum_{i=0}^{\infty} \Pbf^i(t,s) \cdot \iotainit^{\dtmcRestr{C}[\xbf]}(t) \nonumber \\
    & = \sum_{t \in C} \sum_{i=0}^{\infty} \Pbf^i(t,s) \cdot 
    \Bigl( \iotainit^{\dtmc}(t) + \sum_{\substack{t' \in C' \\ C'\hookrightarrow C}} \Pbf(t', t) \cdot \xbf(t') \Bigr). 
  \end{align}
  If $C' \hookrightarrow C$ holds, we conclude that the longest 
  chain of non-bottom SCCs $\kappa \in \chainsTr{\dtmc}$ that ends in $C'$ satisfies 
  $\abs{\kappa} < k$. 
  Hence, by the induction hypothesis, it holds that  
  $\xbf(t') \leq \EXPstateMC{\vt_{t'}}{}{\dtmc} + err(t')$ for every $t' \in C'$. 
  Consequently, we obtain that 
  \begin{align*}
    \EXPstateMC{\vt_s}{}{\dtmcRestr{C}[\xbf]} 
    &\leq \sum_{t \in C} \sum_{i=0}^{\infty} \Pbf^i(t,s) 
    \cdot \Biggl(
      \iotainit^{\dtmc}(t) + 
      \sum_{\substack{t' \in C' \\ C'\hookrightarrow C}} \Pbf(t', t) 
      \cdot \left( \EXPstateMC{\vt_{t'}}{}{\dtmc} + err(t')  \right)
    \Biggr) \nonumber \\
    & = \sum_{t \in C} \sum_{i=0}^{\infty} \Pbf^i(t,s) 
    \cdot 
    \Biggl(
      \iotainit^{\dtmc}(t) 
      + \sum_{\substack{t' \in C' \\ C'\hookrightarrow C}} \Pbf(t', t) \cdot \EXPstateMC{\vt_{t'}}{}{\dtmc} \nonumber\\
    & \phantom{={}} + \sum_{\substack{t' \in C' \\ C'\hookrightarrow C}} \Pbf(t', t) \cdot err(t') \Biggr) \nonumber\\
    & = \sum_{t \in C} \sum_{i=0}^{\infty} \Pbf^i(t,s) 
    \cdot 
    \Biggl(
      \PRstateMC{}{\dtmc}\bigl((S \setminus C )\until \{t\} \bigr) 
      + \hspace{-0.3cm} \sum_{\substack{t' \in C' \\ C'\hookrightarrow C}} \Pbf(t', t) \cdot err(t') 
      \Biggr) \nonumber \tag*{ by \Cref{theorem:EVTs_constrPr}}\\
    & = 
    \sum_{t \in C} \sum_{i=0}^{\infty} \Pbf^i(t,s) \cdot 
    \Biggl(
     \iotainit^{\dtmcRestr{C}}(t) 
    + \sum_{\substack{t' \in C'\\ C'\hookrightarrow C}} \Pbf(t', t) \cdot  err(t') 
    \Biggr) \tag*{by \Cref{definition:param_SCC_restriction}.}  \nonumber
  \end{align*}
  From \Cref{theorem:expected_visiting_times_LEQ} it 
  follows that 
  \[ (\EXPstateMC{\vt_s}{}{\dtmcRestr{C}})_{s \in C} = 
  (\Ibf - (\Qbf^{\dtmcRestr{C}})^T)^{-1} \cdot \iotainit^{\dtmcRestr{C}}. \] 
  Moreover, the inverse of $(\Ibf - (\Qbf^{\dtmcRestr{C}})^T)$ is given by the following converging Neumann series (see~\cite{Neu77}): 
  \[(\Ibf - (\Qbf^{\dtmcRestr{C}})^T)^{-1} = 
  \sum_{i=0}^{\infty} (\Qbf^{\dtmcRestr{C}})^i\] 
  It follows that 
  $\EXPstateMC{\vt_s}{}{\dtmcRestr{C}} =  \sum_{t \in C} \sum_{i=0}^{\infty} \Pbf^i(t,s) \cdot \iotainit^{\dtmcRestr{C}}(t)$. 
  By continuing the reasoning above, we obtain that  
  \begin{align*} 
    \EXPstateMC{\vt_s}{}{\dtmcRestr{C}[\xbf]}  & \leq
    \left( \sum_{t \in C} \sum_{i=0}^{\infty} \Pbf^i(t,s) 
    \cdot \iotainit^{\dtmcRestr{C}}(t)\right) \\
    & \phantom{={}} + \left(  \sum_{t \in C} \sum_{i=0}^{\infty}  \Pbf^i(t,s) \sum_{t' \in C'} \Pbf(t', t) \cdot  err(t') \right) \\
    &\leq \EXPstateMC{\vt_s}{}{\dtmcRestr{C}} + \sum_{t \in C} 
      \sum_{i=0}^{\infty} \Pbf^i(t,s) 
      \cdot \sum_{\substack{t' \in C' \\ C'\hookrightarrow C}}  \Pbf(t', t) \cdot  err(t'). 
  \end{align*}
  \cite[Theorem 3.2.4]{KS76} implies that the EVT of $s$ in the DTMC $\dtmcRestr{C}$ starting from a state $t$ is equal to $\EXPstateMC{\vt_s}{t}{\dtmcRestr{C}}= \sum_{i=0}^{\infty} \Pbf^i(t,s)$. 
  Thus, 
  \begin{align}\label{equationtopoabs2} 
    \EXPstateMC{\vt_s}{}{\dtmcRestr{C}[\xbf]} & \leq 
    \EXPstateMC{\vt_s}{}{\dtmcRestr{C}} +  \sum_{t \in C}
    \EXPstateMC{\vt_s}{t}{\dtmcRestr{C}}
     \cdot \sum_{\substack{t' \in C' \\ C'\hookrightarrow C}}  \Pbf(t', t) \cdot  err(t')  \nonumber\\
    & = 
    \EXPstateMC{\vt_s}{}{\dtmcRestr{C}}  + 
    \sum_{t \in C} 
    \frac{ \PRstateMC{t}{\dtmcRestr{C}}(\event \{s\}) }{1-\PRstateMC{s}{\dtmcRestr{C}}(\nextLTL \event \{s\})} 
    \cdot \sum_{\substack{t' \in C' \\ C'\hookrightarrow C}}  \Pbf(t', t) 
    \cdot  err(t')  \nonumber 
    && \text{by \cite[Lemma 3.5.]{Bai+17}} \\
    & \leq 
     \EXPstateMC{\vt_s}{}{\dtmcRestr{C}} + 
     \frac{1}{1-\PRstateMC{s}{\dtmcRestr{C}}(\nextLTL \event \{s\})} 
     \cdot \sum_{t \in C} \sum_{\substack{t' \in C' \\ C'\hookrightarrow C}} \Pbf(t',t) \cdot 
     err(t'). 
\end{align}
For each SCC chain 
$\kappa \in \chainsTr{\dtmc}$ 
that ends in an SCC $C'$ with $C' \hookrightarrow C$ we have 
that $\abs{\kappa} \leq k-1$. 
Consequently, for each $t' \in C'$, we obtain that 
  \begin{align}\label{ineq1}
    err(t') \leq 
  \epsilon \cdot \left(1+ 
  \frac{1}{1-\PRstateMC{t'}{\dtmcRestr{C}}(\nextLTL \event \{t'\})} \cdot \sum_{i=1}^{k-1} T^i \right). 
  \end{align}
  Moreover, 
  for all $t\in C$, it holds that 
  $\Pbf(t',t)$ is at most the probability of leaving the SCC $C'$. 
  In particular, 
  $\Pbf(t',t) \leq 1-\PRstateMC{t'}{\dtmcRestr{C}}(\nextLTL \event \{t'\})$, which 
  implies that 
  \begin{align}\label{ineq2}
    \frac{\Pbf(t',t)}{1-\PRstateMC{t'}{\dtmcRestr{C}}(\nextLTL \event \{t'\})} \leq \iverson{\Pbf(t',t) >0}.
  \end{align}
  By combining \Cref{ineq1,ineq2}, we obtain that 
  \begin{align*} 
    \Pbf(t',t) \cdot err(t')  
    & \leq \Pbf(t',t) 
    \cdot \epsilon \cdot \left(1+ 
    \frac{1}{1-\PRstateMC{t'}{\dtmcRestr{C}}(\nextLTL \event \{t'\})} \cdot \sum_{i=1}^{k-1} T^i \right) 
    \\ 
    & = \epsilon \cdot \left( \Pbf(t',t)  + 
    \frac{\Pbf(t',t) }{1-\PRstateMC{t'}{\dtmcRestr{C}}(\nextLTL \event \{t'\})} \cdot \sum_{i=1}^{k-1} T^i \right) \\
    & \leq \iverson{\Pbf(t',t) >0} \cdot \epsilon \cdot \Bigl(1 + \sum_{i=1}^{k-1} T^i \Bigr) 
 \end{align*}
  Using this inequality in \Cref{equationtopoabs2} yields 
  \begin{align*}
    \EXPstateMC{\vt_s}{}{\dtmcRestr{C}[\xbf]} 
    & \leq 
    \EXPstateMC{\vt_s}{}{\dtmcRestr{C}} + 
      \frac{1}{1-\PRstateMC{s}{\dtmcRestr{C}}(\nextLTL \event \{s\})}  \\
    & \phantom{={}}
      \cdot 
      \sum_{t \in C} \sum_{\substack{t' \in C' \\ C'\hookrightarrow C}} 
      \iverson{\Pbf(t',t) >0}
      \cdot \epsilon \cdot \Bigl(1 + \sum_{i=1}^{k-1} T^i \Bigr) 
  \end{align*}
  The sum  $\sum_{t \in C} \sum_{\substack{t' \in C' \\ C'\hookrightarrow C}} 
  \iverson{\Pbf(t',t) >0}$ is equal to the number of incoming transitions to $C$. 
  Since $T$ is an upper bound for this number, we conclude that 
  \begin{align*}
    \EXPstateMC{\vt_s}{}{\dtmcRestr{C}[\xbf]} 
    & \leq 
      \EXPstateMC{\vt_s}{}{\dtmcRestr{C}} + 
        \frac{1}{1-\PRstateMC{s}{\dtmcRestr{C}}(\nextLTL \event \{s\})} 
        \cdot T \cdot 
        \epsilon \cdot \Bigl(1 + \sum_{i=1}^{k-1} T^i \Bigr) \\
    &= \EXPstateMC{\vt_s}{}{\dtmcRestr{C}} + 
    \frac{1}{1-\PRstateMC{s}{\dtmcRestr{C}}(\nextLTL \event \{s\})} 
    \cdot \epsilon \cdot \sum_{i=1}^{k} T^i. 
\end{align*}
  By assumption, the vector $(\xbf(s))_{s \in C}$ is $\epsilon$-sound with respect to the absolute difference from the vector $(\EXPstateMC{\vt_s}{}{\dtmcRestr{C}[\xbf]})_{s \in C}$.  
  Consequently, we have that 
  \begin{align*}
    \xbf(s) 
    &\leq \EXPstateMC{\vt_s}{}{\dtmcRestr{C}[\xbf]} + \epsilon  \\
    & \leq \EXPstateMC{\vt_s}{}{\dtmcRestr{C}}+ \frac{1}{1-\PRstateMC{s}{\dtmcRestr{C}}(\nextLTL \event \{s\})} \cdot \epsilon \cdot \sum_{i=1}^{k} T^i +\epsilon\\
    & \leq \EXPstateMC{\vt_s}{}{\dtmcRestr{C}}+ \epsilon   \cdot  \left( 1+ \frac{1}{1-\PRstateMC{s}{\dtmcRestr{C}}(\nextLTL \event \{s\})}  \cdot \sum_{i=1}^{k} T^i \right)  \\
    & \leq \EXPstateMC{\vt_s}{}{\dtmc} + \underbrace{ \epsilon   \cdot \left( 1+ \frac{1}{1-\PRstateMC{s}{\dtmcRestr{C}}(\nextLTL \event \{s\})}  \cdot \sum_{i=1}^{k} T^i \right)}_{=err(s)} && \text{by \Cref{theorem:EVTs_in_SCC_restriction}.}
  \end{align*}
  Thus, we have shown that the claim holds for all $k \in \nat$. 

Next, let $s \in \statesTr$ and $C$ be the SCC $C$ of $s$.
Additionally, assume that $k$ is the length of the longest SCC chain 
$\kappa \in \chainsTr{\dtmc}$ ending in $C$. 
The length of the longest SCC chain containing only transient SCCs in $\dtmc$ is equal to $L$ 
and by the choice of $q$, we have that 
$\frac{ 1 }{1-\PRstateMC{s}{\dtmcRestr{C}}(\nextLTL \event \{s\})}\leq \frac{1}{q}$. 
Thus, 
\begin{align*} 
\xbf(s) &\leq \EXPstateMC{\vt_s}{}{\dtmc} + \epsilon \cdot \left( 1+ \frac{1}{1-\PRstateMC{s}{\dtmcRestr{C}}(\nextLTL \event \{s\})} \cdot \sum_{i=1}^{k} T^i \right)\\
&\leq \EXPstateMC{\vt_s}{}{\dtmc} + \epsilon \cdot  \left( 1+  \frac{1}{q} \cdot \sum_{i=1}^{L} T^i \right), 
\end{align*} 
which concludes the proof of \Cref{theorem:topo_bound_abs}. 
\end{proof}

Using the upper bounds on the absolute differences provided in \Cref{theorem:topo_bound_abs}, we derive an algorithm for the topological computation of $\epsilon$-sound results. 
The value $q$ satisfying $0 < q \leq \min_{s \in \statesTr} \left\{ (1-\PRstateMC{s}{\dtmc}(\nextLTL \event \{s\})) \right\}$, can be computed automatically, for example, by exploiting the graph-based methods proposed by Baier et al.~\cite[Section 3.3]{Bai+17} that bound the recurrence probabilities $\PRstateMC{s}{\dtmc}(\nextLTL \event \{s\})$ from above. 

\begin{algorithm}[t]
  \Indm{} 
      \Input{DTMC $\dtmc$, $\epsilon >0$ 
      for $crit = abs$ and $\epsilon \in (0,1)$ for $crit=rel$}
      \Output{$\xbf  \in \evtsdomain{\statesTr}$ with $\diff{crit}(\xbf, (\EXP{\vt_s})_{s \in \statesTr}) \leq \epsilon$}    
  \Indp{}
    $\{C_1 \dots C_n\} \leftarrow \sccTr{\dtmc}$ \tcp*{obtain the non-bottom SCCs, $C_i \hookrightarrow C_j$ implies $i<j$}
    
    $\xbf \leftarrow \0$
    
    compute $q$ with $0< q \leq \min_{s \in \statesTr} \left\{ (1-\PRstateMC{s}{\dtmc}(\nextLTL \event \{s\})) \right\}$ \tcp{e.g.,~\cite[Section 3.3]{Bai+17}}
    
    $T \leftarrow \max_{C \in \sccTr{\dtmc}} \left\{ \abs{ \{(s,t)  \mid s \in  S \setminus C , t \in C \text{ and }\Pbf(s,t) >0 \}} \right\}$ 
    
    
    $L \leftarrow \max_{\kappa \in \chainsTr{\dtmc}} \left\{\abs{\kappa} \right\}$

    $\sigma \leftarrow \frac{\epsilon}{ 1+ \frac{1}{q} \cdot \sum_{i=1}^{L} T^i }$ \tcp*{compute the absolute threshold}\label{alg:topo_abs_line:sigma}

      \For{$k\leftarrow 1$ \KwTo{} $n$}
      { 

        compute $\widehat{\xbf} \in \evtsdomain{C}$ with 
        $\diff{crit}(\widehat{\xbf}, (\EXPstateMC{\vt_s}{}{\dtmcRestr{C_k}})_{s \in C_k}) \leq \sigma$
       \tcp{using e.g., \hyperref[alg:II]{II}}

         \lFor*{$s \in C_k$}{
           ${\xbf(s)} \leftarrow \widehat{\xbf}(s)$
         }
        
      }
      \Return{$\xbf$} 
    \caption{Sound topological computation of EVTs  $\diff{abs}$.}\label{algorithm:topo_abs}
\end{algorithm}

The absolute $\epsilon$-sound topological computation of EVTs is described in \Cref{algorithm:topo_abs}. 
The procedure works analogously to the relative $\epsilon$-sound topological computation of EVTs described in \Cref{sec:topological_algorithms}. 
The difference is that the choice of $\sigma$ is based on the bounds from \Cref{theorem:topo_bound_abs}. 
As a result, the termination threshold $\sigma$ ensures that the \emph{absolute} difference between the result $\xbf$ and the EVTs $(\EXPstateMC{\vt_s}{}{\dtmc})_{s \in \statesTr}$ does not exceed $\epsilon$. 
\begin{theorem}
      For the input DTMC $\dtmc$, and any threshold $\epsilon>0$ \Cref{algorithm:topo_abs} terminates and returns a vector $\xbf^{res} \in \evtsdomain{\statesTr}$ with $\diff{abs}(\xbf^{res}, (\EXPstateMC{\vt_s}{}{\dtmc})_{s \in \statesTr}) \leq \epsilon$. 
\end{theorem}
\begin{proof}
  After each non-bottom SCC has been considered, the algorithm terminates. 
  We proof $\epsilon$-soundness \wrt{} the absolute difference. 
  Consider $\sigma$ after the execution of \Cref{alg:topo_abs_line:sigma} in \Cref{algorithm:topo_abs}, i.e., 
  $\sigma  = \frac{\epsilon}{ 1+ \frac{1}{q} \cdot \sum_{i=1}^{L} T^i  }$ 
  where \begin{itemize}
    \item $q$ satisfies $0< q \leq \min_{s \in \statesTr} \left\{ (1-\PRstateMC{s}{\dtmc}(\nextLTL \event \{s\})) \right\}$   
    \item $L = \max_{\kappa \in \chainsTr{\dtmc}} \left\{ \abs{\kappa} \right\}$, and 
    \item $T = \max_{C \in \sccTr{\dtmc}} \left\{ \abs*{ \{(t',t)  \mid t' \in  S\setminus C , t \in C \text{ and }\Pbf(t',t) >0 \}} \right\}$. 
  \end{itemize}
  By the definition of $q$ and \Cref{theorem:topo_bound_abs}, we obtain that the returned vector $\xbf^{res} \in \statesTr$ satisfies 
  \begin{align}
    \diff{abs}\left(\xbf^{res}, (\EXPstateMC{\vt_s}{}{\dtmc})_{s \in \statesTr} \right) \leq \sigma \cdot \Bigl(1+  \cdot \sum_{i=1}^{L} T^i \cdot \frac{1}{q} \Bigr). 
  \end{align}
  Hence, we conclude 
  \begin{align*}
    \diff{abs}\left(\xbf^{res}, (\EXPstateMC{\vt_s}{}{\dtmc})_{s \in \statesTr} \right) & 
    \leq \sigma \cdot \Bigl(1+  \cdot \sum_{i=1}^{L} T^i \cdot \frac{1}{q} \Bigr) \\ 
    &= \frac{\epsilon}{ 1+ \sum_{i=1}^{L} T^i \cdot \frac{1}{q} } \cdot \Bigl(1+  \cdot \sum_{i=1}^{L} T^i \cdot \frac{1}{q} \Bigr) \\
    &= \epsilon. 
  \end{align*}
\end{proof}

\newpage
\section{Additional Experimental Results}\label{app:evaluation}

\begin{figure}[h!tb]
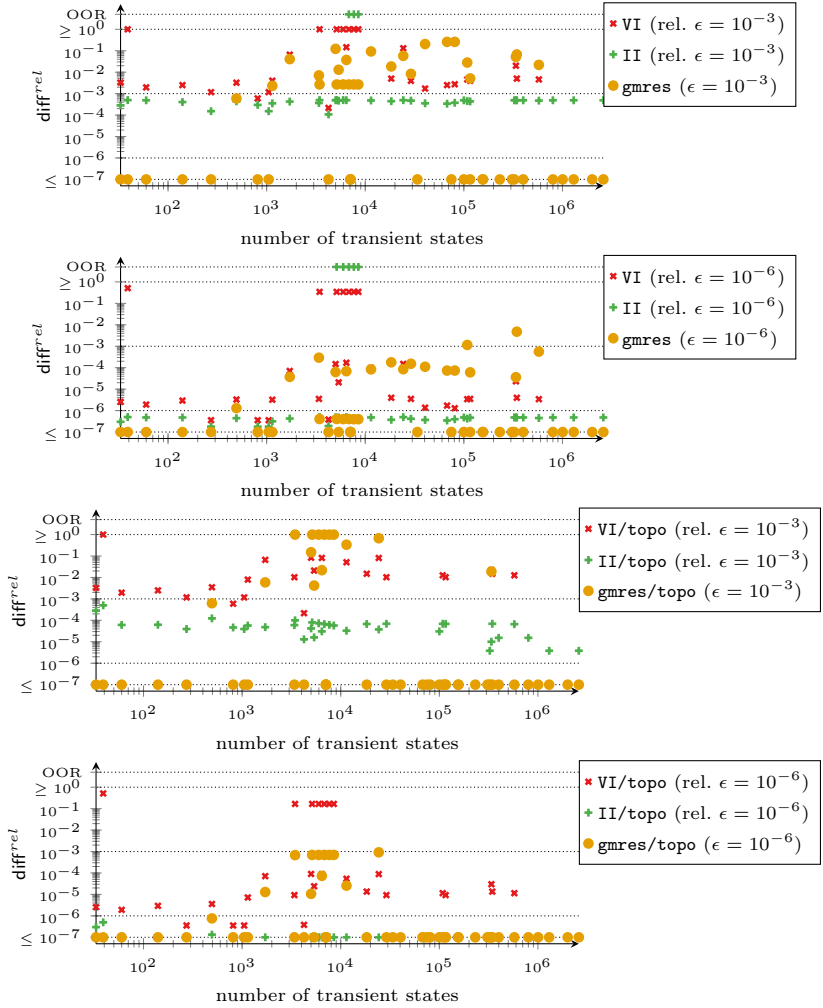

  \centering
    \renewcommand{\plotwidth}{8cm}
    \renewcommand{\plotheight}{4cm}
    \renewcommand{\legendcols}{1}
    \renewcommand{\legendstyle}{at={(1,1.01)},anchor=north west}
		{\precisionplot{evts_csv/relative-error-max-norm-value-1-texpgf-scatter.csv"}%
		{storm.sparse-rel.n-power.0.001, storm.sound-rel.n-ii.0.001, storm.sparse-rel.g-gmres.0.001}%
		{{\method{VI} (rel. $\epsilon=10^{-3}$)},
    {\method{II} (rel. $\epsilon=10^{-3}$)},
    {\method{gmres} ($\epsilon=10^{-3}$)}}%
		{number of transient states}%
		{$\diff{rel}$}{transient-states}}%
    \renewcommand{\plotwidth}{8cm}  
    \renewcommand{\plotheight}{4cm} 
    \renewcommand{\legendstyle}{at={(1,1.01)},anchor=north west}
    \renewcommand{\legendcols}{1}
		{\precisionplot{evts_csv/relative-error-max-norm-value-1-texpgf-scatter.csv"}%
		{storm.sparse-rel.n-power.1e-06, storm.sound-rel.n-ii.1e-06, storm.sparse-rel.g-gmres.1e-06}%
		{{\method{VI} (rel. $\epsilon=10^{-6}$)},
    {\method{II} (rel. $\epsilon=10^{-6}$)},
    {\method{gmres} ($\epsilon=10^{-6}$)}}%
		{number of transient states}%
		{$\diff{rel}$}{transient-states}}%
    \renewcommand{\plotwidth}{8cm}  
    \renewcommand{\plotheight}{4cm} 
    \renewcommand{\legendstyle}{at={(1,1.01)},anchor=north west}
    \renewcommand{\legendcols}{1}
		{\precisionplot{evts_csv/relative-error-max-norm-value-1-texpgf-scatter.csv"}%
		{storm.sparse-rel.topological-n-power.0.001, 
    storm.sound-rel.topological-n-ii.0.001, 
    storm.sparse-rel.topological-g-gmres.0.001}%
		{{\method{VI/topo} (rel. $\epsilon=10^{-3}$)},
    {\method{II/topo} (rel. $\epsilon=10^{-3}$)},
    {\method{gmres/topo} ($\epsilon=10^{-3}$)}}%
		{number of transient states}%
		{$\diff{rel}$}{transient-states}}%
    \renewcommand{\plotwidth}{8cm}  
    \renewcommand{\plotheight}{4cm} 
    \renewcommand{\legendstyle}{at={(1,1.01)},anchor=north west}
    \renewcommand{\legendcols}{1}
		{\precisionplot{evts_csv/relative-error-max-norm-value-1-texpgf-scatter.csv"}%
		{storm.sparse-rel.topological-n-power.1e-06, 
    storm.sound-rel.topological-n-ii.1e-06, 
    storm.sparse-rel.topological-g-gmres.1e-06}%
		{{\method{VI/topo} (rel. $\epsilon=10^{-6}$)},
    {\method{II/topo} (rel. $\epsilon=10^{-6}$)},
    {\method{gmres/topo} ($\epsilon=10^{-6}$)}}%
		{number of transient states}%
		{$\diff{rel}$}{transient-states}}%
   \caption%
 {Four scatter plots for the comparison of the relative difference ($y$-axis) between the true EVTs and results obtained by the different methods for each instance ($x$-axis, represented by the number of transient states).   
 The methods include GMRES, VI and II using the relative termination criterion with a threshold of $\epsilon = 10^{-3}$ or $\epsilon = 10^{-6}$. 
 The 4 dotted lines indicate, from bottom to top, if the result either yields a relative difference that is less or equal to the threshold $10^{-7}$, exactly equal to a precision of $10^{-6}$ or $10^{-3}$, or greater or equal to $1$. 
 The topmost line indicates that no result was obtained, for example, due to timeouts.}
\end{figure}
\begin{figure}[h!tb]
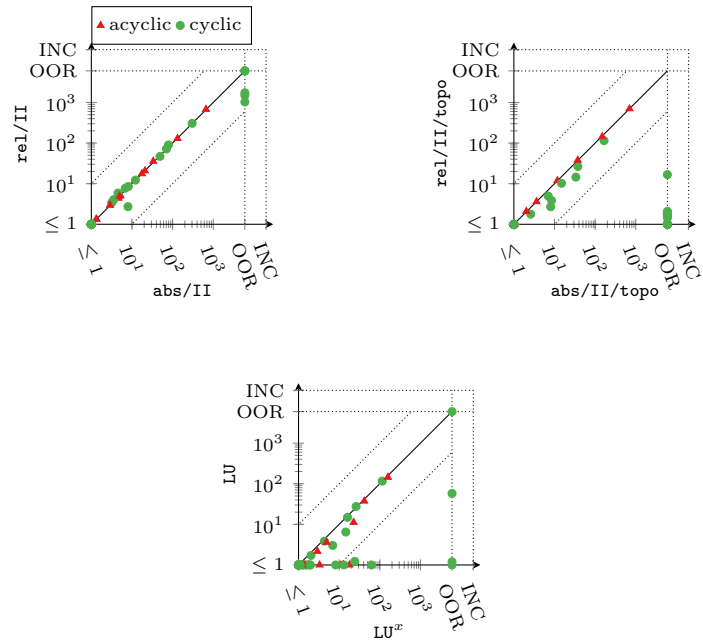

	\centering  
    \renewcommand{\scatterplotsize}{4cm}
    \renewcommand{\legendstyle}{at={(0.0,1.2)},anchor=north west}
    \renewcommand{\legendcols}{2}
    \scatterplotstormcyclic{evts_csv/mc-time-1800-texpgf-scatter.csv}
      {storm.sound-abs.n-ii.0.001}
      {{\method{abs/II}}}
      {storm.sound-rel.n-ii.0.001}
      {{\method{rel/II}}}{true}
    \scatterplotstormcyclic{evts_csv/mc-time-1800-texpgf-scatter.csv}
    {storm.sound-abs.topological-n-ii.0.001}
    {{\method{abs/II/topo}}}
    {storm.sound-rel.topological-n-ii.0.001}
    {{\method{rel/II/topo}}}
    {false}
    \scatterplotstormcyclic{evts_csv/mc-time-1800-texpgf-scatter.csv}
      {storm.exact.topological-e-sparselu.ignored}
      {{\method{LU$^x$}}}
      {storm.sparse.topological-e-sparselu.ignored}
      {{\method{LU}}}
      {false}
	\caption{Scatter plots depicting the runtime for EVTs computation (in seconds).
  The first two plots compare the (topological) \method{II} approach using the absolute stopping criterion and 
  \method{II} using the realtive stopping criterion. 
  The rightmost plot compares \method{LU} factorization using floating point arithmetic and exact arithmetic. 
  }
\end{figure}

\begin{figure}[h!tb]
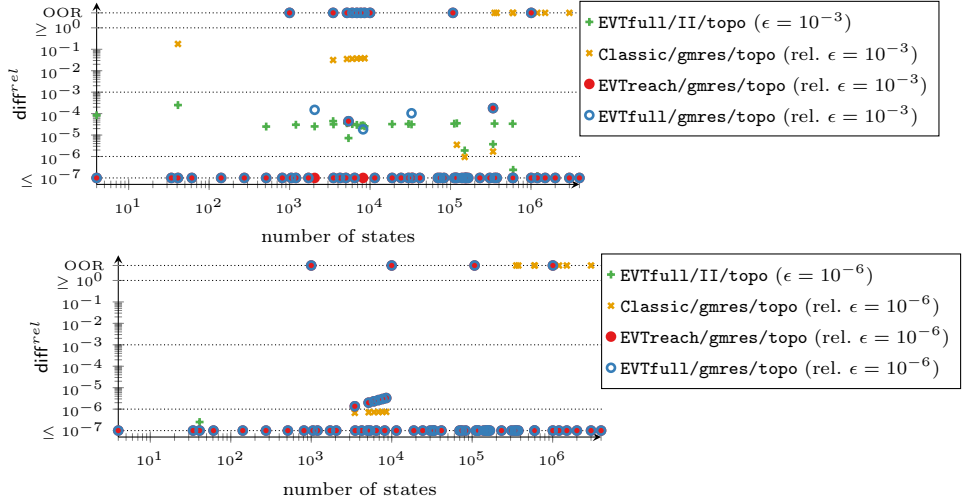

  \centering
  \vspace{0.4cm}
	  \renewcommand{\plotwidth}{8cm}
	  \renewcommand{\plotheight}{4cm}
    \renewcommand{\legendstyle}{at={(1,1.01)},anchor=north west}
    \renewcommand{\legendcols}{1}
		{\precisionplot{stationary_csv/relative-error-max-norm-value-1-texpgf-scatter.csv"}%
		{storm.sound-rel.topological-evt-n-ii.0.001,
    storm.sparse-rel.topological-classic-g-gmres.0.001, 
    storm.sparse-rel.topological-eqsys-g-gmres.0.001, 
    storm.sparse-rel.topological-evt-g-gmres.0.001}%
		{
    {\method{EVTfull/II/topo} ($\epsilon=10^{-3}$)},
    {\method{Classic/gmres/topo} (rel. $\epsilon=10^{-3}$)},
    {\method{EVTreach/gmres/topo} (rel. $\epsilon=10^{-3}$)}, 
    {\method{EVTfull/gmres/topo} (rel. $\epsilon=10^{-3}$)}}%
    {number of states}%
		{$\diff{rel}$}{states}
    }%
  \vspace{0.4cm}
    \renewcommand{\plotwidth}{8cm}
    \renewcommand{\plotheight}{4cm}
    \renewcommand{\legendstyle}{at={(1,1.01)},anchor=north west}
    \renewcommand{\legendcols}{1}
	  {\precisionplot{stationary_csv/relative-error-max-norm-value-1-texpgf-scatter.csv"}%
		{storm.sound-rel.topological-evt-n-ii.1e-06, 
    storm.sparse-rel.topological-classic-g-gmres.1e-06, 
    storm.sparse-rel.topological-eqsys-g-gmres.1e-06, 
    storm.sparse-rel.topological-evt-g-gmres.1e-06
    }%
    {
    {\method{EVTfull/II/topo} ($\epsilon=10^{-6}$)},
    {\method{Classic/gmres/topo} (rel. $\epsilon=10^{-6}$)},
    {\method{EVTreach/gmres/topo} (rel. $\epsilon=10^{-6}$)}, 
    {\method{EVTfull/gmres/topo} (rel. $\epsilon=10^{-6}$)}}%
		{number of states}%
		{$\diff{rel}$}{states}}%
	\vspace{-0.2cm}
   \caption%
 {Scatter plots for the comparison of the relative difference ($y$-axis) between the true stationary distribution and results obtained by the different methods for each instance ($x$-axis, represented by the number of transient states). 
 LU factorization, GMRES and II use the relative termination criterion with a threshold of $\epsilon = 10^{-3}$ or $\epsilon = 10^{-6}$.} 
 \label{plot:stst_accuracy}
\end{figure}

\begin{figure}[h!tb]
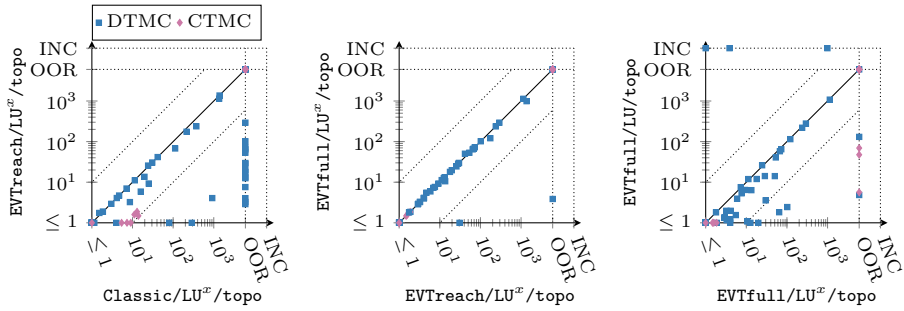

	\centering 
    \renewcommand{\scatterplotsize}{4cm}
    \renewcommand{\legendstyle}{at={(0.0,1.2)},anchor=north west}
    \renewcommand{\legendcols}{2}
    \scatterplotstorm{stationary_csv//wallclock-time-1800-INC-relative-error-max-norm-value-texpgf-scatter.csv"}
      {storm.exact.topological-classic-sparselu.ignored}
      {\method{Classic/LU$^x$/topo}}
      {storm.exact.topological-eqsys-e-sparselu.ignored}
      {\method{EVTreach/LU$^x$/topo}}
      {true}%
    \scatterplotstorm{stationary_csv/wallclock-time-1800-INC-relative-error-max-norm-value-texpgf-scatter.csv"}
    {storm.exact.topological-eqsys-e-sparselu.ignored}
    {\method{EVTreach/LU$^x$/topo}}
    {storm.exact.topological-evt-e-sparselu.ignored}
    {\method{EVTfull/LU$^x$/topo}}
    {false}%
    \scatterplotstorm{stationary_csv/wallclock-time-1800-INC-relative-error-max-norm-value-texpgf-scatter.csv"}
      {storm.exact.topological-evt-e-sparselu.ignored}
      {\method{EVTfull/LU$^x$/topo}}
      {storm.sparse.topological-evt-e-sparselu.ignored}
      {\method{EVTfull/LU/topo}}
      {false}%
	\caption{Scatter plots depicting the runtime for stationary distribution computation (in seconds).
  The first two plots compare the \method{Classic}, \method{EVTreach} and the \method{EVTfull} approach using \method{LU$^x$}. 
  \method{II} using the realtive stopping criterion. 
  The rightmost plot compares \method{LU} factorization using floating point arithmetic and exact arithmetic. 
  }
\end{figure}

\end{document}